\documentclass[jair,twoside,11pt,theapa]{article}
\usepackage{jair, theapa, rawfonts}
\jairheading{1}{20xx}{??}{??}{??}
\ShortHeadings{A new tractable Description Logic under categorical semantics}
{{Le Duc} \& Brieulle}
\usepackage{times}  
\usepackage{helvet}  
\usepackage{courier}
\usepackage{url}
 
\usepackage{algorithmic}
\usepackage{newfloat}
\usepackage{listings}
\usepackage{float}

\usepackage{soul}
\usepackage{url}
\usepackage[utf8]{inputenc}
\usepackage[small]{caption}
\usepackage{graphicx}
\usepackage{amsmath}
\usepackage{amsthm}
\usepackage{booktabs}
\usepackage{algorithm}
\usepackage{algorithmic}
\usepackage{stmaryrd}
\usepackage{stackrel}
\usepackage{mathrsfs}
\usepackage[inline]{enumitem}
\usepackage{commath}
\usepackage{graphicx}
\usepackage{todonotes}
\usepackage{mathtools}
\usepackage{framed}
\usepackage{amsmath,commath}
\usepackage{amssymb}
\usepackage{overlay}

\usepackage{graphicx}

\usepackage{xcolor}

\usepackage{pgf}
\usepackage{pgfcore}
\usepackage{pgfbaseimage}
\usepackage{pgfbaselayers}
\usepackage{pgfbasepatterns}
\usepackage{pgfbaseplot}
\usepackage{pgfbaseshapes}
\usepackage{pgfbasesnakes}
\usepackage{tikz}
\usetikzlibrary{math}
\usetikzlibrary{arrows}
\usetikzlibrary{shapes}
\usetikzlibrary{fit}
\usetikzlibrary{trees}
\usepackage{tikzpeople}
\usetikzlibrary{intersections}
\usepackage{caption}
\usepackage{todonotes}
\usepackage{comment}
\usepackage{makecell}

\newcommand{\tuple}[1]{\langle #1 \rangle}
\newcommand{\psh}{\mathsf{P}(\mathcal{SH})}

\newcommand{\rsh}{\mathsf{R}(\mathcal{SH})}

\newcommand{\pela}{\mathsf{P}(\mathcal{EL}^\rightarrow)}

\newcommand{\pelc}{\mathsf{P}(\mathcal{EL}^\circ_\bot)}

\newcommand{\rela}{\mathbb{R}(\mathcal{EL}^\rightarrow)}

\newcommand{\Homrcat}{\mathsf{Hom}(\mathscr{C}_r\langle \mathcal{O} \rangle)}

\newcommand{\occat}{\mathscr{C}_c\langle \mathcal{O} \rangle}

\newtheorem{example}{Example}
\newtheorem{theorem}{Theorem}

\newtheorem{definition}{Definition}
\newtheorem{lemma}{Lemma}

\newtheorem{corollary}{Corollary}
\newtheorem{claim}{Claim}

\newcommand{\Tb}[1]{\mathcal{#1}}

\newcommand{\role}[1]{\mathfrak{R}_{#1}}
\newcommand{\tO}{\mathcal{O}}
\newcommand{\subtransrole}{\overlay{\sqsubseteq}{\ast}}

\usepackage{enumitem}
\usepackage{hyperref}
\let\BBA\relax

\let\BBOP\relax
\let\BBCP\relax

\let\BBAA\relax
\let\BBAB\relax
\let\BBAY\relax
\let\BBC\relax
\let\BBN\relax

\usepackage{apacite}
\usepackage[round]{natbib} 
\usepackage{apalike}

\makeatletter
\newcommand{\mylabel}[2]{#2\def\@currentlabel{#2}\label{#1}}

\makeatother

\begin{document}

\title{A New Tractable Description Logic  under Categorical Semantics}

\author{\name Chan {Le Duc} \email chan.leduc@univ-paris13.fr \\
       \name Ludovic Brieulle \email ludovic.brieulle@univ-paris13.fr \\
       \addr Universit\'e Sorbonne Paris Nord, Sorbonne Université, INSERM, Limics\\
       74 rue Marcel Cachin, Bobigny, 93017 France}
       
\maketitle

\begin{abstract}
  Biomedical ontologies contain numerous concept or role names involving negative knowledge such as  \textsf{lacks\_part}, \textsf{absence\_of}. Such a representation with labels rather than logical constructors would not allow a reasoner to interpret \textsf{lacks\_part}  as a kind of negation of \textsf{has\_part}.  It is known that adding negation to the tractable  Description Logic (DL) $\mathcal{EL}$ allowing for conjunction, existential restriction and concept inclusion  makes it intractable since the obtained logic includes \emph{implicitly} disjunction and universal restriction which interact with other constructors. In this paper, we propose a new extension of   $\mathcal{EL}$ with a weakened negation allowing to represent  negative knowledge
while retaining tractability. To this end, we introduce categorical semantics of all logical constructors of the DL $\mathcal{SH}$ including 
$\mathcal{EL}$ with disjunction, negation,  universal restriction, role inclusion and transitive roles.
The categorical  semantics of a logical constructor is usually described as a set of categorical properties referring to several objects without using set membership. To restore tractability, we have to weaken  semantics of disjunction and universal restriction by  identifying \emph{independent} categorical properties  that are responsible for intractability, and dropping them from the set of categorical properties. We show that the logic resulting from  weakening semantics  is  more expressive than  $\mathcal{EL}$ with the bottom concept, transitive roles and role inclusion.
We aim through this new member of the $\mathcal{EL}$ family to allow negative knowledge to be involved in biomedical ontologies  while maintaining tractability of reasoning.
\end{abstract}

\section{Introduction}\label{sec:intro}

Description Logics (DLs) are used to represent knowledge from  different application domains. For instance, a large number of biomedical ontologies are represented in the DLs such as  $\mathcal{EL}$, 
 $\mathcal{EL}^{++}$  \cite{baader2005}  or its sublogics (as  OWL 2 Profile). The main reason of the use of these DLs  
is  reasoning  tractability that obliges us to drop negation.  However, as pointed out by  \cite{CEUSTERS2007} and \cite{RectorR06}, a substantial fraction of the observations made by clinicians in patient records is represented by means of negation via role or concept names (labels) such as \emph{lacks\_part}, \emph{absence\_of},    \emph{contraindication}  \cite{CEUSTERS2007,RectorR06}. Such a representation with labels would not allow a reasoner to interpret \emph{lacks\_part} or  \emph{absence\_of} as negation of   \emph{has\_part}. \cite{baader2005} showed that 
adding negation to the $\mathcal{EL}$ family immediately makes the resulting logic lose 
its tractability \cite{baader2005}. A logic is called tractable if the complexity of  main  reasoning problems  in that logic 
is polynomial time.   \cite{baader2005} also proved that augmenting  the expressiveness
of $\mathcal{EL}^{++}$ while keeping tractability in the worst case cannot be achieved if we remain to use the usual set semantics of DLs.  This is why we attempt in this paper to look into category theory for an alternative. 
 
A category consists of two collections of \emph{objects} and \emph{arrows} ($\rightarrow$). 
In most cases, arrows contribute to define the meaning  of objects and  relationships between them, e.g., $C\rightarrow D$ could mean that an object $C$ is ``less than" an object  $D$; or there is an arrow $X\rightarrow C$ ($X$ is ``less than" $C$) for all object $X$. The use of category theory in logic goes back to the work by \cite{law64} who
presented an appropriate axiomatization of the category of sets by replacing set 
membership with the composition of functions. However, it was not indicated whether the categorical axioms are ``semantically"  equivalent to the axioms based on set membership. As pointed out in   \cite{gol06}, this may lead to a very different semantics for negation.    A usual way to define the set-theoretical semantics of  propositional logic consists of using  a Boolean algebra   $\mathbf{BA}=\tuple{\mathcal{P}(D), \subseteq}$ 
where $\mathcal{P}(D)$ denotes the set of subsets of a non-empty set $D$ with  $\mathbf{1}=D$ and  $\mathbf{0}=\varnothing$.     A  \emph{valuation} (or truth function)  $V$ associates a truth-value in $\mathcal{P}(D)$ to each atomic sentence. Then $V$ is extended to complex sentences with $\wedge,\vee,\neg$ using set operations such as intersection $\cap$, union $\cup$, complement $\sim$ over $\mathcal{P}(D)$ with set membership $\in$. 
A notable extension of Boolean algebra is \emph{Heyting algebra} $\mathbf{HA}=\langle H, \subseteq \rangle$ where $H $  is a \emph{lattice} under a partial order $\subseteq$.  The truth function $V$ for $x \wedge y$ and $x \vee y$ is  defined respectively as a \emph{greatest lower bound} and a \emph{least upper bound} of $x,y$ under $\subseteq$ while  negation   is defined as   \emph{pseudo-complement}, i.e., $V(\neg x)$ is defined as a maximum element in $H$ such that $V(\neg x)$ is disjoint with $V(x)$. 
\cite{leduc2021}  and
\cite{brieulle2022}  replaced the lattice from a $\mathbf{HA}$ with categories  and  
rewrote the semantics of the DL $\mathcal{ALC}$ in such a way that  concepts/roles and the order $\subseteq$  are replaced 
with category objects and arrows   respectively.
Due to the absence of set membership  from categories,  arrows ``govern" the categorical semantics of DL constructors with  various properties involving different objects.  This makes the categorical semantics more modular than the usual set semantics because we usually need several  \emph{independent} categorical properties to describe the semantics of a single logical constructor. For instance, we have to use an independent categorical property to characterize  distributivity  of conjunction over disjunction since it does not need to hold in general. 
The present paper exploits the independence of some interesting semantic constraints  appearing in the categorical semantics to introduce a novel tractable DL.

In the literature, there have been several works related to tractable DLs such as $\mathcal{EL}$ \cite{Baader99} allowing for  conjunction $\sqcap$, existential restriction $\exists$ and concept inclusion $\sqsubseteq$, denoted $\mathcal{EL}=\tuple{\sqcap,\exists , \sqsubseteq}$. However,  if we replace existential restriction $\exists$ in $\mathcal{EL}$ with universal restriction $\forall$, then the resulting DL  $\mathcal{FL}_0=\tuple{\sqcap,\forall, \sqsubseteq}$ becomes intractable \cite{neb90,baader2005}. On the one hand, adding disjunction $\sqcup$ to $\mathcal{EL}$  also leads to intractability due to intrinsic   non-determinism of this constructor, which  is confirmed by  \cite{baader2005} who showed that subsumption  in $\mathcal{ELU}=\tuple{\sqcap,\exists, \sqsubseteq,\sqcup}$ is {\sc{EXPTIME}}-complete. On the other hand, 
adding atomic negation $.^{(\neg)}$ (i.e., negation occurring only in front of concept names) to $\mathcal{EL}$ implies an implicit presence of $\forall$ and $\sqcup$ in the resulting DL $\mathcal{EL}^{(\neg)}=\tuple{\sqcap,\exists, \sqsubseteq,.^{(\neg)}}$  \cite{baader2005}. Hence, it is necessary to investigate  behaviors of \emph{hidden} constructors $\forall$, $\sqcup$ and their interactions with other constructors  in reasoning algorithms  rather than the constructor $\neg$ itself in order to design a new tractable DL allowing for negation.

By analyzing the behavior of algorithms for reasoning in a DL allowing for  $\forall, \sqcup$ with the presence of $\sqcap, \exists, \sqsubseteq$, we conjecture that the interactions of the constructors $\sqcup$,  $\forall$ with other ones are   responsible for intractability. 
First, the interaction between existential and universal restrictions is represented by the following property that corresponds to the behavior of the $\forall$-rule in standard tableaux\footnote{If $\exists R.C$ and $\forall R.D$ belong to the label of a node $x$ in a standard tableau, then $x$ has an  $R$-neighbor $y$ such that $C$ belongs to the label of $y$. Then, the $\forall$-rule adds $D$ to the label of $y$. Hence, $\exists R.C\sqcap  \forall R.D$ behaves like $\exists R.(C \sqcap D)$.} \cite[p.84]{baader2017}.  
\begin{align}
\exists R.C\sqcap\forall R.D \sqsubseteq\exists R.(C\sqcap D)\label{intro-existsforall}
\end{align}
This property is a \emph{direct consequence} of the usual set semantics of $\exists$ and $\forall$. Applications of (\ref{intro-existsforall})  when reasoning  may  produce arbitrary \emph{$n$-ary existential restrictions} of the form $\exists R.(C_{i_1}\sqcap \cdots \sqcap C_{i_n})$ for $\{C_{i_1},\cdots,C_{i_n}\}\subseteq \{C_1,\cdots,C_m\}$ where $C_i$ occurs in some concept of the form  $\exists R.C_i$ or $\forall R.C_i$. This situation  happens if one uses DL $\mathcal{ALC}$ axioms to  encode binary numbers    \cite[Section 5.3]{MoSH09a}. To be able to avoid this exponential blow-up, we would weaken universal restriction by replacing Property~(\ref{intro-existsforall}) with the following one:
\begin{align}
X\sqsubseteq  \exists R.C\sqcap \forall R.D\text{ and } C\sqcap D \sqsubseteq \bot \text{ imply } X \sqsubseteq \bot\label{intro-existsforall2}
\end{align}
Thanks to Property~(\ref{intro-existsforall2}), it it is sufficient to consider \emph{binary fillers}   $C\sqcap D_1$ and $
C\sqcap D_2$  rather than  an \emph{$n$-filler}   $C\sqcap D_1\sqcap D_2$  ($n=3$)  when checking satisfiability of a concept $Y$ with 
an axiom such as $Y\sqsubseteq\exists R.C\sqcap \forall R.D_1 \sqcap \forall R.D_2$. The point is that applications of Property~(\ref{intro-existsforall2})   lead to checking  \emph{binary disjointness} such as   $C\sqcap D_1\sqsubseteq \bot, C\sqcap D_2\sqsubseteq \bot$  rather than creating a \emph{fresh} concept   $\exists R.(C\sqcap D_1 \sqcap D_2)$ by Property~(\ref{intro-existsforall})   and checking  \emph{n-disjointness} such as   $C\sqcap D_1 \sqcap D_2\sqsubseteq \bot$. As a result, Property~(\ref{intro-existsforall2}) 
allows us  to obtain unsatisfiability of $Y$  from checking \emph{only}  binary disjointness 
such as   $C\sqcap D_1\sqsubseteq \bot, C\sqcap D_2\sqsubseteq \bot$  
occurring frequently in biomedical ontologies rather than $n$-disjointness  occurring rarely, even not at all. In fact, there does not exist any  $n$-disjointness as axioms in most used biomedical ontologies such as SNOMED-CT, OBO. Note that an obvious unsatisfiability such as $A\sqcap \neg A$ is a particular case of  binary disjointness.
  
Second, distributivity of conjunction over disjunction
  is also a direct consequence of the usual set semantics, \emph{i.e.,} 
\begin{align}
A\sqcap (B\sqcup C) \sqsubseteq  (A\sqcap B) \sqcup (A\sqcap C)\label{intro-dist}
\end{align}
It is known that this interaction can generate an exponential number of conjunctions or force an algorithm to deal with  nondeterminism. For instance, the $\sqcup$-rule used in standard tableaux may produce an exponential number of tableaux   \cite[p.84]{baader2017}. In order to avoid this exponential blow-up, we would weaken disjunction   by removing Property~(\ref{intro-dist}) and retaining    the following properties:
\begin{align}
C_1\sqsubseteq X \text{ and } C_2\sqsubseteq X \text{ iff  }C_1\;\sqcup \;C_2\sqsubseteq X \label{intro:weakdist1}
\end{align}
If the semantics of a DL can be described as set of independent properties such as Properties~(\ref{intro-existsforall})-(\ref{intro:weakdist1}) (i.e., each logical constructor described by several properties instead of  the standard way with set membership), then it is  more flexible to introduce new logics with different complexity of reasoning. A category-based approach offers this possibility. In such an approach, it is surprising that all properties responsible for intractabitity are independent. Hence, a knowledge engineer could renounce Properties~(\ref{intro-existsforall}) and (\ref{intro-dist})   when modelling knowledge. The following example presents a situation where negative knowledge can be represented using a weakened negation/universal restriction constructors without losing expected results.

\begin{example}\label{ex:intro}
In SNOMED CT, there are numerous terms that represent negative knowledge such as  $\mathsf{does\_not\_contain}$,  $\mathsf{absence\_of}$,   $\mathsf{unable\_to},\mathsf{contraindication}$. One can find in this ontology,  a statement   ``$\mathsf{Renal\_agenesis}$ characterizes patients who do not have ($\mathsf{absence\_of}$)  $\mathsf{Kidney}$".
To take advantage of the semantic relationship between a relation  $\mathsf{has\_part}$   and a negative term such as $\mathsf{absence\_of}$,  one can express this relationship as an axiom in a DL allowing for $\tuple{\sqcap, \exists, \forall, \neg, \sqsubseteq}$  as follows:

\begin{itemize}
    \item[]  $\mathtt{Renal\_agenesis \sqsubseteq \forall has\_part. \neg Kidney}$
\end{itemize}
 
In an ontology about Clinical Practice Guidelines \cite{Abidi2009}, 
one can find  statements such as (i)  a patient who suffers from a bacterial infection should be medicated with Penicillin, and  (ii) some patient is allergic to Penicillin and Aspirin. If we wish to   populate such an ontology with \emph{Patient Discharge Summaries}\footnote{e.g., from  \url{https://portal.dbmi.hms.harvard.edu/}}, it would happen that the discharge summary of a patient $\mathsf{\#X}$ verifies both statements (i) and (ii). This situation can be expressed in the following   axiom: 

 $\{\mathsf{\#X}\} \sqsubseteq \mathsf{\exists \mathsf{medWith}.\mathsf{Penicillin} ~~\sqcap} \mathsf{\forall medWith.\neg Penicillin \sqcap  \forall medWith. \neg Aspirin}$\footnote{We adopt this syntax for ABox concept assertions  (Section \ref{sec:alchp})}

\noindent Note that the terms such as $\mathsf{Penicillin}, \mathsf{Aspirin}$ can be found in SNOMED CT. Thanks to Property~(\ref{intro-existsforall}), an algorithm like tableau  can entail that the right-hand side of the axiom is subsumed by  $\mathsf{\exists medWith. (Penicillin \sqcap \neg Penicillin \sqcap \neg Aspirin)}$, which   implies   $\{\mathsf{\#X}\} \sqsubseteq \bot$. In a medical context, it is more plausible  to interpret such an unsatisfiability to mean  ``$\{\mathsf{\#X}\}$ should not receive that treatment".  However, an algorithm  can  use just  Property~(\ref{intro-existsforall2}) to discover that the  right-hand side of the axiom is subsumed by  $\mathsf{\exists medWith. (Penicillin \sqcap \neg Penicillin)}$ which also leads to $\{\mathsf{\#X}\} \sqsubseteq \bot$.  \hfill$\square$
\end{example}

One can observe that best-known biomedical ontologies such as SNOMED CT, GO (Gene Ontology), ICD (International Classification of Diseases) do not contain disjunction in logical axioms.  However, a well-known ontology on diseases, called MONDO\footnote{From \url{https://purl.obolibrary.org/obo/mondo/mondo-base.owl}}, contains disjunction that occurs uniquely in the logical axioms of the form $A\equiv E \sqcup F$ 
 and  there is no axiom of the form  $B\sqsubseteq A$ such that  $B\neq E$ and $B\neq F$. 
 For example, one can find in MONDO the following  axiom: 
 $\mathsf{tuberculosis} \equiv\mathsf{pulmonary\_tuberculosis} \sqcup \mathsf{extrapulmonary\_tuberculosis}$. 
 Such a pattern of axioms with disjunction allows to avoid applying Property~(\ref{intro-dist})  in our algorithm (Section~\ref{sec:sem-equiv}). More general, if $E\sqsubseteq B$ and $F\sqsubseteq B$ hold, then    $B\sqcap (E\sqcup F)$ can be replaced with $E\sqcup F$. 
 

The main goal of our work consists of  weakening the usual set semantics of disjunction and universal restriction such that the weakened semantics is still  sufficient to  represent  negative knowledge in biomedical ontologies while preserving tractability.


The present paper is organized as follows.  In Sections~\ref{sec:alchp} and \ref{sec:cat-semantics}, we present the set-theoretical and categorical semantics of $\mathcal{SH}$, and shows that they are equivalent in Section~\ref{sec:sem-equiv}, i.e., an $\mathcal{SH}$ ontology is set-theoretically inconsistent iff it is categorically inconsistent. The proof of this equivalence  is technical and challenging because it has to build  a mapping that faithfully transfers the information related to the semantics from a fragmented structure generated by a non-deterministic algorithm to a more homogeneous structure generated by a deterministic one.
Based on this construction,  we  introduce in Section~\ref{sec:newDL} a new tractable DL,
namely $\mathcal{EL}^{\rightarrow}$ in a decreasing way, i.e., we drop from  $\mathcal{SH}$ the categorical properties  responsible for intractability. 
Hence,  $\mathcal{EL}^{\rightarrow}$ allows for all full constructors from  $\mathcal{SH}$ except for  $\forall$,  $\sqcup$ and $\neg$ that are weakened to $\bar{\forall}$, $\bar{\sqcup}$ and $\bar{\neg}$   in $\mathcal{EL}^{\rightarrow}$. 
 We have also proved  that $\mathcal{EL}^{\rightarrow}$ deprived of the new weakened DL constructors 
 is equivalent to $\mathcal{EL}_\bot^\circ$ ($\mathcal{EL}$ with the bottom concept and transitive roles). This implies that we can obtain $\mathcal{EL}^{\rightarrow}$ in an increasing way from $\mathcal{EL}_\bot^\circ$, i.e., the weakened constructors $\bar{\forall}$, $\bar{\sqcup}$ and $\bar{\neg}$ are independent from the other constructors in $\mathcal{EL}^{\rightarrow}$.
 Finally, we show that satisfiability in $\mathcal{EL}^{\rightarrow}$ is tractable. In Section~\ref{sec:rw}, we discuss related work. Section~\ref{sec:conc}  presents future work and discusses the obtained results.

%
%

\section{Set-theoretical Semantics of DLs} \label{sec:alchp}

In this section, we present syntax and set-theoretical semantics of  the DL   $\mathcal{SH}$ and $\mathcal{EL}_\bot^\circ$.
Let $\mathbf{C}$, $\mathbf{R}$ and  $\mathbf{I}$ be  non-empty sets of \emph{concept names}, \emph{role names} and  \emph{individuals} respectively.
The set of $\mathcal{SH}$-concepts is defined as the smallest set containing all concept names  in $\mathbf{C}$ with  $\top$, $\bot$ and complex concepts that are inductively defined as follows: $C\sqcap D$, $C\sqcup D$, $\neg C$, $\exists R.C$, $\forall R.C$ where  $C$ and $D$ are $\mathcal{SH}$-concepts, and $R$ is a role name in $\mathbf{R}$. We call
$C\sqsubseteq D$ a \emph{general concept inclusion} (GCI)  where $C$ and $D$ are $\mathcal{SH}$-concepts;  $R\sqsubseteq S$ a \emph{role inclusion} (RI); $S\circ S\sqsubseteq S$  an \emph{inclusion for transitive role} (ITR)  where $R,S\in\mathbf{R}$; $C(a)$ a \emph{concept assertion}  where $C$ is  
an $\mathcal{SH}$-concept
and $a\in\mathbf{I}$ an individual; $R(a,b)$ a \emph{role assertion}   where $R\in\mathbf{R}$ and $a,b\in\mathbf{I}$. For the sake of simplicity in the categorical context, we rewrite each concept assertion $C(a)$ as a \emph{concept assertion axiom} (CAA) $\{a\}\sqsubseteq C$, each role assertion $R(a,b)$ as a \emph{role assertion axiom} (RAA) $\{(a,b)\}\sqsubseteq R$.
An \emph{ontology} $\mathcal{O}$ is a finite set of axioms including GCIs, RIs, ITRs, concept assertions, role assertions. Note that nominals are not allowed in $\mathcal{SH}$. This rewriting of ABox assertions aims just to make statements in an $\mathcal{SH}$ ontology more homogeneous.  
A relation $\overlay{\ast}{\sqsubseteq}$ is inductively  defined as the transitive-reflexive closure $\mathcal{R}^+\tuple{\mathcal{O}}$  of $\sqsubseteq$ on the set of all RAAs,  RIs and ITRs of $\mathcal{O}$, i.e., (i)
$R\overlay{\ast}{\sqsubseteq}S\in  \mathcal{R}^+\tuple{\mathcal{O}}$ if either $R\sqsubseteq S$ is a RAA, RI or $R\overlay{\ast}{\sqsubseteq} R', R'\overlay{\ast}{\sqsubseteq} S\in \mathcal{R}^+\tuple{\mathcal{O}}$; 
and (ii) $S\circ S\overlay{\ast}{\sqsubseteq}S'\in  \mathcal{R}^+\tuple{\mathcal{O}}$ if   $S\circ S\overlay{\sqsubseteq}{\ast} R, R\overlay{\ast}{\sqsubseteq} S' \in \mathcal{R}^+\tuple{\mathcal{O}}$.   

The DL $\mathcal{EL}_\bot^\circ$ allows for  GCIs, the bottom concept,
conjunction, existential restriction, RIs, ITRs  while the DL $\mathcal{SH}$ includes $\mathcal{EL}_\bot^\circ$,   concept negation, disjunction and universal restriction.

We follow the DL literature in which a truth function is called an interpretation function in a Tarskian interpretation.  Such an interpretation $\mathcal{I}=\langle\Delta^{\mathcal{I}},\cdot^{\mathcal{I}}\rangle$ 
 consists of a non-empty set $\Delta^{\mathcal{I}}$ ({\em domain}), and a function 
 $\cdot^{\mathcal{I}}$ ({\em interpretation function}) which  associates a subset of 
 $\Delta^{\mathcal{I}}$ to each concept name, an element  in $\Delta^\mathcal{I}$ to each 
 individual, and a subset of $\Delta^{\mathcal{I}}\times \Delta^{\mathcal{I}}$ to each role
 name, such that  $\top^{\mathcal{I}} :=\Delta^{\mathcal{I}}$,  
$\bot^{\mathcal{I}}:=\varnothing$, $(C\sqcap D)^{\mathcal{I}}:=C^{\mathcal{I}}\cap D^{\mathcal{I}}$, 
$(C\sqcup D)^{\mathcal{I}}:=C^{\mathcal{I}}\cup D^{\mathcal{I}}, (\neg C)^{\mathcal 
I}:=\Delta^{\mathcal{I}}\setminus C^{\mathcal{I}}$, $(\exists R.C)^{\mathcal{I}}:= \{x\!\in\!\Delta^{\mathcal{I}}\!\mid\!\exists 
y\!\in\!C^{\mathcal{I}},\! (x,y)\!\in\!{R^{\mathcal{I}}}\}$ 
and
$(\forall R.C)^{\mathcal{I}}:=  \{x\!\in\!\Delta^{\mathcal{I}}\!\mid\! (x,y)\!\in\! {R^{\mathcal{I}}}\!\implies \! y\!\in\! 
C^{\mathcal{I}}\}$.
 
An interpretation $\mathcal{I}$ satisfies  a GCI   $C\sqsubseteq D$  if 
$C^{\mathcal{I}}\subseteq D^{\mathcal{I}}$; a SI $R\sqsubseteq S$ if $R^{\mathcal{I}}\subseteq S^{\mathcal{I}}$; an ITR $S\circ S\sqsubseteq S$    if $(x,y)\in S^\mathcal{I}, (y, z)\in S^\mathcal{I}$ imply   $(x, z)\in S^\mathcal{I}$; a CAA  $\{a\}\sqsubseteq C$ if $a^\mathcal{I}\in C^\mathcal{I}$; and a RAA $\{(a,b)\}\sqsubseteq R$ if $(a^\mathcal{I}, b^\mathcal{I})\in R^\mathcal{I}$.  We say that 
$\mathcal{I}$ is a model of  $\mathcal{O}$, written $\mathcal{I}\models  \mathcal{O}$, if   $\mathcal{I}$ satisfies  all axioms  of $\mathcal{O}$.   In this case, we say that $\mathcal{O}$ is {\em set-theoretically consistent}.
We say that an ontology is NNF (Negation Normal Form) if all concepts occurring in the ontology are written  in NNF,   i.e., negation occurs only in front of a concept name.  Since  De Morgan's laws, double negation $C\equiv \neg \neg C$, and dualities $\neg (\exists R.C)\equiv \forall R.\neg C$,  $\neg (\forall R.C)\equiv \exists  R.\neg C$  hold in $\mathcal{SH}$, all concepts in an $\mathcal{SH}$ ontology   are  polynomially convertible into NNF. 

 
%
%

\section{Categorical semantics}\label{sec:cat-semantics}

In this section, we introduce categorical semantics of different DL  constructors using \emph{objects} and  \emph{arrows} ($\rightarrow$) from  categories. This work extends the categorical properties for $\mathcal{ALC}$   presented by \citeauthor{leduc2021}  and \citeauthor{brieulle2022} to those for $\mathcal{SH}$ with individuals.  
To this end, we need  a concept category $\mathscr{C}_c$ and a role category $\mathscr{C}_r$  whose objects are respectively roles and concepts occurring in an   $\mathcal{SH}$ ontology. As usual, we use a \emph{valuation} or truth function $V$ that associates an object in $\mathscr{C}_c$ to each concept occurring in an ontology.  Then we use  a \emph{product} and  \emph{coproduct} of two objects $C,D$ in $\mathscr{C}_c$ to  ``interpret"  the categorical semantics of  $C \sqcap D$ and $C\sqcup D$, i.e., $V(C \sqcap D)=``\text{product of } C \text{ and } D"$, $V(C \sqcup D)=``\text{coproduct of } C \text{ and } D"$.  Note that \emph{product} and  \emph{coproduct} in a category $\mathscr{C}$ correspond to  a \emph{greatest lower bound} and a \emph{least upper bound} in a Heyting algebra $\mathbf{HA}=\tuple{H,\subseteq}$ if  $H$ and $\subseteq$ are considered as $\mathscr{C}$ and $\rightarrow$. We define   $V(\neg C)$ as simultaneously the greatest object (under $\rightarrow$) \emph{disjoint} with $V(C)$, and the least object  such that  the union of $V(C)$ and $V(\neg C)$ is the greatest object (or \emph{terminal object}) in $\mathscr{C}_c$. To define $V(\exists R.C)$ and $V(\forall R.C)$, we adapt \emph{left} and \emph{right adjoints} in a category that were used to characterize  unqualified existential and universal quantifiers in FOL \cite[p. 57]{saunders92}. Moreover, all axioms from an ontology are interpreted as arrows in  $\mathscr{C}_c$  and  $\mathscr{C}_r$. When all $V(C \sqcap D), V(C \sqcup D), V(\neg C), V(\exists R.C)$ and $V(\forall R.C)$ are defined in  $\mathscr{C}_c$, a concept $C$ is \emph{categorically unsatisfiable}   iff an arrow $V(C)\rightarrow \bot$ belongs to $\mathscr{C}_c$ where $\bot$ is the \emph{initial object} of $\mathscr{C}_c$. When  it is clear from the context, e.g. referring to category objects,   we  write $C$  instead of  $V(C)$. The following definition introduces basic properties of concept and role categories, and functors between them. 

\begin{definition}[concept and role categories]\label{def:syntax-cat}
Let  $\mathbf{R}$,  $\mathbf{C}$ and $\mathbf{I}$  be non-empty and disjoint sets of \emph{role names}, \emph{concept names} and \emph{individuals} respectively. We define a \emph{role category}
$\mathscr{C}_r$, and a \emph{concept category} $\mathscr{C}_c$  such that they are \emph{pre-order} catagories and the following properties are satisfied:

\begin{enumerate}[leftmargin=!, labelwidth=1em, align=left]
\item[\mylabel{org-syntax-role}{$\mathsf{P_r}$}$:$] 
    If $R\in\mathbf{R}$ 
    or   $R=\{(a,b)\}$ with $(a,b)\in \mathbf{I}\times \mathbf{I}$, 
    then $R$ is 
    a role object of $\mathscr{C}_r$. We use $\role{\bot}$ and $\role{\top}$ to denote 
    initial and terminal role objects of $\mathscr{C}_r$, i.e., $R\rightarrow \role{\top}$ and $\role{\bot}\rightarrow R$ are arrows for all role object $R$ in $\mathscr{C}_r$. Furthermore, there is an identity arrow $R\rightarrow R$ for all role object $R$, and composition of arrows holds in $\mathscr{C}_r$. 
    
    \item[\mylabel{org-syntax-concept}{$\mathsf{P_c}$}$:$] 
    If $C\in\mathbf{C}$
    or $C=\{a\}$  with $a\in \mathbf{I}$, 
    then $C$  is a concept object of 
    $\mathscr{C}_c$,  $\bot$ and $\top$ are respectively 
    \emph{initial} and \emph{terminal} concept objects of 
    $\mathscr{C}_c$, i.e., $C\rightarrow \top$ and $\bot\rightarrow C$ are arrows for all concept object  $C$ in $\mathscr{C}_c$. Furthermore, there is an identity arrow $C\rightarrow C$ for all concept  object $C$, and composition of arrows holds  in $\mathscr{C}_c$. For every role object $R$ of $\mathscr{C}_r$, there are concept objects 
    $R_{l}$ (left projection) and $R_{r}$ (right projection) in $\mathscr{C}_c$.
     
    \item[\mylabel{org-syntax-functor}{$\mathsf{P_f}$}$:$] 
    There are two functors $p_l, p_r : 
    \mathscr{C}_r \rightarrow \mathscr{C}_c$ such that if 
    $\mathscr{C}_r$ has an object $R$, then   it holds that 
    (i) $p_l(R)= R_{l}$, $p_r(R)= R_{r}$; (ii) $p_l(R)\rightarrow \bot$ iff $R\rightarrow \mathfrak{R}_\bot$, and (iii) 
      $p_r(R)\rightarrow\bot$ iff $R\rightarrow \mathfrak{R}_\bot$. Furthermore, these functors preserve initial and terminal objects, i.e., $p_l(\mathfrak{R}_\top)$ and $ p_r(\mathfrak{R}_\top)$ are isomorphic to $\top$; $p_l(\mathfrak{R}_\bot)$ and $p_r(\mathfrak{R}_\bot)$ are isomorphic to $\bot$.
    \end{enumerate}

\noindent We use  $\mathsf{Ob}(\mathscr{C})$ and $\mathsf{Hom}(\mathscr{C})$  to 
denote the sets of objects and arrows of $\mathscr{C}$ respectively. 
\end{definition}
 
Note that $\mathsf{Ob}(\mathscr{C})$ and $\mathsf{Hom}(\mathscr{C})$  are not necessarily  sets in a general category. The functors $p_l$ and $p_r$
are closely related to projections of a binary relation in set theory. To illustrate this, let 
$\mathfrak{R} = \{(x_1, y_1), (x_2, y_2)\}$ be a binary relation under set-theoretical semantics. 
Then $p_l$ represents the image of the left projection of $\mathfrak{R}$, 
$p_l(\mathfrak{R})=\{x_1, x_2\}$, and $p_r$  the right projection, $p_r(\mathfrak{R})=\{y_1,y_2\}$. As  all functors,  $p_l, p_r$  preserve arrows, i.e.,  $R\rightarrow S\in \mathsf{Hom}(\mathscr{C}_r)$ implies $p_l(R) \rightarrow p_l(S), p_r(R) \rightarrow p_r(S)\in \mathsf{Hom}(\mathscr{C}_c)$.  Notice that if a role or concept category has two arrows $f: X\rightarrow Y$ and $g: Y\rightarrow X$, written as $Y\leftrightarrows X$,   then $X$ and $Y$ are isomorphic. In fact, this implies that $\mathsf{id_X}=g\circ f$ and $\mathsf{id_Y}=f\circ g$  since these categories are pre-order.

 
Definition~\ref{def:syntax-cat} provides basic properties of concept and role categories from category theory. We need to equip $\mathscr{C}_c$ and $\mathscr{C}_r$ with more semantic constraints  coming from logical constructors, TBox axioms and ABox assertions of an $\mathcal{SH}$ ontology $\mathcal{O}$. 
In the sequel, we present the categorical  semantics of   logical constructors with the  set-theoretical  semantics counterpart. An interpretation $\mathcal{I}$ is employed whenever  referring to the set-theoretical semantics of a concept and role.  When an arrow $X\rightarrow Y$ is present in $\mathsf{Hom}(\mathscr{C})$, it  implies the presence of $X,Y$ in $\mathsf{Ob}(\mathscr{C})$ and  identity arrows $X\rightarrow X, Y\rightarrow Y$ in $\mathsf{Hom}(\mathscr{C})$. We write $X \leftrightarrows Y$ for $X \rightarrow Y$ and  $Y \rightarrow X$. 

Given an $\mathcal{SH}$ ontology $\mathcal{O}$, we define a signature $\mathcal{S}(\mathcal{O})=\tuple{\mathbf{C}, \mathbf{R}, \mathbf{I}}$ such that $\mathbf{C}$, $\mathbf{R}$ and $\mathbf{I}$ contain all concept names, role names and individuals  occurring in $\mathcal{O}$ respectively. Furthermore, $\mathbf{C}$ contains concept objects of the form $\{a\}$ with $a\in \mathbf{I}$, and $\mathbf{R}$ contains role objects of the form $\{(a,b)\}$ where $R(a,b)$ is an RAA in $\mathcal{O}$. In addition, $\mathbf{R}$ will be extended with fresh role names,  as described in the categorical properties below.  In this case,   $\mathbf{C}$, $\mathbf{R}$ and $\mathbf{I}$ in Definition~\ref{def:syntax-cat} refer to $\mathcal{S}(\mathcal{O})$.

\noindent\textbf{Axiom} (GCI, CAA, RAA, RI, ITR) $X\sqsubseteq Y$. 

\noindent\textit{Formal meaning.}  
\begin{enumerate}
 [leftmargin=0cm, itemindent=0.75cm] 
\item[\mylabel{org-axioms}{$\mathsf{P_\sqsubseteq}$}:] 
If $X\sqsubseteq   Y\in \mathcal{O}$ is a GCI or  CAA, then $X\rightarrow Y\in \mathsf{Hom}(\mathscr{C}_c)$. Furthermore, if $X\sqsubseteq   Y$ is a RAA, RI or ITR, then    $X\rightarrow Y\in \mathsf{Hom}(\mathscr{C}_r)$ and $p_l(X)\rightarrow p_l(Y), p_r(X)\rightarrow p_r(Y)\in \mathsf{Hom}(\mathscr{C}_c)$  by preservation of arrows via functors.  
\end{enumerate}

\noindent\textit{Informal meaning.} A knowledge engineer interprets  $X\sqsubseteq Y$ as set inclusion in the usual set semantics with set membership while he/she can interpret $X\rightarrow Y$ to mean that object $Y$ \emph{more general} than object $X$ (or $X$ is \emph{more specific} than $Y$) in the categorical semantics. Since this relation can define a partial order over objects, we can say that an object $Y$ is  \emph{the most general} or \emph{the most specific}. Thereafter, we use this order relation to provide informal meaning of other constructors. 

Regarding the operator of the composition of roles $S \circ S$ that occurs in an ITR such as $S\circ S\sqsubseteq S$, its semantics will be  taken into account in categorical properties related to universal and existential restrictions.

\noindent\textbf{Conjunction} $C\sqcap D$. 

\noindent\textit{Formal meaning.} There is a \emph{product object} $C\sqcap D$ of given objects $C,D$  in $\mathscr{C}_c$, that means \cite[p.46]{gol06},
\begin{enumerate}[leftmargin=!, labelwidth=1em, align=left]
\item[\mylabel{org-conj-d}{$\mathsf{P^d_\sqcap}$}:] 
$C\sqcap D\in \mathsf{Ob}(\mathscr{C}_c)$  implies $C\sqcap D \rightarrow C,  C\sqcap D \rightarrow D \in \mathsf{Hom}(\mathscr{C}_c)$\label{org-conj-1}
\item[\mylabel{org-conj-c}{$\mathsf{P^c_\sqcap}$}:]  $X\rightarrow C,  X\rightarrow D\in \mathsf{Hom}(\mathscr{C}_c), C\sqcap D\in \mathsf{Ob}(\mathscr{C}_c) \text{ imply }  X\rightarrow C\sqcap D \in \mathsf{Hom}(\mathscr{C}_c)$
\end{enumerate}

\noindent\textit{Informal meaning.} A knowledge engineer interprets  $C\sqcap D$  as the set of all elements simultaneously belonging  to $C$ and $D$  in the usual set semantics  while he/she can interpret   
$C\sqcap D$ to mean that it is the most general object but more specific than both $C$ and $D$  in the  categorical semantics.

\medskip

\noindent If we replace ``$\rightarrow$"  in \ref{org-conj-d} ($\mathsf{d}$ for decomposition) and \ref{org-conj-c}  ($\mathsf{c}$ for composition) with ``$\sqsubseteq$", then the modified \ref{org-conj-d}-\ref{org-conj-c} characterize   the usual set semantics of $\sqcap$ under $\mathcal{I}$ in the following sense.

\begin{lemma}\label{lem:sem-conj}
$(U\sqcap V)^\mathcal{I}=U^\mathcal{I}\cap V^\mathcal{I}$ for all concepts $U,V$  iff
\begin{align}
&(C \sqcap D)^\mathcal{I} \subseteq  C^\mathcal{I},    (C \sqcap D)^\mathcal{I} \subseteq  D^\mathcal{I}   \label{conj001} \\
&X\subseteq \Delta^\mathcal{I}, X \subseteq  C^\mathcal{I}, X  \subseteq  D^\mathcal{I} \text{ implies }
X  \subseteq (C \sqcap D)^\mathcal{I}  \label{conj002}
\end{align}
\end{lemma}

\begin{proof}
 \noindent ``$\Longleftarrow$". Due to (\ref{conj001}) we have $(U\sqcap V)^\mathcal{I}  \subseteq U^\mathcal{I}\cap V^\mathcal{I}$. Let $X\subseteq \Delta^\mathcal{I}$   such that $X =U^\mathcal{I}\cap V^\mathcal{I}$. This implies that $X \subseteq U^\mathcal{I}$ and $X \subseteq V^\mathcal{I}$.   From (\ref{conj002}), we have  $X\subseteq (U\sqcap V)^\mathcal{I}$.  
 
 \noindent ``$\Longrightarrow$".  From $(C\sqcap D)^\mathcal{I}=C^\mathcal{I}\cap D^\mathcal{I}$   we obtain (\ref{conj001}).    Moreover, if   $X\subseteq C^\mathcal{I}$ and $X\subseteq D^\mathcal{I}$ then $X\subseteq C^\mathcal{I}\cap D^\mathcal{I}=(C\sqcap D)^\mathcal{I}$ by the hypothesis. Thus, (\ref{conj002}) is proved.
\end{proof}

\noindent\textbf{Disjunction} $C\sqcup D$. 

\noindent\textit{Formal meaning.} There is a \emph{co-product object} $C\sqcup D$ of objects $C,D$  in $\mathscr{C}_c$, that means \cite[p.54]{gol06}, 
\begin{enumerate}[leftmargin=!, labelwidth=1em, align=left]
\item[\mylabel{org-disj-d}{$\mathsf{P^d_\sqcup}$}:]     $C\sqcup D \in  \mathsf{Ob}(\mathscr{C}_c)$  implies  $C \rightarrow C \sqcup D,   D \rightarrow  C \sqcup D \in \mathsf{Hom}(\mathscr{C}_c)$.   

\item[\mylabel{org-disj-c}{$\mathsf{P^c_\sqcup}$}:]  $C\rightarrow  X, D \rightarrow  X  \in \mathsf{Hom}(\mathscr{C}_c)$, $C\sqcup D\in \mathsf{Ob}(\mathscr{C}_c)$
imply $C \sqcup D  \rightarrow X \in \mathsf{Hom}(\mathscr{C}_c)$.
\end{enumerate}

\noindent\textit{Informal meaning.}
A knowledge engineer interprets  $C\sqcup D$  as a set of all elements of $C$ and $D$ in the usual set semantics   while he/she can interpret
$C\sqcup D$ to mean that it is  the most specific object but more general than both  $C$ and $D$ in the categorical semantics.

\medskip
\noindent If we replace ``$\rightarrow$" in \ref{org-disj-d} and \ref{org-disj-c}  with ``$\sqsubseteq$", then  the modified ones characterize  the usual set semantics of $\sqcup$ under $\mathcal{I}$ in the following sense.  

\begin{lemma}\label{lem:sem-disj}  
$(C\sqcup D)^\mathcal{I}=C^\mathcal{I}\cup D^\mathcal{I}$  iff
\begin{align}
&C^\mathcal{I} \subseteq (C  \sqcup D)^\mathcal{I  }, D^\mathcal{I}  \subseteq  (C  \sqcup D)^\mathcal{I}  \label{disj1}   \\
&X\subseteq \Delta^\mathcal{I},  C^\mathcal{I}    \subseteq  X,   D^\mathcal{I}  \subseteq  X  \text{ implies } (C \sqcup D)^\mathcal{I}  \subseteq X   \label{disj2} 
\end{align}
\end{lemma}
\begin{proof}
 \noindent ``$\Longleftarrow$". Due to (\ref{disj1}) we have $C^\mathcal{I}\cup D^\mathcal{I} \subseteq (C\sqcup D)^\mathcal{I}$. Let   $X=C^\mathcal{I}\cup D^\mathcal{I}$. Due to (\ref{disj2}) we have  $(C\sqcup D)^\mathcal{I} \subseteq X=C^\mathcal{I}\cup D^\mathcal{I}$. 
 
 \noindent ``$\Longrightarrow$".  From $(C\sqcup D)^\mathcal{I}=C^\mathcal{I}\cup D^\mathcal{I}$, we have (\ref{disj1}).  Let $x\in (C\sqcup D)^\mathcal{I}$. Due to $(C\sqcup D)^\mathcal{I}=C^\mathcal{I}\cup D^\mathcal{I}$, we have  $x\in C^\mathcal{I}$ or $x\in D^\mathcal{I}$. Hence, $x\in X$  since $C^\mathcal{I}    \subseteq  X$ and  $D^\mathcal{I}  \subseteq  X$.  
\end{proof}

Note that distributivity of conjunction over disjunction cannot be proved from the categorical properties related to $\sqcap$ and $\sqcup$, as showed in the following example.

\begin{example}\label{ex:distrib} Let  $\mathsf{C}_0=\mathsf{F} \sqcap (\mathsf{D}\sqcup \mathsf{S}) \sqcap (~(\mathsf{F}\sqcap \mathsf{D}) \sqcup  (\mathsf{F}\sqcap \mathsf{S})~)$, and $\mathscr{C}_c$ a concept ontology category including  $\mathsf{C_0}$. 
 If we apply exhaustively \ref{org-conj-d} and \ref{org-disj-d} to $\mathscr{C}_c$, then 
 $\mathsf{Hom}(\mathscr{C}_c)$ contains the following arrows (for the sake of simplicity, we omit initial, terminal and identity arrows as well as compositions of arrows):

{
\centering
\begin{tabular}{lrlr}
$\mathsf{C_0}$ & $\rightarrow$ &$\mathsf{F} \sqcap (\mathsf{D}\sqcup \mathsf{S})$  & by \ref{org-conj-d} \\
$\mathsf{C_0}$ & $\rightarrow$  & $(\mathsf{F}\sqcap \mathsf{D}) \sqcup  (\mathsf{F}\sqcap 
\mathsf{S})$ & (idem)\\
$\mathsf{F} \sqcap (\mathsf{D}\sqcup \mathsf{S})$ &$\rightarrow$ & $\mathsf{F}$  & (idem)\\
$\mathsf{F} \sqcap (\mathsf{D}\sqcup \mathsf{S})$ & $\rightarrow$ & $\mathsf{D}\sqcup 
\mathsf{S}$ &  (idem)\\
$\mathsf{D}$ & $\rightarrow$ & $\mathsf{D}\sqcup \mathsf{S}$ & by  \ref{org-disj-d}\\
$\mathsf{S}$ & $\rightarrow$ & $\mathsf{D}\sqcup \mathsf{S}$ & (idem)\\
$\mathsf{F}\sqcap \mathsf{D}$ & $\rightarrow$ & $\mathsf{F}$   & by  \ref{org-conj-d}\\
$\mathsf{F}\sqcap \mathsf{D}$ & $\rightarrow$ & $\mathsf{D}$   &   (idem)\\
$\mathsf{F}\sqcap \mathsf{D}$ & $\rightarrow$ & $(\mathsf{F}\sqcap 
\mathsf{D})\sqcup (\mathsf{F}\sqcap \mathsf{S})$  & by \ref{org-disj-d}\\
$\mathsf{F}\sqcap \mathsf{S}$ & $\rightarrow$ & $\mathsf{F}$ & by \ref{org-conj-d}\\
$\mathsf{F}\sqcap \mathsf{S}$ & $\rightarrow$ & $\mathsf{S}$ & (idem)\\
$\mathsf{F}\sqcap \mathsf{S}$ & $\rightarrow$ & $(\mathsf{F}\sqcap \mathsf{D})\sqcup 
(\mathsf{F}\sqcap \mathsf{S})$ & by \ref{org-disj-d}
\end{tabular}\par
}
 
\medskip

The application of  \ref{org-conj-d} and \ref{org-disj-c} to $\mathscr{C}_c$ leads to  adding only  identity arrows. This suggests that   $\mathsf{F} \sqcap (\mathsf{D}\sqcup \mathsf{S})\rightarrow (\mathsf{F}\sqcap \mathsf{D}) \sqcup  (\mathsf{F}\sqcap \mathsf{S})$ is not derivable  from \ref{org-conj-d} to \ref{org-disj-c}.
\end{example}

\noindent\textbf{Negation} $\neg C$. We can define negation by using categorical product and co-product via conjunction and disjunction.

\noindent\textit{Formal meaning.}

\begin{enumerate}[leftmargin=!, labelwidth=1em, align=left]
\item[\mylabel{org-neg-sqcap}{$\mathsf{P}_\neg^\sqcap$}$:$]$C\in  \mathsf{Ob}(\mathscr{C}_c)$ \textit{implies}    
$C \sqcap \neg C \rightarrow \bot \in  \mathsf{Hom}(\mathscr{C}_c)$
\item[\mylabel{org-neg-sqcup}{$\mathsf{P}_\neg^\sqcup$}$:$]$C\in  \mathsf{Ob}(\mathscr{C}_c)$ \textit{implies}   
$   \top\rightarrow C \sqcup \neg C\in  \mathsf{Hom}(\mathscr{C}_c)$
\item[\mylabel{org-neg-bot}{$\mathsf{P}_\neg^{\bot}$}$:$]  $C \sqcap X \rightarrow \bot \in  \mathsf{Hom}(\mathscr{C}_c)$ \textit{implies}  $X \rightarrow \neg C \in  \mathsf{Hom}(\mathscr{C}_c)$

\item[\mylabel{org-neg-top}{$\mathsf{P}_\neg^{\top}$}$:$] $\top  \rightarrow C \sqcup X \in  \mathsf{Hom}(\mathscr{C}_c)$ \textit{implies}   $\neg C \rightarrow X \in  \mathsf{Hom}(\mathscr{C}_c)$
\end{enumerate}

\noindent\textit{Informal meaning.}  A knowledge engineer can interpret  $\neg C$  as  a set of all elements in $\top$  that are not in $C$ in the usual set semantics  while he/she can interpret
$\neg C$ to mean that  it is the most specific  object such that $C \sqcup \neg C$ is more general than   $\top$, and the most general object such that $C \sqcap \neg C$  is  more specific  than  $\bot$ in the categorical semantics.

\medskip

If we replace ``$\rightarrow$" in \ref{org-neg-sqcap}, \ref{org-neg-sqcup}, \ref{org-neg-bot} and  \ref{org-neg-top}  with ``$\sqsubseteq$", then the modified ones characterize the usual set semantics of $\neg$ under $\mathcal{I}$  in the following sense.

\begin{lemma}\label{lem:sem-neg}
$\neg C^\mathcal{I}=\Delta^\mathcal{I}\setminus C^\mathcal{I}$ iff
\begin{align}
&C^\mathcal{I} \cap \neg C^\mathcal{I}\subseteq \bot^\mathcal{I}, \top^\mathcal{I} \subseteq  C^\mathcal{I} \cup \neg C^\mathcal{I}\label{neg2}\\
& (X \sqcap C)^\mathcal{I} \subseteq \bot^\mathcal{I} \text{ implies } X^\mathcal{I} \subseteq (\neg C)^\mathcal{I} \label{neg3} \\
& \top^\mathcal{I} \subseteq (C \sqcup X)^\mathcal{I}  \text{ implies } (\neg C)^\mathcal{I} \subseteq X^\mathcal{I} \label{neg4}
\end{align}

\end{lemma}

\begin{proof}
``$\Longrightarrow$". 
Assume  $\neg C^\mathcal{I}=\Delta^\mathcal{I}\setminus C^\mathcal{I}$. It is obvious that (\ref{neg2}) holds.
Let $x\in X^\mathcal{I}$ with $(C\sqcap X)^\mathcal{I}\subseteq \bot^\mathcal{I}$.  Due to Lemma~\ref{lem:sem-conj}, we have $(C\sqcap X)^\mathcal{I}=C^\mathcal{I} \cap X^\mathcal{I}$. Therefore, $x\notin C^\mathcal{I}$, and thus $x\in (\neg C)^\mathcal{I}$. Let $x\in (\neg C)^\mathcal{I}$ with $\top^\mathcal{I} \subseteq (C\sqcup X)^\mathcal{I} $. It follows that $x\notin C^\mathcal{I}$. Due to Lemma~\ref{lem:sem-disj}, we have $(C\sqcup X)^\mathcal{I}=C^\mathcal{I} \cup X^\mathcal{I}$. Therefore,   $x\in X^\mathcal{I}$.

\noindent ``$\Longleftarrow$". This direction trivially holds since (\ref{neg2}) implies $\neg C^\mathcal{I}=\Delta^\mathcal{I}\setminus C^\mathcal{I}$.
\end{proof}

\noindent\textbf{Existential restriction}. 
The set-theoretical semantics of unqualified existential quantifier $\exists R$ in FOL is defined as all elements $x$ such that there is some $(x,y)\in R$. To rewrite it in categorical language over the category of sets, one can use left adjoint to the pullback (or inverse) functor $p^\ast$ of the projection $p$   from $\Delta\times \Delta$ to $\Delta$   \cite[Theorem 1, p.57]{saunders92}. More precisely, if we define $p^\ast(X)=\{(x,y)\in \Delta\times \Delta\mid x\in X\}$ for $X\subseteq \Delta$, then it holds that $\exists R\subseteq X$ iff $R\subseteq p^\ast(X)$ ($\dagger$). Hence,   $\exists R$  can be viewed as   left adjoint to $p^\ast(X)$.   To adapt this result to qualified existential quantifier $\exists R.C$, we need to introduce $p_C^\ast(X)=\{(x,y)\in \Delta\times \Delta\mid x\in X \wedge y\in C\}$ for $X\subseteq \Delta$. Then,   we can show that for $X\subseteq \Delta$, $\exists R.C\subseteq X$ iff $R\subseteq p_C^\ast(X)$   ($\dagger\dagger$), i.e  $\exists R.C$ describes the greatest  set including all $x$ such that $(x,y)\in R$ and $y\in C$. To capture the right-hand side of ($\dagger\dagger$), we use a categorical property with a new role object $\mathfrak{R}^\varnothing_{(\exists R.C)}$ such that $\mathfrak{R}^\varnothing_{(\exists R.C)}  \rightarrow R$ and $p_r(\mathfrak{R}^\varnothing_{(\exists R.C)}) \rightarrow  C$ (the superscript $\varnothing$ means $\mathfrak{R}^\varnothing_{(\exists R.C)}$ depends on no universal restriction). For the left-hand side of ($\dagger\dagger$), we use  a second categorical property that expresses maximality of  $\exists R.C$. 

 
\noindent\textit{Formal meaning.}  
\begin{enumerate}[leftmargin=!, labelwidth=1em, align=left]
\item[\mylabel{org-exists-filler}{$\mathsf{P_\exists^\mathfrak{R}}$}$:$]
  $\exists R.C\in \mathsf{Ob}(\mathscr{C}_c)$ implies  $\mathfrak{R}^\varnothing_{(\exists R.C)}  \rightarrow R\in \mathsf{Hom}(\mathscr{C}_r), p_r(\mathfrak{R}^\varnothing_{(\exists R.C)})  \rightarrow C$, $p_l(\mathfrak{R}^\varnothing_{(\exists R.C)})\leftrightarrows \exists R.C\in \mathsf{Hom}(\mathscr{C}_c)$
  
\item[\mylabel{org-exists}{$\mathsf{P_\exists}$}$:$]  $\exists R.C\in \mathsf{Ob}(\mathscr{C}_c), R'\rightarrow R\in \mathsf{Hom}(\mathscr{C}_r)$,  $p_r(R') \rightarrow C \in \mathsf{Hom}(\mathscr{C}_c) \text{ imply }  p_l(R')\rightarrow \exists R.C \in \mathsf{Hom}(\mathscr{C}_c)$ 
\end{enumerate}

\noindent\textit{Informal meaning.} A knowledge engineer interprets $\exists R.C$ as all  elements having a relation $R$ to an element in $C$ in the usual set semantics while he/she can interpret $\exists R.C$ as  the left projection of the most general role object $R'$ but more specific than $R$ such that the right projection of $R'$ is more specific than $C$ in  the categorical semantics.

\medskip

Intuitively, $\mathfrak{R}^\varnothing_{(\exists R.C)}$  represents all pairs $x,y$ such that $(x,y)\in R^\mathcal{I}$ and $y\in C^\mathcal{I}$ in the set-theoretical semantics. If we replace ``$\rightarrow$" in \ref{org-exists-filler}-\ref{org-exists}  with ``$\sqsubseteq$", then  the modified \ref{org-exists-filler}-\ref{org-exists} characterize   the usual set semantics of $\exists$ under $\mathcal{I}$  in the following sense.

\begin{lemma}\label{lem:sem-exists}
Assume that 
\begin{align}
&(\mathfrak{R}^\varnothing_{(\exists R.C)} )^\mathcal{I} \subseteq R^\mathcal{I} \text{ and }p_r(\mathfrak{R}^\varnothing_{(\exists R.C)} )^\mathcal{I} \subseteq C^\mathcal{I} \label{ex-arrow2}
\end{align}
It holds that
\begin{align}
& p_l(\mathfrak{R}^\varnothing_{(\exists R.C)} )^\mathcal{I}\text{=}\{x\in \Delta^\mathcal{I}\mid \exists y\in C^\mathcal{I} \wedge (x,y)\in R^\mathcal{I}\} \label{ex-arrow3}\\
&\text{ iff } {R'}^\mathcal{I}\subseteq R^\mathcal{I}, p_r(R') \subseteq C^\mathcal{I}\text{ imply } p_l(R')^\mathcal{I}\subseteq p_l(\mathfrak{R}^\varnothing_{(\exists R.C)})^\mathcal{I} \label{ex-arrow4}
\end{align}
\end{lemma}

\begin{proof}
\noindent ``$\Longrightarrow$". Assume that (\ref{ex-arrow3}) holds.
For $R'$  such that ${R'}^\mathcal{I}\subseteq R^\mathcal{I}$,  $p_r(R')^\mathcal{I} \subseteq C^\mathcal{I}$, let   $x\in p_l(R')^\mathcal{I} $. There is some  $y\in p_r(R')^\mathcal{I} \subseteq C^\mathcal{I}$ such that $(x,y)\in {R'}^\mathcal{I}\subseteq R^\mathcal{I}$. This implies that $x\in \exists R.C^\mathcal{I}$. By the hypothesis  $\exists R.C^\mathcal{I}=p_l(\mathfrak{R}^\varnothing_{(\exists R.C)} )^\mathcal{I}$, we have  $x\in p_l(\mathfrak{R}^\varnothing_{(\exists R.C)} )^\mathcal{I}$.
    
\noindent ``$\Longleftarrow$". Assume that (\ref{ex-arrow4}) holds. Let  $x\in \{x'\in \Delta^\mathcal{I}\mid  \exists y\in C^\mathcal{I} \wedge (x',y)\in R^\mathcal{I}\}$. This implies that there is some $y\in C^\mathcal{I}$  with $(x,y)\in R^\mathcal{I}$. We take $R'$ such that   ${R'}^\mathcal{I}=\{(x,y)\}$.  We have ${R'}^\mathcal{I}\subseteq {R}^\mathcal{I}$ and  $p_r(R')^\mathcal{I} \subseteq C^\mathcal{I}$. Hence, $x\in p_l(\mathfrak{R}^\varnothing_{(\exists R.C)} )^\mathcal{I}$ by hypothesis.

Let  $x\in p_l(\mathfrak{R}^\varnothing_{(\exists R.C)} )^\mathcal{I}$.  This implies that there is some $y\in p_r(\mathfrak{R}^\varnothing_{(\exists R.C)} )^\mathcal{I}$  with $(x,y)\in (\mathfrak{R}^\varnothing_{(\exists R.C)} )^\mathcal{I}$. By   (\ref{ex-arrow2}) we have   $(x,y)\in R^\mathcal{I}$ and $y\in C^\mathcal{I}$. This implies that $x\in \{x'\in \Delta^\mathcal{I}\mid  \exists y\in C^\mathcal{I} \wedge (x',y)\in R^\mathcal{I}\}$.
\end{proof}

\noindent\textbf{Universal restriction}. If  unqualified universal quantifier  $\forall R$ in FOL represents all elements $x$ such that $(x,y)\in R$ for all $y$, then \cite{saunders92} showed that  it can be viewed as right adjoint to the pullback  functor $p^\ast$ of the projection $p$   from $\Delta\times \Delta$ to $\Delta$ \cite[Theorem 1, p.57]{saunders92}.
However, it is not relevant to exploit this rewriting   since  $\forall R$ in FOL  is semantically different from   $\forall R.C$  in DLs while the latter can be expressed using $\exists$ and $\neg$ thanks to  $\forall R.C\equiv \neg \exists R.\neg C$ under the set-theoretical semantics. Hence, we define  the categorical semantics of $\forall R.C$ using the categorical semantics of $\neg$ and $\exists$ as follows.

\noindent\textit{Formal meaning.}
\begin{enumerate}[leftmargin=0cm, itemindent=0.8cm]
\item[\mylabel{org-forall-weakened}{$\mathsf{P_\forall^-}$}$:$] $\forall R.C\in \mathsf{Ob}(\mathscr{C}_c)$ implies   $\forall R.C \leftrightarrows \neg \exists R.\neg C\in  \mathsf{Hom}(\mathscr{C}_c)$
\end{enumerate}

\noindent It it straightforward to  show that if we replace ``$\leftrightarrows$" in \ref{org-forall-weakened}   with  ``$\equiv$" (i.e., $\sqsubseteq$ and $\sqsupseteq$), then  the modified  \ref{org-forall-weakened} characterizes the usual set semantics of $\forall$ under $\mathcal{I}$. 
  
\noindent If we use only the properties from \ref{org-syntax-role} to \ref{org-forall-weakened}
 to characterize the categorical semantics of logical constructors, some interactions between them are missing. As indicated in Example~\ref{ex:distrib}, distributivity of conjunction over disjunction is not provable   from the categorical properties related to conjunction and disjunction. In this case, we say that the distributivity is   \emph{independent} from   \ref{org-conj-d}-\ref{org-disj-c}  while it is prouvable from conjunction and disjunction under the set-theoretical semantics. This phenomenon happens again to  the interaction between existential and universal restrictions. 
 For this reason, we need further categorical properties to be able to establish the equivalence between set-theoretical and categorical semantics. A knowledge engineer can consider them as options rather than  consequences of the semantics resulting from their previous choices of constructors.

\noindent\textbf{Interaction between $\sqcap$ and $\sqcup$}.  In a DL allowing for both $\sqcap$ and $\sqcup$, the interaction between them can be expressed by  the following property.

\begin{enumerate}[leftmargin=0cm, itemindent=0.8cm]
\item[\mylabel{org-dist}{$\mathsf{P_\sqcap^\sqcup}$}$:$] $X \rightarrow C, X\rightarrow D\sqcup E\in \mathsf{Hom}(\mathscr{C}_c)$ implies $X \rightarrow (C \sqcap  D ) \sqcup (C\sqcap E)\in \mathsf{Hom}(\mathscr{C}_c)$. 
\end{enumerate}

\noindent\textbf{Interaction between $\exists$ and $\circ$}.  
\begin{enumerate}[leftmargin=!, labelwidth=1em, align=left]
\item[\mylabel{org-exists-comp}{$\mathsf{P_\exists^{\circ}}$}$:$]  $X\rightarrow \exists S.C,  C\rightarrow \exists S.D \in \mathsf{Hom}(\mathscr{C}_c)$, $S\circ S\rightarrow S\in \mathsf{Hom}(\mathscr{C}_r)$ imply  $X\rightarrow \exists S.D \in \mathsf{Hom}(\mathscr{C}_c)$.
\end{enumerate}

\noindent\textbf{Interaction between $\exists$, $\forall$ and $\circ$}.  
\begin{enumerate}[leftmargin=!, labelwidth=1em, align=left]
\item[\mylabel{org-exists-forall}{$\mathsf{P_\exists^{\forall}}$}$:$]  $X\rightarrow p_l(\mathfrak{R}^\mathcal{C}_{(\exists P.C)}), X\rightarrow \forall R.D \in \mathsf{Hom}(\mathscr{C}_c)$, $P\rightarrow R\in \mathsf{Hom}(\mathscr{C}_r)$ 
imply  $\mathfrak{R}^{\mathcal{C}\cup\{\forall R.D\}}_{(\exists P.C)}\rightarrow \mathfrak{R}^{\mathcal{C}}_{(\exists P.C)}\in \mathsf{Hom}(\mathscr{C}_r)$, $    X\leftrightarrows p_l(\mathfrak{R}^{\mathcal{C}\cup\{\forall R.D\}}_{(\exists P.C)})$, $p_r(\mathfrak{R}^{\mathcal{C}\cup\{\forall R.D\}}_{(\exists P.C)})\rightarrow D \in \mathsf{Hom}(\mathscr{C}_c)$.

\item[\mylabel{org-exists-forall-ind}{$\mathsf{P_\exists^{\forall I}}$}$:$]
$\{(a,b)\}\rightarrow P \rightarrow R\in \mathsf{Hom}(\mathscr{C}_r)$ and $ \{a\}\rightarrow \forall R.D  \in \mathsf{Hom}(\mathscr{C}_c)$ imply $\{b\}\rightarrow  D\in \mathsf{Hom}(\mathscr{C}_c)$.
  
\item[\mylabel{org-exists-forall-comp}{$\mathsf{P_{\exists\circ}^{\forall}}$}$:$]  $X\rightarrow  p_l(\mathfrak{R}^\mathcal{C}_{(\exists P.C)})  , X\rightarrow  \forall R.D \in\mathsf{Hom}(\mathscr{C}_c),  P\rightarrow S, S\circ S\rightarrow S, S\rightarrow R \in \mathsf{Hom}(\mathscr{C}_r)$  imply $\mathfrak{R}^{\mathcal{C}\cup \{\forall R.D\}}_{(\exists P.C)}\rightarrow \mathfrak{R}^{\mathcal{C}}_{(\exists P.C)}\in \mathsf{Hom}(\mathscr{C}_r)$, $X \leftrightarrows  p_l(\mathfrak{R}^{\mathcal{C}\cup \{\forall R.D\}}_{(\exists P.C)}), p_r(\mathfrak{R}^{\mathcal{C}\cup \{\forall R.D\}}_{(\exists P.C)}) \rightarrow   \forall S.D \in \mathsf{Hom}(\mathscr{C}_c)$.  

\item[\mylabel{org-exists-forall-comp-ind}{$\mathsf{P_{\exists\circ}^{\forall I}}$}$:$] $\{(a,b)\}\rightarrow P \rightarrow S \rightarrow R, S\circ S\rightarrow S\in \mathsf{Hom}(\mathscr{C}_r)$ and $ \{a\}\rightarrow \forall R.D  \in \mathsf{Hom}(\mathscr{C}_c)$ imply $\{b\}\rightarrow  \forall S.D\in \mathsf{Hom}(\mathscr{C}_c)$.
\end{enumerate}

We can observe that \ref{org-exists-forall} and  \ref{org-exists-forall-comp} lead to adding a fresh role object $\mathfrak{R}^{\mathcal{C}' }_{(\exists P.C)}$ that  depends only on    $\exists P.C$ and a set of universal restrictions $\mathcal{C}'=\mathcal{C}\cup \{\forall R.D\}$. That means that if there are arrows $X\rightarrow   p_l(\mathfrak{R}^{\{\forall R.D_1\}}_{(\exists P.C)})$,    $X\rightarrow \forall R.D_2$,  $P\rightarrow R$  in the categories, then it is necessary to introduce a fresh role object $\mathfrak{R}^{\mathcal{C}}_{(\exists P.C)}$ with $\mathcal{C}=\{\forall R.D_1, \forall R.D_2\}$ and an arrow $p_r(\mathfrak{R}^{\mathcal{C}}_{(\exists P.C)})\rightarrow D_2$   according to \ref{org-exists-forall}. 
Therefore, $\mathfrak{R}^{\mathcal{C}}_{(\exists P.C)}$ represents pairs $(x,y)\in P^\mathcal{I}$ such that    $y \in C^\mathcal{I} \cap     D_1^\mathcal{I} \cap     D_2^\mathcal{I}$ in the usual set semantics. Furthermore, if some $Y$ occurs in the place of  $X$ in the arrows above, i.e   $Y\rightarrow \exists P.C\leftrightarrows p_l(\mathfrak{R}^\varnothing_{(\exists P.C)})$, $Y\rightarrow \forall R.D_1$ and   $Y\rightarrow \forall R.D_2$, then no new fresh role object is needed.

If we replace ``$\rightarrow$" in \ref{org-dist}, \ref{org-exists-comp}, 
\ref{org-exists-forall}, \ref{org-exists-forall-ind},   \ref{org-exists-forall-comp}, and  \ref{org-exists-forall-comp-ind}  with ``$\sqsubseteq$", then  the modified properties  trivially hold  in  the usual set semantics   under $\mathcal{I}$ in the following sense.

\begin{lemma}\label{lem:sem-interactions} It holds that 
    \begin{enumerate}[leftmargin=0cm, itemindent=0.5cm]
        \item\label{lem:sem-interactions:0} If $X^\mathcal{I}\subseteq C^\mathcal{I}, X^\mathcal{I} \subseteq (D^\mathcal{I}\cup E^\mathcal{I})$, then   $X^\mathcal{I}\subseteq (C^\mathcal{I}\cap D^\mathcal{I})\cup (C^\mathcal{I}\cap E^\mathcal{I})$

        \item\label{lem:sem-interactions:1} If $X^\mathcal{I}\subseteq   \exists S.C^\mathcal{I}$, $C \subseteq   \exists S.D^\mathcal{I}, S\circ S^\mathcal{I} \subseteq  S^\mathcal{I}$, then $X^\mathcal{I}\subseteq   \exists S.D^\mathcal{I}$.

        \item\label{lem:sem-interactions:2} If $X^\mathcal{I}\subseteq  p_l(\mathfrak{R}^\mathcal{C}_{\exists P.C})^\mathcal{I}, X^\mathcal{I} \subseteq   \forall R.D^\mathcal{I}$, $P^\mathcal{I}\subseteq  R^\mathcal{I}$, then there exists a fresh role  $\mathfrak{R}^{\mathcal{C}\cup \{\forall R.D\}}_{(\exists P.C)}$ such that         $(\mathfrak{R}^{\mathcal{C}\cup \{\forall R.D\}}_{(\exists P.C)})^\mathcal{I} \subseteq (\mathfrak{R}^{\mathcal{C}}_{(\exists P.C)})^\mathcal{I}$, $X^\mathcal{I}=  p_l(\mathfrak{R}^{\mathcal{C}\cup \{\forall R.D\}}_{(\exists P.C)})^\mathcal{I}$, and  $p_r(\mathfrak{R}^{\mathcal{C}\cup \{\forall R.D\}}_{(\exists P.C)})^\mathcal{I}  \subseteq D^\mathcal{I}$.

        \item\label{lem:sem-interactions:22} If $\{(a,b)\} \in   P^\mathcal{I}, P^\mathcal{I} \subseteq R^\mathcal{I}$, and $\{a\}\in \forall R.D^\mathcal{I}$, then $b \in D^\mathcal{I}$.

        \item\label{lem:sem-interactions:3} If $X^\mathcal{I} \subseteq   p_l(\mathfrak{R}^\mathcal{C}_{\exists P.C })^\mathcal{I}, X \subseteq   \forall R.D^\mathcal{I}$, $P^\mathcal{I}\subseteq S^\mathcal{I}$, $S\circ S^\mathcal{I}\subseteq S^\mathcal{I}$, $S^\mathcal{I}\subseteq R^\mathcal{I}$, then there exists a fresh role $\mathfrak{R}^{\mathcal{C}\cup \{\forall R.D\}}_{(\exists P.C)}$ such that   $(\mathfrak{R}^{\mathcal{C}\cup \{\forall R.D\}}_{(\exists P.C)})^\mathcal{I} \subseteq (\mathfrak{R}^{\mathcal{C}}_{(\exists P.C)})^\mathcal{I}$, $X^\mathcal{I} = p_l(\mathfrak{R}^{\mathcal{C}\cup \{\forall R.D\}}_{(\exists P.C)})^\mathcal{I}$, and        
        $p_r(\mathfrak{R}^{\mathcal{C}\cup \{\forall R.D\}}_{(\exists P.C)})^\mathcal{I}  \subseteq \forall R.D^\mathcal{I}$.

        \item\label{lem:sem-interactions:32} If $\{(a,b)\} \in   P^\mathcal{I}, P^\mathcal{I} \subseteq S^\mathcal{I} \subseteq R^\mathcal{I}$, $(S\circ S) ^\mathcal{I} \subseteq S^\mathcal{I}$, and  $\{a\}\in \forall R.D^\mathcal{I}$, then $b \in \forall S.D^\mathcal{I}$.
    \end{enumerate}
\end{lemma}
\begin{proof}  Items~\ref{lem:sem-interactions:0}, 
\ref{lem:sem-interactions:1},  
\ref{lem:sem-interactions:22} and \ref{lem:sem-interactions:32} are trivial.  To prove Item~\ref{lem:sem-interactions:2}, we define $\mathfrak{R}^{\mathcal{C}\cup \{\forall R.D\}}_{(\exists P.C)}$ such that  $(\mathfrak{R}^{\mathcal{C}\cup \{\forall R.D\}}_{(\exists P.C)})^\mathcal{I}\subseteq  (\mathfrak{R}^{\mathcal{C}}_{(\exists P.C)})^\mathcal{I}$ and $p_l(\mathfrak{R}^{\mathcal{C}\cup \{\forall R.D\}}_{(\exists P.C)})^\mathcal{I}=X^\mathcal{I}$. 
We have to show $p_r(\mathfrak{R}^{\mathcal{C}\cup \{\forall R.D\}}_{(\exists P.C)})^\mathcal{I}  \subseteq D^\mathcal{I}$. Let $y\in  p_r(\mathfrak{R}^{\mathcal{C}\cup \{\forall R.D\}}_{(\exists P.C)})^\mathcal{I}$. By definition, there is some $(x,y)\in  (\mathfrak{R}^{\mathcal{C}\cup \{\forall R.D\}}_{(\exists P.C)})^\mathcal{I}\subseteq (\mathfrak{R}^{\mathcal{C}}_{(\exists P.C)})^\mathcal{I} \subseteq P^\mathcal{I} \subseteq R^\mathcal{I}$
 and $x\in X^\mathcal{I}$. Due to $X^\mathcal{I}\subseteq \forall R.D^\mathcal{I}$, we have $y\in D^\mathcal{I}$.   To prove Item~\ref{lem:sem-interactions:3},  we define $\mathfrak{R}^{\mathcal{C}\cup \{\forall R.D\}}_{(\exists P.C)}$ with the same conditions. We have to show $p_r(\mathfrak{R}^{\mathcal{C}\cup \{\forall R.D\}}_{(\exists P.C)})^\mathcal{I}  \subseteq \forall R. D^\mathcal{I}$ with the hypothesis $P^\mathcal{I}\subseteq S^\mathcal{I}$, $S\circ S^\mathcal{I}\subseteq S^\mathcal{I}$, $S^\mathcal{I}\subseteq R^\mathcal{I}$.  Let $y\in  p_r(\mathfrak{R}^{\mathcal{C}\cup \{\forall R.D\}}_{(\exists P.C)})^\mathcal{I}$. Using the same argument above, we obtain $x\in X^\mathcal{I}$, $(x,y)\in P^\mathcal{I} \subseteq S^\mathcal{I} \subseteq R^\mathcal{I}$. Let $(y,z)\in S^\mathcal{I}$. This implies that $(x,z)\in S^\mathcal{I} \subseteq R^\mathcal{I}$. Due to $x\in \forall R.D^\mathcal{I}$, we have $z\in D^\mathcal{I}$.
\end{proof}

\begin{table}[H]
  \centering
  \begin{tabular}{|l|l|l|}
  \hline
   Constructor    &  Set-theoretical semantics & Categorical semantics\\
   \hline 
   $X\sqsubseteq Y$  & $X^\mathcal{I}\subseteq Y^\mathcal{I}$ & $X$ is more specific  than $Y$ (\ref{org-axioms}) \\ 
   \hline
    $C \sqcap D$  & intersection of $C^\mathcal{I}$ and $D^\mathcal{I}$ & the greatest object smaller than $C$ and $D$ (\ref{org-conj-d}, \ref{org-conj-c})\\
  \hline  
   $C \sqcup D$  & union of $C^\mathcal{I}$ and $D^\mathcal{I}$ & the smallest object greater than  $C$ and $D$ (\ref{org-disj-d}, \ref{org-disj-c})\\
   \hline
   $\neg C$  & complement of  $C^\mathcal{I}$  &  \makecell[l]{$C\sqcap \neg C$ is the  smallest object, and \\  $C\sqcup \neg C$ is the greatest object (\ref{org-neg-sqcap}, \ref{org-neg-sqcup})}      \\
   \hline
   $\exists R.C$  & \makecell[l]{elements having a relation \\ $R$  with an element of  $C^\mathcal{I}$}  &   \makecell[l]{left projection of the most general object $R'$ but \\
   more specific than $R$ such that  right projection of \\ $R'$ is more specific than $C$ (\ref{org-exists-filler}, \ref{org-exists})} \\
   \hline
    $\forall R.C$  & \makecell[l]{elements having a relation  \\ $R$ only with  elements of  $C^\mathcal{I}$}  &  isomorphic to $\neg (\exists  R.\neg C)$ (\ref{org-forall-weakened}) \\
   \hline
  \end{tabular}\caption{Correspondences of the two semantics}
  \label{tab:setsem-catsem}
\end{table}
  
We have so far introduced a set of necessary categorical properties to characterize  the semantics of all $\mathcal{SH}$ constructors. All Lemmas~\ref{lem:sem-conj}-\ref{lem:sem-interactions} ensure that if an arrow $X\rightarrow Y$  is \emph{directly derived}  from the categorical properties \ref{org-syntax-role}-\ref{org-exists-forall-comp} (i.e it occurs in the conclusion of a property), then $X^\mathcal{I}\subseteq Y^\mathcal{I}$ holds for all interpretation $\mathcal{I}$ under the set-theoretical semantics. In the sequel, we use   $\psh$ to denote all categorical properties, namely   \ref{org-syntax-role}, \ref{org-syntax-concept},\ref{org-syntax-functor} (for general concept and role categories), \ref{org-axioms} (for ontology axioms), 
\ref{org-conj-d}, \ref{org-conj-c} (for conjunction), \ref{org-disj-d}, \ref{org-disj-c} (for disjunction), \ref{org-neg-sqcap}, \ref{org-neg-sqcup},  \ref{org-neg-bot}, \ref{org-neg-top} (for negation),  \ref{org-exists-filler}, \ref{org-exists} (for existential restriction), \ref{org-forall-weakened} (for universal restriction), \ref{org-dist} (for the interaction between conjunction and disjunction),  \ref{org-exists-comp} (for the interaction between  existential restriction and transitive role), 
\ref{org-exists-forall} (for the interaction between  existential and universal restrictions), \ref{org-exists-forall-ind} (for the interaction between  existential, universal restrictions and individuals), \ref{org-exists-forall-comp} (for the interaction between  existential and universal restrictions and transitive roles) and \ref{org-exists-forall-comp-ind} (for the interaction between  existential and universal restrictions, transitive roles and individuals).
 The following lemmas provide   properties that are \emph{indirectly derived} from  $\psh$.

\begin{lemma}\label{lem:commu-asso-custom}
 The following properties hold in a concept  category $\mathscr{C}_c$.
\begin{align}
    &C\sqcup D\leftrightarrows D\sqcup C \label{com01}\\
    &C\sqcap D\leftrightarrows D\sqcap C \label{com02}\\
    &(C\sqcup D)\sqcup E\leftrightarrows C\sqcup(D\sqcup E) \label{asso01}\\
    &(C\sqcap D)\sqcap E\leftrightarrows C\sqcap(D\sqcap E) \label{asso02}\\
    &C\sqcap C\leftrightarrows C \label{idcap}\\
    &C\sqcup C\leftrightarrows C \label{idcup}
\end{align}

\noindent In other words, $\sqcup$ and $\sqcap$ are commutative, associative and idempotent.
\end{lemma}
\begin{proof} The lemma can  be straightforwardly showed using  \ref{org-syntax-role} (composition of arrows), \ref{org-conj-c}, \ref{org-conj-d}, \ref{org-disj-c}, \ref{org-disj-d} if all objects involved in the arrows belong to the category. 
\end{proof}

\begin{lemma}\label{lem:de-morgan}
The following properties hold in $\mathscr{C}_c\tuple{\mathcal{O}}$. 
\begin{align}
&C \leftrightarrows  \neg \neg C \label{neg-double}\\ 
&C\rightarrow \neg D \text{ iff }  D\rightarrow \neg C \label{neg-dual}\\
&C \sqcap D \rightarrow \bot \text{ iff } C\rightarrow \neg D \text{ (or } D\rightarrow \neg C)\label{neg-disjoint}\\
&\neg (C \sqcap D) \leftrightarrows   \neg C \sqcup  \neg D \label{neg-conj}\\
&\neg (C \sqcup D) \leftrightarrows   \neg C \sqcap  \neg D \label{neg-disj}\\
&\top\rightarrow\neg C\sqcup D\text{ iff } C\rightarrow D
\label{top-sub}\\
&D\rightarrow E \text{ implies }C\sqcup D\rightarrow C\sqcup E,  
    C\sqcap D\rightarrow C\sqcap E \label{trans}
\end{align}
\end{lemma}

\begin{proof}
\begin{enumerate}[wide, labelwidth=!, labelindent=0pt]
\item By  \ref{org-neg-sqcap} and \ref{org-neg-bot}  where $C$ gets $\neg C$ and $X$ gets $C$, we have $C\rightarrow \neg \neg C$. Analogously,  by \ref{org-neg-sqcap} and \ref{org-neg-top}, we obtain  $\neg \neg C \rightarrow C$.  Hence, double negation is proved in $\mathscr{C}_c\tuple{\mathcal{O}}$.

\item  Assume  $C \rightarrow \neg D$ is in the category.  By \ref{org-conj-c}  and \ref{org-neg-sqcap}, we obtain $C\sqcap D \rightarrow D \sqcap \neg D\rightarrow \bot$. By $\mathsf{P_\neg^\bot}$, we have $D \rightarrow \neg C$. In the same way, we can show $C \rightarrow \neg D$ from $D \rightarrow \neg C$.  Hence, (\ref{neg-dual}) is proved. 

\item To prove (\ref{neg-disjoint}), we need $C\sqcap D\rightarrow \bot$ which follows from $C\sqcap D \rightarrow C\rightarrow \neg D$ and $C\sqcap D \rightarrow D$, and thus   $C\sqcap D \rightarrow D\sqcap \neg D\rightarrow \bot$. Conversely, if $D\rightarrow \neg C$, then we have ${C\sqcap D\rightarrow D\rightarrow \neg C}$ and ${C\sqcap D\rightarrow C}$ thus ${C\sqcap D\rightarrow\bot}$. Hence, (\ref{neg-disjoint}) is proved.

\item  We have $C\sqcap D\rightarrow C$ and $C\sqcap D\rightarrow D$. Due to (\ref{neg-dual}), we obtain $\neg C \rightarrow \neg(C\sqcap D)$ and $\neg D \rightarrow \neg(C\sqcap D)$. By \ref{org-disj-c}, we obtain $\neg C \sqcup \neg D \rightarrow \neg(C \sqcap D)$. To prove the inverse, we take arrows $\neg C \rightarrow \neg C  \sqcup  \neg D$ and $\neg D \rightarrow \neg C  \sqcup  \neg D$ from \ref{org-disj-d}. Due to double negation and (\ref{neg-dual}), it follows that $\neg (\neg C  \sqcup  \neg D) \rightarrow C$ and $\neg (\neg C  \sqcup  \neg D) \rightarrow D$. By \ref{org-conj-c}, we obtain $\neg(\neg C \sqcup \neg D) \rightarrow  C \sqcap D$. Due to (\ref{neg-dual}) and double negation, it follows that $\neg(C \sqcap D) \rightarrow  \neg C \sqcup \neg D$. Hence, (\ref{neg-conj})  is proved.

\item Analogously, we can prove (\ref{neg-disj}) by starting with arrows  $\neg C  \sqcap  \neg D  \rightarrow \neg C$ and $\neg C  \sqcap  \neg D \rightarrow  \neg D$ from  \ref{org-conj-d}. Due to (\ref{neg-dual}), we have $C\rightarrow \neg(\neg C  \sqcap  \neg D)$ and $D \rightarrow  \neg(\neg C  \sqcap  \neg D)$. By  \ref{org-disj-c}, we have $C \sqcup D \rightarrow \neg(\neg C  \sqcap  \neg D)$, and by double negation  and (\ref{neg-dual}) we obtain  $(\neg C  \sqcap  \neg D) \rightarrow \neg(C \sqcup D)$. To prove the inverse, we take arrows $C \rightarrow C  \sqcup  D$ and $D \rightarrow C  \sqcup  D$ obtained from \ref{org-disj-c}. Due to  (\ref{neg-dual}), we have $\neg(C  \sqcup  D) \rightarrow \neg C$ and $\neg(C  \sqcup  D) \rightarrow  \neg D$. By \ref{org-conj-c}, we have $\neg(C  \sqcup  D) \rightarrow \neg C \sqcap \neg D$.  

\item To prove (\ref{top-sub}), assume that $C\rightarrow D$. 
 Double negation implies   ${C\rightarrow\neg\neg D}$. From 
(\ref{neg-disjoint}), the previous arrow is equivalent to $C\sqcap\neg 
D\rightarrow\bot$, which in turn is equivalent to $\top\rightarrow\neg 
C\sqcup D$ from Properties~(\ref{neg-dual}), double negation, (\ref{trans}) 
    and (\ref{neg-conj}). For the other direction, we apply \ref{org-neg-top},
    which implies that $\neg\neg C\rightarrow D$ holds, then we apply double negation and
    arrow composition to obtain $C\rightarrow D$.

\item By hypothesis and \ref{org-disj-d}, we have  $D\rightarrow E\rightarrow C\sqcup E$ and $C\rightarrow C\sqcup E$. By \ref{org-disj-c}, we obtain $C \sqcup D\rightarrow C\sqcup E$. In the same way, by hypothesis and \ref{org-conj-d}, we have  $C\sqcup D \rightarrow D\rightarrow E$ and $C\sqcup D \rightarrow C$. By \ref{org-conj-c}, we obtain $C \sqcap D\rightarrow C\sqcap E$.
\end{enumerate}
\end{proof}

\begin{lemma}\label{lem:forall-prop} Assume that all objects involved in the following arrows are present in $\mathscr{C}_c\tuple{\mathcal{O}}$ and $\mathscr{C}_r\tuple{\mathcal{O}}$. The following properties hold. 
\begin{align}
& C \rightarrow D \text{ implies } \exists R.C  \rightarrow \exists R.D \label{exists-sub-filler}
\\
&C \rightarrow D \text{ implies }\forall R.C \longrightarrow  \forall R.D \label{forall-sub-filler}\\ 
&\exists R.C \sqcap \forall R.D \rightarrow  \exists R.(C \sqcap  D)  \label{forall-exists-inter} 
\end{align}
\end{lemma}

\begin{proof} 

    We show (\ref{exists-sub-filler}). We consider two objects $\mathfrak{R}^\varnothing_{(\exists R.C)}$ and $\mathfrak{R}^\varnothing_{(\exists R.D)}$ in $\mathscr{C}_r$. By \ref{org-exists-filler} and the hypothesis, we have $\mathfrak{R}^\varnothing_{(\exists R.C)}\rightarrow R$ and  $p_r(\mathfrak{R}^\varnothing_{(\exists R.C)})\rightarrow C \rightarrow D$. By \ref{org-exists}, we obtain $p_l(\mathfrak{R}^\varnothing_{(\exists R.C)})\rightarrow p_l(\mathfrak{R}^\varnothing_{(\exists R.D)})$.
        We show (\ref{forall-sub-filler}).  Due to (\ref{neg-dual}), $C\rightarrow D$ implies $\neg D \rightarrow \neg C$. By (\ref{exists-sub-filler}), we obtain $\exists R.\neg D\rightarrow \exists R.\neg C$. Again, due to (\ref{neg-dual}), we have $\neg \exists R.\neg C\rightarrow \neg \exists R.\neg D$. From \ref{org-forall-weakened}, it follows $\forall R.C \rightarrow \forall R.D$.
    
    We show (\ref{forall-exists-inter}). We have $X\rightarrow \exists R.C$ and $X\rightarrow \forall R.D$ with $X=\exists R.C\sqcap \forall R.D$ due to \ref{org-conj-d}. By \ref{org-exists-forall} and \ref{org-syntax-functor}, there are arrows $\mathfrak{R}^{\{\forall R.D\}}_{(\exists R.C)}\rightarrow \mathfrak{R}^{\varnothing}_{(\exists R.C)}\rightarrow R$, $X\rightarrow p_l(\mathfrak{R}^{\{\forall R.D\}}_{(\exists R.C)})$,
    $ p_r(\mathfrak{R}^{\{\forall R.D\}}_{(\exists R.C)})\rightarrow p_r(\mathfrak{R}^{\varnothing}_{(\exists R.C)}) \rightarrow C$ and $ p_r(\mathfrak{R}^{\{\forall R.D\}}_{(\exists R.C)})\rightarrow D$. Therefore,  $ p_r(\mathfrak{R}^{\{\forall R.D\}}_{(\exists R.C)})\rightarrow C\sqcap D$ due to \ref{org-conj-c}, and thus   $X\rightarrow  p_l(\mathfrak{R}^{\{\forall R.D\}}_{(\exists R.C)})\rightarrow \exists R.(C\sqcap D)$ due to \ref{org-exists}.

\end{proof}

We can observe that the categorical properties in $\psh$ do not prevent concept and role categories from admitting an arbitrary arrow $X\rightarrow Y$ that is not necessarily provable from  $\psh$, even if $\mathcal{O}\not\models X\sqsubseteq Y$. To define categorical inconsistency of an ontology, we need a particular  concept category such that it  contains only necessary arrows satisfying  the properties in $\psh$. The following theorem shows that such a concept category exists.

\begin{theorem}\label{thm:cat-minimality} Let $\mathcal{O}$ 
be an $\mathcal{SH}$ ontology. There exists a unique minimal (with respect to set inclusion)  concept   category of $\mathcal{O}$ satisfying  all categorical properties in $\psh$. 
\end{theorem}
\begin{proof}  
Since there is an identity arrow  $X\rightarrow X\in \mathsf{Hom}(\mathscr{C})$ for all object $X \in \mathsf{Ob}(\mathscr{C})$, and $X\rightarrow Y\in \mathsf{Hom}(\mathscr{C})$   implies  $X,Y\in \mathsf{Ob}(\mathscr{C})$ for every concept or role category $\mathscr{C}$, it holds that  minimality of $ \mathsf{Hom}(\mathscr{C})$ implies minimality of $\mathsf{Ob}(\mathscr{C})$. Hence, $\mathsf{Hom}(\mathscr{C})$ determines $\mathscr{C}$. We use $\Gamma_c\tuple{\mathcal{O}}$ to denote the set of all concept categories of $\mathcal{O}$  satisfying  all properties in $\psh$. By definition, each $\mathscr{C}'_c\in \Gamma_c\tuple{\mathcal{O}}$ has a corresponding  role category, denoted $\mathscr{C}'_{c,r}$,  involved in properties from $\psh$.
Let $\mathsf{H}_{c}(\mathscr{C}_m)=\displaystyle{\bigcap_{\mathscr{C}'_c\in \Gamma_c\tuple{\mathcal{O}}}\mathsf{Hom}(\mathscr{C}'_c)}$, and  $\mathsf{O}_{c}\tuple{\mathscr{C}_m}$ is the set of all objects involved in all arrows in $\mathsf{H}_{c}(\mathscr{C}_m)$.
We need to prove that     $\mathsf{H}_c(\mathscr{C}_m)$  is the set of arrows of a concept   category $\mathscr{C}_m$ for $\mathcal{O}$ such that  all properties in $\psh$ are satisfied in $\mathscr{C}_m$.  Given an ontology $\mathcal{O}$, we use  $\mathbf{C}_0$, $\mathbf{R}_0$ and $\mathbf{I}_0$  to denote respectively the sets of concept names, role names and individuals occurring in  $\mathcal{O}$. First, we determine a signature $\mathcal{S}=\tuple{\mathbf{C}, \mathbf{R}, \mathbf{I}}$ for $\mathscr{C}_m$ from $\mathcal{O}$ as follows: (i) the sets $\mathbf{C}, \mathbf{R}, \mathbf{I}$ are the smallest sets such that $\mathbf{C}_0\subseteq \mathbf{C}$, $\mathbf{R}_0\subseteq \mathbf{R}$ and $\mathbf{I}_0\subseteq \mathbf{I}$, and (ii) for each concept of the form $\exists R.C$ occurring in $\mathcal{O}$, and each subset $\mathcal{C}$ of the set of all concepts of the form $\forall R.D$ occurring in $\mathcal{O}$, $\mathbf{R}$ contains a role name $\mathfrak{R}^\mathcal{C}_{(\exists R.C)}$, and  $\mathbf{C}$ contains  two concept names $p_l(\mathfrak{R}^\mathcal{C}_{(\exists R.C)})$ and $p_r(\mathfrak{R}^\mathcal{C}_{(\exists R.C)})$. We have to  show that $\mathsf{H}_c(\mathscr{C}_m)$  satisfies all categorical properties in $\psh$. 

\begin{itemize}[wide, labelwidth=!, labelindent=0pt]
    \item Each property $\mathsf{P}\in \psh$ can be written as  implications each of which is composed of a premise and conclusion. For instance,  
    \begin{itemize}[wide, labelwidth=!, labelindent=0pt]
    \item \ref{org-syntax-role} : If $R\in\mathbf{R}$ or   $R=\{(a,b)\}$ with $(a,b)\in \mathbf{I}\times \mathbf{I}$, then $R\rightarrow R, R\rightarrow \role{\top}, \role{\bot} \rightarrow R\in  \mathsf{Hom}(\mathscr{C}_r)$.  Furthermore, if  $P\rightarrow  R\rightarrow S \in  \mathsf{Hom}(\mathscr{C}_r)$, then $P\rightarrow  S \in  \mathsf{Hom}(\mathscr{C}_r)$.   
    
    \item \ref{org-syntax-concept}: If $C\in\mathbf{C}$ or $C=\{a\}$  with $a\in \mathbf{I}$, then $C\rightarrow C, C\rightarrow \top, \bot \rightarrow R\in  \mathsf{Hom}(\mathscr{C}_c)$. Furthermore, if  $C\rightarrow  D\rightarrow E \in  \mathsf{Hom}(\mathscr{C}_c)$, then $C\rightarrow  E \in  \mathsf{Hom}(\mathscr{C}_c)$.  
     
    \item \ref{org-syntax-functor}: If $R\rightarrow S \in \mathsf{Hom}(\mathscr{C}_r)$, then $p_l(R)\rightarrow p_l(S), p_r(R)\rightarrow p_r(S)  \in \mathsf{Hom}(\mathscr{C}_r)$, $p_l(\role{\bot})\leftrightarrows p_r(\role{\bot}) \leftrightarrows \bot, p_l(\role{\top})\leftrightarrows p_r(\role{\top}) \leftrightarrows \top\in \mathsf{Hom}(\mathscr{C}_c)$. Furthermore,   $p_l(R)\rightarrow \bot \in \mathsf{Hom}(\mathscr{C}_c)$ iff $R\rightarrow \mathfrak{R}_\bot \in \mathsf{Hom}(\mathscr{C}_r)$, and  
      $p_r(R)\rightarrow\bot \in \mathsf{Hom}(\mathscr{C}_c)$ iff $R\rightarrow \mathfrak{R}_\bot \in \mathsf{Hom}(\mathscr{C}_r)$.  
      
    \end{itemize}
For   $\mathsf{P}$ that is different from \ref{org-syntax-role}, \ref{org-syntax-concept},  \ref{org-syntax-functor}, $\mathsf{P}$ is already   written as an implication.  Note that if a condition like ``$C\in \mathsf{Ob}(\mathscr{C}_c)$" occurs in  some property $\mathsf{P}$, then it can be rewritten as an arrow $C\rightarrow C\in \mathsf{Hom}(\mathscr{C}_c)$.


\item 
We show that if the premise of $\mathsf{P}$ holds in $\mathsf{H}_c(\mathscr{C}_m)$, then the arrows in the conclusion of $\mathsf{P}$ also holds in $\mathsf{H}_c(\mathscr{C}_m)$. For instance, 
    assume that the premise of \ref{org-exists-filler} holds in $\mathsf{H}_c(\mathscr{C}_m)$, i.e $\exists R.C\in \mathsf{O}_c\tuple{\mathcal{O}}$ and $\exists R.C\rightarrow \exists R.C\in \mathsf{H}_c(\mathscr{C}_m)$.  Since $\mathsf{H}_c(\mathscr{C}_m) \subseteq \mathsf{Hom}(\mathscr{C}'_c)$  for all $\mathscr{C}'_c\in \Gamma_c\langle \mathcal{O}\rangle$, we have   $\exists R.C\in \mathsf{Ob}(\mathscr{C}'_c)$ for all $\mathscr{C}'_c\in \Gamma_c\langle \mathcal{O}\rangle$. This implies that the premise of \ref{org-exists-filler} holds in all $\mathscr{C}'_c\in \Gamma_c\langle \mathcal{O}\rangle$, and thus the arrows in the conclusion of \ref{org-exists-filler} are included in $\mathsf{Hom}(\mathscr{C}'_c)$ or $\mathsf{Hom}(\mathscr{C}'_{c,r})$  for all $\mathscr{C}'_c \in \Gamma_c\langle \mathcal{O}\rangle$. Therefore, the arrows in the conclusion of \ref{org-exists-filler} are also included in  $\mathsf{H}_{c}(\mathscr{C}_m)$. We can use the same argument for every other property $\mathsf{P}$ since if the premise of  $\mathsf{P}$ holds in $\mathsf{H}_c(\mathscr{C}_m)$, then it also holds in  all $\mathscr{C}'_c\in \Gamma_c\langle \mathcal{O}\rangle$ due to the definition of $\mathsf{H}_c(\mathscr{C}_m)$.
    \end{itemize}
Therefore,  we have proved that      $\mathsf{H}_{c}(\mathscr{C}_m)$  is the set of arrows of a concept   category $\mathscr{C}_m$ for $\mathcal{O}$. To prove minimality of  $\mathscr{C}_m$, assume that there is some concept category  $\mathscr{C}'_m$ such that  all properties in $\psh$ are satisfied in $\mathscr{C}_m$ and $\mathsf{H}_{c}(\mathscr{C}'_m)\subset \mathsf{H}_{c}(\mathscr{C}_m)$. In this case, we have  $\mathscr{C}'_m\in \Gamma_c\tuple{\mathcal{O}}$, and thus $\mathsf{H}_{c}(\mathscr{C}_m)\neq \displaystyle{\bigcap_{\mathscr{C}'_c\in \Gamma_c\tuple{\mathcal{O}}}\mathsf{Hom}(\mathscr{C}'_c)}$, which contradicts the definition of $\mathsf{H}_{c}(\mathscr{C}_m)$.
Uniqueness of   $\mathsf{H}_{c}(\mathscr{C}_m)$ follows from its definition. 
\end{proof}

\begin{corollary}\label{cor:arrow-subsumption} Let $\mathcal{O}$ 
be an $\mathcal{SH}$ ontology   and  $\mathscr{C}_c\tuple{\mathcal{O}}$ be the minimal concept  category of $\mathcal{O}$ satisfying  all categorical properties in $\psh$. It holds that $X\rightarrow Y\in \mathsf{Hom}(\mathscr{C}_c\tuple{\mathcal{O}})$ iff there is some $\mathsf{P}\in\psh$ such that 
 $X\rightarrow Y$ occurs in the conclusion of  $\mathsf{P}$ and the premise of $\mathsf{P}$ holds in $\mathscr{C}_c\tuple{\mathcal{O}}$.
\end{corollary}

\begin{proof} For the ``$\Longleftarrow$" direction, it follows immediately from satisfaction of all properties in $\psh$.

\noindent For the ``$\Longrightarrow$" direction, assume that $X\rightarrow Y\in \mathsf{Hom}(\mathscr{C}_c\tuple{\mathcal{O}})$. By contradiction, assume that $X\rightarrow Y$ never occurs in the conclusion  of any application of any property $\mathsf{P}\in \psh$.  This implies that  $\mathscr{C}'_c\tuple{\mathcal{O}}$  with  $\mathsf{Hom}(\mathscr{C}'_c\tuple{\mathcal{O}})=\mathsf{Hom}(\mathscr{C}_c\tuple{\mathcal{O}})\setminus \{X\rightarrow Y\}$    satisfies also  all properties in $\psh$, which contradicts minimality of $\mathscr{C}_c\tuple{\mathcal{O}}$. Note that  \ref{org-syntax-role} and \ref{org-syntax-concept} ensure the presence of identity arrows, outgoing arrows from $\bot$ and ingoing arrows to $\top$ while  \ref{org-syntax-functor} ensures the presence of arrows  related to the functors $p_l, p_r$.   
\end{proof}

The following theorem is a consequence of Lemmas~\ref{lem:sem-conj}-\ref{lem:sem-interactions} and Corollary~\ref{cor:arrow-subsumption}. 

\begin{theorem}\label{thm:cat-set}
    Let $\mathcal{O}$ 
be an $\mathcal{SH}$ ontology   and  $\mathscr{C}_c\tuple{\mathcal{O}}$ be the minimal concept  category of $\mathcal{O}$ satisfying  all categorical properties  $\psh$. It holds that $X\rightarrow Y\in \mathsf{Hom}(\mathscr{C}_c\tuple{\mathcal{O}})$ implies  $\mathcal{O}\models X\sqsubseteq   Y$.    
\end{theorem}

Theorem~\ref{thm:cat-minimality} with Corollary~\ref{cor:arrow-subsumption}  allows us to introduce the notion of categorical  inconsistency of an $\mathcal{SH}$ ontology.

\begin{definition}[Categorical consistency]\label{def:SH-consistency} Let $\mathcal{O}$ 
be an $\mathcal{SH}$ ontology  where $\mathbf{C}$, $\mathbf{R}$  and $\mathbf{I}$ are respectively non-empty and disjoint sets of concept names, role names and  individuals  defined from  $\mathcal{O}$. We use $\mathscr{C}_c\tuple{\mathcal{O}}$ and  $\mathscr{C}_{r}\tuple{\mathcal{O}}$ to denote the minimal concept and role categories of $\mathcal{O}$ satisfying  all categorical properties in $\psh$.  
We say that $\mathcal{O}$ is categorically inconsistent if 
$\{a\}\rightarrow \bot\in \mathsf{Hom}(\mathscr{C}_c\tuple{\mathcal{O}})$ for some $a\in \mathbf{I}$.  
\end{definition}

We can observe that \ref{org-conj-c}, \ref{org-disj-c} and the properties related to negation such as \ref{org-neg-sqcap} may lead to forming  an infinite number of objects of the forms  $C\sqcap C\sqcap \cdots  $, $C\sqcup C  \sqcup \cdots$ and $\neg \cdots \neg  C$. However, $C\sqcap C$, $C\sqcup C$ and $\neg \neg C$  are isomorphic  to $C$ due to  idempotence of conjunction and disjunction, and double negation   by Lemmas~\ref{lem:commu-asso-custom} and ~\ref{lem:de-morgan}. Moreover, 
$C\leftrightarrows  \mathsf{NNF}(C)$ holds in  the categorical  semantics due to  \ref{org-forall-weakened} and Lemma~\ref{lem:de-morgan}. 
Hence,  $C$ is isomorphic to  $\mathsf{NNF}(C)$  in concept categories, and thus the minimal concept category.



\section{Semantic Equivalence}\label{sec:sem-equiv}

In this section, we show that the categorical semantics of $\mathcal{SH}$ described by $\psh$ is equivalent to its set-theoretical semantics, that means, an $\mathcal{SH}$ ontology is categorically inconsistent iff it is set-theoretically inconsistent.

\begin{theorem}[semantic equivalence]\label{thm:sem-equiv}
Let $\mathcal{O}$ be an $\mathcal{SH}$ ontology. It holds that $\mathcal{O}$ is set-theoretically inconsistent iff $\mathcal{O}$ is categorically inconsistent.
\end{theorem}

The ``category-set" direction  of Theorem~\ref{thm:sem-equiv} follows from Theorem~\ref{thm:cat-set}.
To prove the direction ``set-category" of Theorem~\ref{thm:sem-equiv}, we need to introduce further construction: a universal tableau  for describing all models of an $\mathcal{SH}$ ontology under the set-theoretical semantics. In the literature, tableau method  has been well studied for DLs from $\mathcal{ALC}$ to  very expressive DLs with complex roles such as   $\mathcal{SROIQ}$ \cite{hor2000}, \cite{HoKS06a}. We adapt these results to propose a universal tableau for $\mathcal{SH}$ that has some differences from those in the literature.

\begin{definition}[Universal Tableau for $\mathcal{SH}$]\label{def:uni-tableau} Let $\mathcal{O}$ be  an $\mathcal{SH}$ ontology   with  signature $\langle \mathbf{C}, \mathbf{R},\mathbf{I}\rangle$.
A universal tableau for $\mathcal{O}$ is a graph  $T=\langle V_0,\cdots,V_n, E,   L\rangle$ for  $\mathcal{O}$ with $n$ minimal where $V_i$ with $i\in \{0..n\}$ are   the smallest (with respect to set inclusion)  disjoint sets of nodes, $E$ a set of edges, and $L$  a labelling function  that associates the smallest set of concepts in NNF to each $x\in V_i$  such that  the following properties are satisfied.

\begin{enumerate}[leftmargin=!, labelwidth=1em, align=left]
\item[\mylabel{def:uni-tableau:ind}{$\mathbf{T_I}$}$:$] For each individual $a\in \mathbf{I}$, there is a node $x^a\in V_0$. Furthermore, if $\{a\}\sqsubseteq C\in \mathcal{O}$, then $C\in L(x^a)$; and if $\{(a,b)\}\sqsubseteq R\in \mathcal{O}$ and $R \overlay{\sqsubseteq}{\ast} S\in \mathcal{R}^+\tuple{\mathcal{O}}$, then $x^a$ has an  $S$-candidate root successor  $x^b$ (or $x^a$ is an $S$-candidate root predecessor of $x^b$, or there is an $S$-edge  $(x^a,y^b)\in E$). 

\item[\mylabel{def:uni-tableau:ax}{$\mathbf{T_{\sqsubseteq}}$}$:$] For each $x\in V_i$ with $i\in\{0,n\}$,  if  $C\sqsubseteq D\in \mathcal{O}$, then $\mathsf{NNF}(\neg C)\sqcup D\in  L(x)$. 

\item[\mylabel{def:uni-tableau:conj}{$\mathbf{T_{\sqcap}}$}$:$] For each $x\in V_i$ with  $i\in\{0,n\}$,   if $C\sqcap D\in L(x)$, then $\{C,D\}\subseteq  L(x)$.

\item[\mylabel{def:uni-tableau:disj}{$\mathbf{T_{\sqcup}}$}$:$] For each $x\in V_i$ with $E\sqcup F\in L(x)$ and $i\in\{0,n\}$,  there is the smallest set $\mathbb{S}(x)\subseteq V_i$  of siblings of $x$  with $x\in \mathbb{S}(x)$ such that 
\begin{enumerate}   
    \item\label{def:uni-tableau:disj:0} If $x$ is an $S$-candidate (root) successor of some node $y$, then all $x'\in \mathbb{S}(x)$ is also an $S$-candidate (root) successor of $y$ with $(y,x)\in E$.

   \item\label{def:uni-tableau:disj:1} $x'\in \mathbb{S}(x)$ and  $C\sqcup D\in L(x')$ imply $L(x')\cap\{C,D\}\neq \varnothing$;
    
    \item\label{def:uni-tableau:disj:2} 
    $x'\in \mathbb{S}(x)$, $C\sqcup D\in L(x')$, $C\neq D$ and $\{C, D\}\not\subseteq   L(x')$ imply that there is a node    $x''\in \mathbb{S}(x)$  such that $X\in L(x')$, $Y\in L(x'')$, $Y\notin L(x')$ and $X\notin L(x'')$, $L(x')\setminus \{X\}\subseteq L(x'')$ and $L(x'')\setminus \{Y\}\subseteq L(x')$ where $X,Y\in \{C,D\}$ and $X\neq Y$.
\end{enumerate}
     
\item[\mylabel{def:uni-tableau:exists}{$\mathbf{T_{\exists}}$}$:$] For each $x\in V_i$ with  $i\in\{0,n\}$, 
if $\exists R.C\in L(x)$, $R \overlay{\sqsubseteq}{\ast} S\in \mathcal{R}^+\tuple{\mathcal{O}}$, and  there is no node $x'\in V_{j}$ with $j< i$  such that $x'$ is a candidate ancestor of $x$ with $L(x)\subseteq L(x')$ (i.e. $x$ is blocked by $x'$), then  $x$ has an $S$-candidate successor $y\in V_{i+1}$ with an $S$-edge   $(x,y)\in E$ (or $x$ is  an $S$-candidate predecessor of $y$) such that 
$C\in L(y)$. 

\item[\mylabel{def:uni-tableau:forall}{$\mathbf{T_{\forall}}$}$:$] For each $x\in V_i$ with  $i\in\{0,n\}$,  if $\forall R.C\in L(x)$ and $x$ has a $P$-candidate successor $y\in V_{i+1}$ or a $P$-candidate root successor   $y\in V_{0}$ with  $P \overlay{\sqsubseteq}{\ast} R\in \mathcal{R}^+\tuple{\mathcal{O}}$,   then   $C\in L(y)$.

\item[\mylabel{def:uni-tableau:forall-trans}{$\mathbf{T^\circ_{\forall}}$}$:$] For each $x\in V_i$ with  $i\in\{0,n\}$,  if $\forall R.C\in L(x)$ and $x$ has a $P$-candidate successor $y\in V_{i+1}$ or $P$-candidate root successor $y\in V_{0}$ with  $P \overlay{\sqsubseteq}{\ast} S, S\circ S \overlay{\sqsubseteq}{\ast} S, S \overlay{\sqsubseteq}{\ast} R\in \mathcal{R}^+\tuple{\mathcal{O}}$,   then   $\forall S.C\in L(y)$.
\end{enumerate}

We say that a node $x\in V_i$ has a  direct clash if if  $\{A,\neg A\} \subseteq L(x)$ for  a concept name $A$, and $x$ has an indirect clash if each $y'\in \mathbb{S}(y)$ has a clash where $y$ is an $R$-candidate successor of $x$. A node has a clash if it has a direct or indirect clash.

\end{definition}

By definition, each $x^a\in V_0$ is the root of a tree whose nodes are located in $V_i$ with $i>0$ except for $x^a$. Hence, each node $x\in V_i$  with $i>0$ has at most one (candidate) predecessor. 
Minimality of  $V_i$ and  $L(x)$ 
ensures that no element of these sets can be removed  such that  \ref{def:uni-tableau:ind}, \ref{def:uni-tableau:ax}, \ref{def:uni-tableau:conj}, \ref{def:uni-tableau:disj}, \ref{def:uni-tableau:exists}, \ref{def:uni-tableau:forall}, \ref{def:uni-tableau:forall-trans} remain to be satisfied. Moreover,  there is a finite number of node layers $V_i$ due to the blocking condition that stipulates that if $x\in V_i$ is blocked by a node $x'\in V_j$ with $j< i$ then $x$ has no candidate successor. Note that $\mathbb{S}(x)$ specified by \ref{def:uni-tableau:disj}  represents nondeterminism resulting from disjunctions.

\begin{definition}\label{def:tree} Let $T=\langle V_0,\cdots,V_n, E,   L\rangle$  be a universal tableau of an $\mathcal{SH}$ ontology $\mathcal{O}$. Let $x\in V_i$ and $i\in \{0,n\}$ with a set $\mathbb{S}(x)$ of siblings of $x$. We define a minimal tree $\mathbb{T}(x) =\tuple{V', E',  L'}$   associated with  $\mathbb{S}(x)$ as follows.

\begin{enumerate}[leftmargin=0cm, itemindent=0.5cm]
    \item\label{claim:1:0}   $\mathbb{T}(x)$ is rooted at $v_0$ such that $C\in L'(v_0)$ iff $C\in L(x)$  is due to  all properties in Definition~\ref{def:uni-tableau} except for \ref{def:uni-tableau:disj}.

    \item\label{claim:1:1} For each node $v\in V'$, if $C\sqcap D\in L'(v)$ then $\{C,D\}\subseteq L'(v)$.
    
    \item\label{claim:1:2} For each node $v\in V'$, $v$ has two successors $v_1,v_2$ in $\mathbb{T}(x)$  with $(v,v_1), (v,v_2)\in E'$ iff  there is a disjunction $C\sqcup D\in L'(v)$ such that $C\neq D$,  $L'(v)\cap \{C,D\}=\varnothing$. In this case,  $L'(v_1)=L'(v)\cup \{C\}$ and $L'(v_2)=L'(v)\cup \{D\}$.
\end{enumerate}
\end{definition}

\begin{claim}\label{claim:tab} $x'\in \mathbb{S}(x)$ iff there is  a leaf node $v$ of $\mathbb{T}(x)$ such that  $L(x')=L'(v)$.
\end{claim}

\begin{proof} ``$\Longrightarrow$". Assume that $x'\in \mathbb{S}(x)$. By the definition of $\mathbb{T}(x)$, we have $L'(v_0)\subseteq L(x')$ where $v_0$ is the root of $\mathbb{T}(x)$. By induction on level of $\mathbb{T}(x)$, let $v\in V'$ such that $L'(v)\subseteq L(x')$. If $C\in L'(v)$ due to $C\sqcap C'\in L'(v)$ according to Item~\ref{claim:1:1} in Definition~\ref{def:tree}, then   $C\in L(x')$ due to  \ref{def:uni-tableau:conj}. Let $v',v''$ be successors of $v$. This implies that there is some $C\sqcup D\in L'(v)$ with $\{C,D\}\cap L'(v)=\varnothing$ and $C\neq D$ such that $L'(v')=L'(v)\cup \{C\}$ and $L'(v'')=L'(v)\cup \{D\}$ (or $L'(v')=L'(v)\cup \{D\}$ and $L'(v'')=L'(v)\cup \{C\}$).  By  \ref{def:uni-tableau:disj}-(c), we have $C\in L(x')$ (or $D\in L(x')$). Hence, $L'(v')\subseteq L(x')$ (or $L'(v'')\subseteq L(x')$). By induction, it holds that   $L'(v)\subseteq L(x')$ where $v$ is a leaf. By minimality of $L(x')$ we have $L(v)= L(x')$. 

\noindent ``$\Longleftarrow$". Assume that   $v$ is a leaf node of $\mathbb{T}(x)$. We prove that there is some node $x'\in \mathbb{S}(x)$ such that $L(x')=L'(v)$. By construction,   $L'(v_0)\subseteq L'(v)$ where $v_0$ is the root of  $\mathbb{T}(x)$. This implies that if $C\in L(x')$ is due to  all properties in Definition~\ref{def:uni-tableau} except for \ref{def:uni-tableau:disj} 
, then $C\in L'(v)$. Hence,   $L'(v_0)\subseteq L(x')$ for all $x'\in \mathbb{S}(x)$. By induction on level of $\mathbb{T}(x)$, let $v'$ be an ancestor of $v$ such that there is some $x'\in \mathbb{S}(x)$ with $L'(v')\subseteq L(x')$. Let $v''$ be a successor of $v'$. If $C\in L'(v'')$ due to $C\sqcap C'\in L'(v'')$ (Item~\ref{claim:1:1}, Definition~\ref{def:tree}), then $C\in L(x'')$ by \ref{def:uni-tableau:conj}. If $C\in L'(v'')$ is due to $C\sqcup D\in L'(v')$ with $C\neq D$ and $\{C,D\}\cap L(v')=\varnothing$ (Item~\ref{claim:1:2}, Definition~\ref{def:tree}), then $L'(v'')=L'(v')\cup \{C\}$ and there is some $x''\in \mathbb{S}(x)$ such that $C\in  L(x'')$, $D\in  L(x')$, and $L(x')\setminus \{D\}\subseteq L(x'')$ by \ref{def:uni-tableau:disj}-(c). This implies $L'(v'')\subseteq L(x'')$. By induction, there is some $x_0\in \mathbb{S}(x)$ such that $L'(v)\subseteq L(x_0)$. When reaching the leaf $v$, all conjunctions and disjunctions in $L'(v)$ are satisfied. By minimality of $L(x_0)$, it holds that   all conjunctions and disjunctions in $L(x_0)$ are also satisfied with respect to \ref{def:uni-tableau:conj} and \ref{def:uni-tableau:disj}, and thus $L'(v)= L(x_0)$.
\end{proof}

The following definition describes a standard tableau from a universal tableau $T$ by choosing one node $x'\in \mathbb{S}(x)$ for all node $x$ in $T$.

\begin{definition}[Tableau for $\mathcal{SH}$]\label{def:sh-tableau}

Let $\mathcal{O}$ be  an $\mathcal{SH}$ ontology, and $T=\langle V_0,\cdots,V_n, E,   L\rangle$ be a universal tableau of $\mathcal{O}$. A tableau $\mathbf{T}=\langle \mathcal{V}_0,\cdots,\mathcal{V}_n, \mathcal{E},   \mathcal{L}\rangle$  is defined from  $T$ as follows.   

\noindent$\bullet$  $\mathcal{V}_0 \subseteq  V_0$ such that $\lvert\mathcal{V}_0\cap \mathbb{S}(x^a)\rvert =1$  for each individual $a\in \mathbf{I}$. If $x^b$ is an $R$-candidate root successor of $x^a$ in $T$ and $x^b\in \mathcal{V}_0$, then $x^b$ is an $R$-root successor of $x^a$ in $\mathbf{T}$.

\noindent$\bullet$ $\mathcal{V}_i \subseteq  V_i$ for $i>0$ such that each $y\in \mathcal{V}_i$ is a candidate successor of some node $x\in \mathcal{V}_{i-1}$ in $T$ with $\lvert\mathcal{V}_i\cap \mathbb{S}(y)\rvert =1$. In this case, $y$ is called a successor of $x$ in $\mathbf{T}$. We define $(x,y)\in \mathcal{E}$.

\noindent$\bullet$ $\mathcal{L}(x)=L(x)$  for $x\in \mathcal{V}_{i}$ with $0\leqslant i\leqslant n$. We inductively define an $S$-\emph{path} $(x,y)$ over $\mathbf{T}$ as follows:  \begin{itemize}
    \item If $(x,y)$ is  an  $S$-edge and $y$ is non-blocked, then $(x,y)$ is an $S$-path.
    
    \item If $(x,y')$ is  an  $S$-edge and $y$ blocks $y'$, then $(x,y)$ is an $S$-path.
    
    \item If $(x,z)$  is an $S$-path, $(z,y)$ is an $S$-edge and $S\circ S \overlay{\sqsubseteq}{\ast} S\in \mathcal{R}^+\tuple{\mathcal{O}}$, then $(x,y)$ is an $S$-path.
    
    \item If $(x,y)$ is an $S$-path  and $S \overlay{\sqsubseteq}{\ast} R\in \mathcal{R}^+\tuple{\mathcal{O}}$, then we also say that $(x,y)$ is an $R$-path.
\end{itemize} 

A node $x\in \mathcal{V}_i$  has a direct clash if  $\{A,\neg A\} \subseteq \mathcal{L}(x)$ for  a concept name $A$. A node $x\in \mathcal{V}_i$  has an indirect clash if it has a descendant (possibly via candidate root successor) that has a direct  clash.  A node $x\in \mathcal{V}_i$  has a clash if it has  a direct or indirect clash.
A tableau $\mathbf{T}$ has a  clash
if there is an individual $a\in \mathbf{I}$ such that  $x^a\in \mathcal{V}_0$ has a clash. 
A \emph{tableau interpretation} $\mathcal{I}=\langle \Delta^{\mathcal{I}}, \cdot^{\mathcal{I}}\rangle$ is defined from $\mathbf{T}=\langle \mathcal{V}_0,\cdots,\mathcal{V}_n, \mathcal{E}, \mathcal{L}\rangle$ as follows.
    
\noindent$\bullet$ $\Delta^{\mathcal{I}} =\bigcup _{i\in \{0,n\}}\mathcal{V}_i\setminus \mathbf{B}$ where $\mathbf{B}$ is the set of blocked nodes.   
     
\noindent$\bullet$ $A^{\mathcal{I}}:= \{x\in \Delta^\mathcal{I} \mid  A \in \mathcal{L}(x)\}$ 
    
\noindent$\bullet$ $R^{\Tb{I}} :=\{\langle{x, y}\rangle\in \Delta^\mathcal{I} \times \Delta^\mathcal{I}    \mid (x,y)  \text{ is an  R-path}\}$ for each role name $R$.
\end{definition}

The tableau described by Definition~\ref{def:sh-tableau} for $\mathcal{SH}$ is slightly different from the standard one given in \cite{hor99} for $\mathcal{SHIQ}$  including $\mathcal{SH}$. For instance, the $\sqcup$-rule from standard tableau performs a choice $C'\in \{C,D\}$ when it is  applied to a concept $C\sqcup D\in L(x)$ while    Definition~\ref{def:sh-tableau} makes a choice of nodes from $\mathbf{S}(x)$ in the universal tableau.

Note that Definitions~\ref{def:uni-tableau} and \ref{def:sh-tableau} do not need to explicitly present an algorithm. They just describe  tableaux for an $\mathcal{SH}$ ontology, i.e a finite representation of a model of the ontology. In  the following, we need to show that if  an $\mathcal{SH}$ ontology has a free-clash tableau, then the ontology is set-theoretically consistent, and conversely, if  an $\mathcal{SH}$ ontology  is set-theoretically consistent, then the ontology has a free-clash tableau. These proofs are adapted from those for the standard tableau due to the difference between Definition~\ref{def:sh-tableau} and the standard tableau. Proofs of Lemmas~\ref{lem:label-concept}, \ref{lem:tab-mod} and  \ref{lem:mod-tab} can be found in Appendix.

\begin{lemma}\label{lem:label-concept} 
Let $\mathbf{T}=\langle \mathcal{V}_1,\cdots,\mathcal{V}_n, \mathcal{E},   \mathcal{L}\rangle$ be a clash-free tableau  and $\mathcal{I}=\langle \Delta^\mathcal{I}, \cdot^\mathcal{I} \rangle$  be the tableau interpretation of an  ontology $\mathcal{O}$ from $\mathbf{T}$. It holds that $C\in \mathcal{L}(x)$ implies $x\in C^\mathcal{I}$ for all $x\in \mathcal{V}_i$ for all $i$.
\end{lemma}

\begin{lemma}\label{lem:tab-mod} If  there is a clash-free tableau for an $\mathcal{SH}$ ontology  $\mathcal{O}$ then $\mathcal{O}$ is set-theoretically consistent.  
\end{lemma}

\begin{lemma}\label{lem:mod-tab} If an  $\mathcal{SH}$ ontology  $\mathcal{O}$ is set-theoretically consistent then  there is a clash-free tableau  for $\mathcal{O}$.
\end{lemma}

Lemmas~\ref{lem:label-concept}, \ref{lem:tab-mod} and  \ref{lem:mod-tab} allow us to know that if each tableau $\mathbf{T}$ defined from a universal tableau $T$  for  $\mathcal{O}$ has clash, then  $\mathcal{O}$ is set-theoretically inconsistent. The following lemma establishes the ``set-category" direction of Theorem~\ref{thm:sem-equiv}.

\begin{lemma}\label{lem:set-cat} Let  $\mathcal{O}$  be an $\mathcal{SH}$ ontology, and $T=\langle V_0,\cdots,V_n, E,   L\rangle$  be a universal tableau of $\mathcal{O}$. If all tableau $\mathbf{T}$ built from $T$ have a clash,  then   $\mathcal{O}$ is categorically inconsistent.
\end{lemma}

\begin{proof} 
Assume that there is a universal tableau $T$ such that each tableau $\mathbf{T}$ built from $T$ has a clash. This implies that there is  a sequence $S$ of   applications of properties from  Definition~\ref{def:uni-tableau}, namely     \ref{def:uni-tableau:ind}, \ref{def:uni-tableau:ax}, \ref{def:uni-tableau:conj}, \ref{def:uni-tableau:disj}, \ref{def:uni-tableau:exists}, \ref{def:uni-tableau:forall}, \ref{def:uni-tableau:forall-trans} such that  $S$  allows to build $T$.  Recall that due to minimality of $T$ including $L(x)$, $V_i$ and $\mathbb{S}(x)$, each element in these sets must be added by applying a property in Definition~\ref{def:uni-tableau} (otherwise, i.e if there is a node or concept in $T$ that does not result from an application of a property, then it can be removed from $T$ such that the modified $T$ remains a universal tableau). We show that there is a sequence of applications of properties in $\psh$   that ensures that each  concept category $\mathscr{C}_c\langle \mathcal{O}\rangle$ contains $\{a\}\rightarrow \bot\in \mathsf{Hom}(\mathscr{C}_c\langle \mathcal{O}\rangle)$ for some individual $a$. For each $\mathscr{C}_c\langle \mathcal{O}\rangle$ we use  $\mathscr{C}_{c,r}\langle \mathcal{O}\rangle$ to denote the corresponding  role category. We show that for each node $x\in V_i$ there is some object $C_x$ such that $C_x\rightarrow  C\in \mathsf{Hom}(\mathscr{C}_c\langle \mathcal{O}\rangle)$ for all $C\in L(x)$.

\begin{enumerate}[leftmargin=0cm, itemindent=0.5cm]
 \item\label{thm:set-cat:proof:1}  For  $x^a\in V_0$, let $C_{x^a}=\{a\}$. We show that if a concept  $X$  added to  $L(x^a)$  by some property in Definition~\ref{def:uni-tableau} except for \ref{def:uni-tableau:disj},  then $C_x\rightarrow X\in \mathsf{Hom}(\mathscr{C}_c\langle \mathcal{O}\rangle)$.  We prove   by induction on the length of $S$. Let's consider the following cases.
\begin{itemize}[leftmargin=0cm, itemindent=0.6cm]
\item $X$ is added to $L(x^a)$ by \ref{def:uni-tableau:ind} with $\{a\}\sqsubseteq X\in \mathcal{O}$. By
\ref{org-axioms}, we have $C_{x^a}=\{a\}\rightarrow X\in \mathsf{Hom}(\mathscr{C}_c\langle \mathcal{O}\rangle)$.

\item $X$ is added to $L(x^a)$ by   \ref{def:uni-tableau:ax} with  $X=\mathsf{NNF}(\neg C)\sqcup D$ for $C \sqsubseteq D\in \mathcal{O}$. By \ref{org-syntax-concept}, 
\ref{org-axioms}, \ref{org-neg-top}, \ref{org-disj-d} and \ref{org-disj-c}, we have $C_{x^a}=\{a\}\rightarrow \top \rightarrow  \mathsf{NNF}(\neg C)\sqcup C\rightarrow \mathsf{NNF}(\neg C)\sqcup D\in \mathsf{Hom}(\mathscr{C}_c\langle \mathcal{O}\rangle)$.

\item  $X$ is added to $L(x^a)$ by \ref{def:uni-tableau:conj} with $X\sqcap C\in L(x^a)$. By induction, we have $C_{x^a}=\{a\}\rightarrow X\sqcap C\in \mathsf{Hom}(\mathscr{C}_c\langle \mathcal{O}\rangle)$. By
\ref{org-conj-d}, we have $C_{x^a}\rightarrow X\in \mathsf{Hom}(\mathscr{C}_c\langle \mathcal{O}\rangle)$.

\item  $X$  is never added by \ref{def:uni-tableau:exists} since $x^a\in V_0$, and thus \ref{def:uni-tableau:exists} is not applicable to any node in $V_0$.

\item  $X$  is added to $L(x^a)$ by \ref{def:uni-tableau:forall} where $x^a$ is a $P$-candidate root successor of $x^b$ and $\forall R.X\in L(x^b)$ with $ \{(a,b)\} \rightarrow P, P\rightarrow R\in \mathsf{Hom}(\mathscr{C}_r\langle \mathcal{O}\rangle)$. By induction, $\{b\}\rightarrow \forall R.X$. By \ref{org-exists-forall-ind},  $C_{x^a}=\{a\}\rightarrow X\in \mathsf{Hom}(\mathscr{C}_c\langle \mathcal{O}\rangle)$.

\item  $X$ is added by \ref{def:uni-tableau:forall-trans} to $L(x^a)$ with $X=\forall S.D$    where $x^a$ is a $P$-candidate root successor of $x^b$ and $\forall R.D\in L(x^b)$ with $\{(a,b)\} \rightarrow P, P\rightarrow S, S\circ S\rightarrow S, S\rightarrow R \in \mathsf{Hom}(\mathscr{C}_{c,r}\langle \mathcal{O}\rangle)$.  By induction, $\{b\}\rightarrow \forall R.D$, and thus by \ref{org-exists-forall-comp-ind},  $C_{x^a}=\{a\}\rightarrow \forall S.D\in \mathsf{Hom}(\mathscr{C}_c\langle \mathcal{O}\rangle)$.
\end{itemize}

\item\label{thm:set-cat:proof:2} We now consider \ref{def:uni-tableau:disj}. For $x^a\in V_0$, it holds that $C_{x^a}=\{a\}\rightarrow C \in \mathsf{Hom}(\mathscr{C}_c\langle \mathcal{O}\rangle)$ for all concept $C$  added to $L(x^a)$ by properties in Definition~\ref{def:uni-tableau} different from \ref{def:uni-tableau:disj} due to Item~\ref{thm:set-cat:proof:1}. 
By the construction of $\mathbb{T}(x^a)$ (Definition~\ref{def:tree}) and \ref{org-dist}, for each  interior node $v$ of $\mathbb{T}(x^a)$ it holds that 

$\displaystyle{C_{x^a} \rightarrow  (\bigsqcap_{C\in L(v')} C \sqcap E) \sqcup (\bigsqcap_{C\in L(v'')} C \sqcap F)\in \mathsf{Hom}(\mathscr{C}_c\langle \mathcal{O}\rangle)}$  where $E\sqcup F\in L(v)$, $E\neq F$ and $v',v''$ are successors of $v$. According to Claim~\ref{claim:tab}, the leaves of $\mathbb{T}(x^a)$ are nodes in $\mathbb{S}(x^a)$. By \ref{org-syntax-concept} (transitivity), \ref{org-dist}, 
it holds that   $C_{x^a} \rightarrow  \displaystyle{ \bigsqcup_{y^a\in \mathbb{S}(x^a)} (\bigsqcap_{C\in L(y^a)} C)}  \in \mathsf{Hom}(\mathscr{C}_c\langle \mathcal{O}\rangle)$.

Let $x\in V_i$ with $i>0$ such that $x$ has a predecessor $y$. By construction and minimality of $T$, $\mathbf{T}_\exists$ is applied  to  some  $\exists P.C\in L(y)$ and $C\in L(x)$.  By \ref{org-exists-filler}, we have $p_r(\mathfrak{R}^\varnothing_{(\exists P.C)}) \rightarrow C\in  \mathsf{Hom}(\mathscr{C}_c\langle \mathcal{O}\rangle)$. Due to \ref{org-exists-forall} and \ref{org-exists-forall-comp}, if $\mathcal{C}$ is the set of all concepts of the form $\forall R.D\in L(y)$ where $P\rightarrow R \in  \mathsf{Hom}(\mathscr{C}_c\langle \mathcal{O}\rangle)$,  or $P\rightarrow S \rightarrow R, S\circ S\rightarrow S\in   \Homrcat$, then $p_r(\mathfrak{R}^\mathcal{C}_{(\exists P.C)}) \rightarrow X\in  \mathsf{Hom}(\mathscr{C}_c\langle \mathcal{O}\rangle)$ where $X=C$, $X=D$ or $X=\forall S.D$. Moreover, $p_r(\mathfrak{R}^\mathcal{C}_{(\exists P.C)}) \rightarrow \top \rightarrow \mathsf{NNF}(\neg C) \sqcup C \rightarrow \mathsf{NNF}(\neg C) \sqcup D\in  \mathsf{Hom}(\mathscr{C}_c\langle \mathcal{O}\rangle)$ for each $C\sqsubseteq D\in \mathcal{O}$ by  \ref{org-syntax-concept}, 
\ref{org-axioms}, \ref{org-neg-top}, \ref{org-disj-d} and \ref{org-disj-c}. Let $C_x=p_r(\mathfrak{R}^\mathcal{C}_{(\exists P.C)})$. This implies that $C_x \rightarrow X \in \mathsf{Hom}(\mathscr{C}_c\langle \mathcal{O}\rangle) $ for all $X\in L(v_0)$ where $v_0$ is the root of $\mathbb{T}(x)$ associated with $\mathbb{S}(x)$ in Definition~\ref{def:tree}.  By \ref{org-syntax-concept} (transitivity), 
\ref{org-dist}
it holds that
$C_x\rightarrow  \displaystyle{ \bigsqcup_{x'\in \mathbb{S}(x)} (\bigsqcap_{C\in L(x')} C)}  \in \mathsf{Hom}(\mathscr{C}_c\langle \mathcal{O}\rangle)$. By induction on  level of $T$, $C_y\rightarrow X \in \mathsf{Hom}(\mathscr{C}_c\langle \mathcal{O}\rangle)$ for  all $X\in L(y)$, and $ X\rightarrow p_l(\mathfrak{R}^\mathcal{C}_{(\exists P.C)})\in \mathsf{Hom}(\mathscr{C}_c\langle \mathcal{O}\rangle)$ for  some $X\in L(y)$ due to \ref{org-exists-forall} or \ref{org-exists-forall-comp}. Thus,  $C_y\rightarrow p_l(\mathfrak{R}^\mathcal{C}_{(\exists P.C)}) \in \mathsf{Hom}(\mathscr{C}_c\langle \mathcal{O}\rangle)$.
   
\item\label{thm:proof:4} We show that if each tableau $\mathbf{T}$ built from $T$ has a clash, then there is some individual $a$ such that each $y^a\in \mathbb{S}(x^a)$  has a clash in $T$, i.e  $y^a\in \mathbb{S}(x^a)$ has a descendant $y$ in $T$ such that each $y'\in \mathbb{S}(y)$ has a direct or indirect  clash. We need to show that such a clash can be propagated from $y'$ to $y^a$ thanks to \ref{org-syntax-functor}.

According to Item~\ref{thm:set-cat:proof:2}, for each $y$ of $T$ that has a candidate successor $x$ there is a role object $\mathfrak{R}^\mathcal{C}_{(\exists P.C)}$ such that $\exists P.C\in L(y)$,  $\mathcal{C}\subseteq  L(y)$, and 
\begin{align}
&C_x\rightarrow  \displaystyle{ \bigsqcup_{x'\in \mathbb{S}(x)} (\bigsqcap_{C\in L(x')} C)}  \in \mathsf{Hom}(\mathscr{C}_c\langle \mathcal{O}\rangle) 
\label{thm:set-sat:proof-31}\\ 
&C_y  \rightarrow   p_l(\mathfrak{R}^\mathcal{C}_{(\exists P.C)}) \in \mathsf{Hom}(\mathscr{C}_c\langle \mathcal{O}\rangle) 
\label{thm:set-sat:proof-32}
\end{align}
Therefore, the clashes are propagated until $V_0$ via $p_r(\mathfrak{R}^\mathcal{C}_{(\exists P.C)})$,  $p_l(\mathfrak{R}^\mathcal{C}_{(\exists P.C)})$ and \ref{org-syntax-functor}. Indeed, if $x$ is a candidate successor of $y$ and each $x'\in \mathbb{S}(x)$ has a direct or indirect clash, then $\displaystyle{\bigsqcap_{C\in L(x')} C}  \rightarrow \bot\in \mathsf{Hom}(\mathscr{C}_c\langle \mathcal{O}\rangle)$ due to \ref{org-neg-sqcap} and \ref{org-conj-d}. By (\ref{thm:set-sat:proof-31}) and \ref{org-disj-c}, we have $C_x \rightarrow \bot$. By \ref{org-syntax-functor}, we obtain $p_l(\mathfrak{R}^\mathcal{C}_{(\exists P.C)})\rightarrow \bot$, and (\ref{thm:set-sat:proof-32}) we have $C_y \rightarrow \bot$.
Hence, we have   
 $\{a\}\rightarrow  \displaystyle{ \bigsqcup_{y^a\in \mathbb{S}(x^a)} (\bigsqcap_{C\in L(y^a)} C)}  \in \mathsf{Hom}(\mathscr{C}_c\langle \mathcal{O}\rangle)$ with $   \displaystyle{\bigsqcap_{C\in L(y^a)} C} \rightarrow \bot  \in \mathsf{Hom}(\mathscr{C}_c\langle \mathcal{O}\rangle)$. Therefore, $\{a\}\rightarrow \bot \in \mathsf{Hom}(\mathscr{C}_c\langle \mathcal{O}\rangle)$, and thus $\mathcal{O}$ is categorically inconsistent  according to Definition~\ref{def:SH-consistency}.
\end{enumerate}
\end{proof}

The following corollary is a consequence of the construction of  $\mathsf{Hom}(\mathscr{C}_c\langle \mathcal{O}\rangle)$ from a universal tableau $T$  presented in the proof of Lemma~\ref{lem:set-cat}.

\begin{corollary}\label{cor:rulesforSH}
     The construction of a concept category $\mathsf{Hom}(\mathscr{C}\langle \mathcal{O}\rangle)$ containing $\{a\}\rightarrow \bot$ from a set-theoretically inconsistent $\mathcal{SH}$ ontology $\mathcal{O}$ in NNF does not require the properties \ref{org-conj-c}, \ref{org-neg-bot}, \ref{org-neg-top}, \ref{org-exists-comp} and \ref{org-forall-weakened} from $\psh$.
\end{corollary}

The non-requirement of \ref{org-neg-bot}, \ref{org-neg-top} and \ref{org-forall-weakened} is due to the fact that concepts in a universal tableau $T$ are in NNF. \ref{org-conj-c} is covered by   \ref{org-dist}, \ref{org-conj-d}, \ref{org-disj-c}, \ref{org-disj-d} and \ref{org-syntax-concept}    
 while \ref{org-exists-comp} results from   \ref{org-dist} , \ref{org-neg-top} and \ref{org-exists-forall-comp}.

%
%

\section{Checking Categorical Consistency}\label{sec:sh-consistency}
 To reason on an $\mathcal{SH}$ ontology  under the categorical semantics, we need a set of rules in Table~\ref{tab:SHrules}, denoted $\rsh$, to build sets of objects and arrows. These rules represent a reduced version of $\psh$ since the  properties guaranteeing the NNF conversion  are no longer needed when concept objects are written in NNF.  
 Each rule $\mathsf{R}$ in Table~\ref{tab:SHrules} is composed of three components as follows: $\mathsf{R= \displaystyle{\frac{[premise]}{[conclusion]} :  [condition]} }$ 
  where $[\mathsf{premise}]$ (possibly empty) and $[\mathsf{conclusion}]$ can contain several arrows. These rules aim  to build from an $\mathcal{SH}$ ontology $\mathcal{O}$ a set $\mathsf{Hom}\tuple{\mathcal{O}}$ of arrows and a set $\mathsf{Ob}\tuple{\mathcal{O}}$ of objects such that $C\in \mathsf{Ob}\tuple{\mathcal{O}}$ iff there is an arrow $C\rightarrow D \in \mathsf{Hom}\tuple{\mathcal{O}}$ or $D\rightarrow C \in \mathsf{Hom}\tuple{\mathcal{O}}$.
 $\mathsf{R}$ is applicable if (i) one of the arrows in $[\mathsf{conclusion}]$ is not included in $\mathsf{Hom}\tuple{\mathcal{O}}$, (ii)  the arrows in  $[\mathsf{premise}]$ are present in $\mathsf{Hom}$,  and (iii) $[\mathsf{condition}]$ (possibly empty) holds. For instance, if $\mathsf{R}= \displaystyle{\frac{}{X\rightarrow X~~~X\rightarrow \top~~~\bot\rightarrow X}} $, then the $\mathsf{premise}$ of $\mathsf{R}$ is empty, the $\mathsf{condition}$ of $\mathsf{R}$ is empty (and thus they always hold), and the $\mathsf{conclusion}$ of $\mathsf{R}$ is $\{X\rightarrow X, X\rightarrow \top, \bot\rightarrow X\}$.
\vspace{0.1cm}
\begin{table}[H]   
  \hrule
  \vspace{0.2cm}
\begin{enumerate}[leftmargin=!, labelwidth=1em, align=left]

\item[\mylabel{sh-id}{$\mathsf{R_{id}}$}$:$] 
 
$\displaystyle{\frac{}{X \rightarrow X ~~~X \rightarrow \top~~~\bot \rightarrow X}}$

\item[\mylabel{sh-trans}{$\mathsf{R_{tr}}$}$:$] 
 
$\displaystyle{\frac{X \rightarrow Y~~~Y \rightarrow Z}{X \rightarrow Z} }$

\item[\mylabel{sh-functor}{$\mathsf{R_{f}}~$}$:$] 
 
$\displaystyle{\frac{R \rightarrow S}{p_l(R) \rightarrow p_l(S) ~~~p_r(R) \rightarrow p_r(S)}}~~~$
$\displaystyle{\frac{p_r(R) \rightarrow \bot}{p_l(R) \rightarrow \bot} }$

\item[\mylabel{sh-axioms}{$\mathsf{R_\sqsubseteq}$}$:$]
 
 $\displaystyle{ \frac{}{X\rightarrow Y  ~~~\top \rightarrow  \mathsf{NNF}(\neg X)\sqcup X ~~~\mathsf{NNF}(\neg X)\sqcup Y \rightarrow  \mathsf{NNF}(\neg X)\sqcup Y} : \mathsf{GCI} ~X\sqsubseteq Y \in \mathcal{O} }$

\item[\mylabel{sh-conj-d}{$\mathsf{R}_\sqcap^\mathsf{d}$}$:$]

$\displaystyle{\frac{}{C \sqcap D \rightarrow C~~~ C \sqcap D \rightarrow  D } : C\sqcap D\in  \mathsf{Ob}\langle \mathcal{O}\rangle}$

\item[\mylabel{sh-disj-d}{$\mathsf{R}^\mathsf{d}_\sqcup$}$:$]

$\displaystyle{\frac{C\sqcup D\rightarrow X}{C \rightarrow  X ~~~D\rightarrow X  } }$

\item[\mylabel{sh-disj-c}{$\mathsf{R}^\mathsf{c}_\sqcup$}$:$]

$\displaystyle{\frac{C \rightarrow X ~~~ D \rightarrow X }{C\sqcup D \rightarrow  X} : C\sqcup D\in  \mathsf{Ob}\langle \mathcal{O}\rangle }$

\item[\mylabel{sh-neg-sqcap}{$\mathsf{R}_\neg^\sqcap$}$:$] 
 
$\displaystyle{\frac{X  \rightarrow A ~~~ X  \rightarrow \neg A }{X \rightarrow   \bot}  }$

\item[\mylabel{sh-exists-filler}{$\mathsf{R}_\exists^\mathfrak{R}$}$:$]

$\displaystyle{\frac{ }{ \mathfrak{R}_{(\exists R.C)}\rightarrow R~~~ \exists R.C\leftrightarrows p_l(\mathfrak{R}_{(\exists R.C)}) ~~~p_r(\mathfrak{R}_{(\exists R.C)})\rightarrow C} : \exists R.C\in  \mathsf{Ob}\langle \mathcal{O}\rangle}$

\item[\mylabel{sh-exists}{$\mathsf{R}_\exists$}$:$]

$\displaystyle{\frac{R'\rightarrow R ~~~ p_r(R') \rightarrow C}{ p_l(R') \rightarrow p_l(\mathfrak{R}_{(\exists R.C)})}    :  \exists R.C\in  \mathsf{Ob}\langle \mathcal{O}\rangle}$

\item[\mylabel{sh-dist}{$\mathsf{R}_\sqcap^\sqcup$}$:$] 

$\displaystyle{\frac{ X\rightarrow  C ~~~ X \rightarrow  D\sqcup E  }{ X \rightarrow  (C\sqcap D)\sqcup (C\sqcap E) }}$

\item[\mylabel{sh-exists-forall}{$\mathsf{R_\exists^{\forall}}$}$:$]  

$\displaystyle{\frac{ X\rightarrow  p_l(\mathfrak{R}^\mathcal{C}_{\exists P.C}) ~~~  X\rightarrow \forall R.D~~~  P \rightarrow  R}{  \mathfrak{R}^{\mathcal{C}\cup\{\forall R.D\} }_{(\exists P.C)}\rightarrow \mathfrak{R}^\mathcal{C}_{(\exists P.C)} ~~~ X \leftrightarrows p_l(\mathfrak{R}^{\mathcal{C}\cup\{\forall R.D\}}_{(\exists P.C)}) ~~~p_r(\mathfrak{R}^{\mathcal{C}\cup\{\forall R.D\}}_{(\exists P.C)})\rightarrow D   } }$

\item[\mylabel{sh-exists-forall-ind}{$\mathsf{R}_\exists^{\forall \mathsf{I}}$}$:$]

$\displaystyle{\frac{\{(a,b)\}\rightarrow P ~~~ \{a\}\rightarrow \forall R.D ~~~P \rightarrow  R}{ \{b\}\rightarrow D}   }$

\item[\mylabel{sh-exists-forall-comp}{$\mathsf{R_{\exists \circ}^{\forall}}$}$:$]

$\displaystyle{\frac{X\rightarrow  p_l(\mathfrak{R}^\mathcal{C}_{\exists P.C}) ~~~ X\rightarrow\forall R.D  ~~~ P\rightarrow  S ~~~ S\circ S \rightarrow  S ~~~ S \rightarrow  R ~~~   }{\mathfrak{R}^{\mathcal{C}\cup \{\forall R.D\}}_{(\exists P.C)}\rightarrow\mathfrak{R}^\mathcal{C}_{(\exists P.C)}~~~ X \leftrightarrows  p_l(\mathfrak{R}^{\mathcal{C}\cup \{\forall R.D\}}_{(\exists P.C)})  ~~~ p_r(\mathfrak{R}^{\mathcal{C}\cup \{\forall R.D\}}_{(\exists P.C)}) \rightarrow \forall S.D}}$

\item[\mylabel{sh-exists-forall-comp-ind}{$\mathsf{R}_{\exists \circ}^{\forall\mathsf{I}}$}$:$]  

$\displaystyle{\frac{\{(a,b)\}\rightarrow P~~~\{a\}\rightarrow\forall R.D~~~P\rightarrow  S ~~~ S\circ S \rightarrow  S ~~~ S \rightarrow  R}{\{b\}\rightarrow   \forall S.D }  }$
 
\end{enumerate}
\hrule
\vspace{0.2cm}
 \caption{Categorical Rules $\rsh$ for $\mathcal{SH}$}
  \label{tab:SHrules}
\end{table}
 
It seems  that  $\rsh$ is weaker than $\psh$
since there are some of properties from $\psh$ such as \ref{org-conj-c}, \ref{org-neg-bot}, \ref{org-neg-top},   \ref{org-forall-weakened} and \ref{org-exists-comp}  
that are not explicitly expressed  in $\rsh$. However, the following lemma  affirms that  $\rsh$ is equivalent to $\psh$ if concept objects in categories are written in NNF. 

\begin{lemma}\label{lem:sh-consistency}
Let $\mathcal{O}$  be  an  $\mathcal{SH}$ ontology. 
Let  $\mathsf{Hom}\tuple{\mathcal{O}}$ be the set of arrows built from $\mathcal{O}$ by applying  rules  $\rsh$  in  Table~\ref{tab:SHrules} until no rule is applicable. It holds that 
\begin{enumerate}[leftmargin=0.1cm, itemindent=0.5cm] 


\item\label{lem:sh-consistency:1}   $\mathcal{O}$  is categorically inconsistent iff $\mathsf{Hom}\tuple{\mathcal{O}}$ contains   $\{a\}\rightarrow \bot$ for some individual $a$. 

\item\label{lem:sh-consistency:2} The cardinality of $\mathsf{Hom}\tuple{\mathcal{O}}$ is bounded by an exponential function in the size of $\mathcal{O}$ up to object isomorphism.

\end{enumerate}
\end{lemma}

\begin{proof} 
\begin{enumerate}[leftmargin=0.1cm, itemindent=0.5cm] 

\item  The direction  ``$\Longleftarrow$" of Item~\ref{lem:sh-consistency:1}. It holds that each rule from $\rsh$ adds an arrow $X\rightarrow Y$ to  $\mathsf{Hom}\tuple{\mathcal{O}}$, then  $\mathcal{O}\models X \sqsubseteq Y$  according to Theorem~\ref{thm:cat-set}.
Hence, if $\mathsf{Hom}\tuple{\mathcal{O}}$ contains   $\{a\}\rightarrow \bot$, then $\mathcal{O}\models \{a\}\sqsubseteq \bot$, and thus  $\mathcal{O}$ is set-theoretically inconsistent. By Theorem~\ref{thm:sem-equiv}, $\mathcal{O}$ is categorically inconsistent. 

The direction ``$\Longrightarrow$" of Item~\ref{lem:sh-consistency:1}. Assume that $\mathcal{O}$ is categorically  inconsistent. We have to show that there is a sequence  of applications of rules in $\rsh$ that allows to build $\mathsf{Hom}\tuple{\mathcal{O}}$ containing   $\{a\}\rightarrow \bot$. 
By Theorem~\ref{thm:sem-equiv},  $\mathcal{O}$ is set-theoretically inconsistent. By Lemma~\ref{lem:mod-tab},  every tableau $\mathbf{T}$ built from  universal tableau $T$ has a clash. According to Lemma~\ref{lem:set-cat} and Corollary~\ref{cor:rulesforSH},  there is a sequence $S$ of applications of properties from $\psh$  such that $S$ builds from $T$ a concept category $\occat$ containing $\{a\}\rightarrow \bot$ and  $S$ does not contain any application of  \ref{org-conj-c}, \ref{org-neg-bot},  \ref{org-neg-top}, \ref{org-forall-weakened} and \ref{org-exists-comp} under the hypothesis that all concept objects are written in NNF. Moreover,  \ref{org-neg-sqcup} is applied only if $C$ occurs in an axiom $C\sqsubseteq D\in \mathcal{O}$, and  \ref{org-neg-sqcap} is replaced with  the following : $X\rightarrow A$ and $X\rightarrow \neg A$ imply $X\rightarrow \bot$.

We use $\mathsf{Hom_i}(\occat)$ to denote the set of arrows added by $S$ from the beginning until an application $\mathsf{S^i_P}\in S$ where $\mathsf{P}\in \psh \setminus \{\text {\ref{org-conj-c}, \ref{org-neg-bot},  \ref{org-neg-top}, \ref{org-forall-weakened},\ref{org-exists-comp}}\}$. 
Hence, it suffices to show by induction on the length of $S$  that there are applications of rules in $\rsh$ for each $\mathsf{S^i_P}\in S$.  We use $\mathsf{Hom_i}\tuple{\mathcal{O}}$ to denote the set of arrows added by rules from $\rsh$ when $\mathsf{S^i_P}$ is applied. Assume that $\mathsf{Hom_0}(\occat)=\mathsf{Hom_0}\tuple{\mathcal{O}}=\varnothing$. For each update by $\mathsf{S^i_P}$, we show that $\mathsf{Hom_i}(\occat)\subseteq \mathsf{Hom_i}\tuple{\mathcal{O}}$.

\begin{itemize}[leftmargin=0.1cm, itemindent=0.5cm]
    \item Assume  that $\mathsf{P}\in \{\text{\ref{org-syntax-role},  \ref{org-syntax-concept},  \ref{org-syntax-functor}, \ref{org-axioms}}\}$ ensures the existence of an arrow $C\rightarrow D$ in $\mathsf{Hom_i}(\occat)$.  In this case, the premise of one of rules \ref{sh-id}, \ref{sh-trans}, \ref{sh-functor} and \ref{sh-axioms} holds  due to  $\mathsf{Hom_{i-1}}(\occat)\subseteq \mathsf{Hom_{i-1}}\tuple{\mathcal{O}}$ by induction. Hence, it can add  $C\rightarrow D$ to $\mathsf{Hom_i}\tuple{\mathcal{O}}$. Therefore, $\mathsf{Hom_i}(\occat)\subseteq \mathsf{Hom_i}\tuple{\mathcal{O}}$ holds.

    \item Assume  that $\mathsf{P}\in \{\text{\ref{org-neg-sqcup},  \ref{org-neg-sqcap}}\}$ ensures the existence of an arrow $C\rightarrow D$ in $\mathsf{Hom_i}(\occat)$ under the context described above. In this case,  the premise of one of  rules \ref{sh-axioms}, \ref{sh-neg-sqcap} holds  by induction. Hence, it can add
    $C\rightarrow D$ to $\mathsf{Hom_i}\tuple{\mathcal{O}}$   Hence, $\mathsf{Hom_i}(\occat)\subseteq \mathsf{Hom_i}\tuple{\mathcal{O}}$ holds.

    \item Assume  that $\mathsf{P}\in \{\text{\ref{org-conj-d}, \ref{org-disj-d}, \ref{org-disj-c},  \ref{org-exists-filler}, \ref{org-exists}, \ref{org-dist}, \ref{org-exists-forall}, \ref{org-exists-forall-ind}, \ref{org-exists-forall-comp}, \ref{org-exists-forall-comp-ind}   }\}$ ensures the existence of an arrow $C\rightarrow D$ in $\mathsf{Hom_i}(\occat)$. In this case,  the premise of one of rules \ref{sh-conj-d},  
    \ref{sh-disj-d}, \ref{sh-disj-c},  \ref{sh-exists-filler}, \ref{sh-exists}, \ref{sh-dist}, \ref{sh-exists-forall}, \ref{sh-exists-forall-ind}, \ref{sh-exists-forall-comp}, \ref{sh-exists-forall-comp-ind} holds  by induction. Hence, it can add
    $C\rightarrow D$ to $\mathsf{Hom_i}\tuple{\mathcal{O}}$   Hence, $\mathsf{Hom_i}(\occat)\subseteq \mathsf{Hom_i}\tuple{\mathcal{O}}$ holds.
     
\end{itemize}

Since $S$ is finite, there is some $i$ such that $\mathsf{S^i_P}$ adds $\{a\}\rightarrow \bot$ to $\mathsf{Hom_i}(\occat)$. Hence,  $\{a\}\rightarrow \bot\in \mathsf{Hom_i}\tuple{\mathcal{O}}$ due to $\mathsf{Hom_i}(\occat)\subseteq \mathsf{Hom_i}\tuple{\mathcal{O}}$.
   
\item We show Item~\ref{lem:sh-consistency:2}. Since there are at most two arrows between 2 objects, and a rule is not applicable if the arrows in the conclusion of the rule are already in  $\mathsf{Hom}\tuple{\mathcal{O}}$, it suffices to compute the number of objects added by rules to determine the cardinality of $\mathsf{Hom}\tuple{\mathcal{O}}$.
If an ontology $\mathcal{O}$ can be seen as a string of length $k$, then  $\lvert \mathcal{O}\rvert=k$ where $\lvert \mathcal{O}\rvert$   denotes the size of $\mathcal{O}$. Hence, the number of concept objects from $\mathcal{O}$ is bounded by  $O(k^2)$. We use  $\mathbf{C}_0$, $\mathbf{R}_0$ and $\mathbf{I}_0$  to denote respectively the sets of concept names, role names and individuals occurring in  $\mathcal{O}$.  We determine  extensions $\mathbf{C}, \mathbf{R}$ and $\mathbf{I}$ from  $\mathbf{C}_0$, $\mathbf{R}_0$ and $\mathbf{I}_0$   as follows: (i)   $\mathbf{C}, \mathbf{R}, \mathbf{I}$ are the smallest sets such that $\mathbf{C}_0\subseteq \mathbf{C}$, $\mathbf{R}_0\subseteq \mathbf{R}$ and $\mathbf{I}_0\subseteq \mathbf{I}$, and (ii) for each concept of the form $\exists R.C$ occurring in $\mathcal{O}$, and each subset $\mathcal{C}$ of the set of all concepts of the form $\forall R.D$ occurring in $\mathcal{O}$, $\mathbf{R}$ contains a role name $\mathfrak{R}^\mathcal{C}_{(\exists R.C)}$, and  $\mathbf{C}$ contains  two concept names $p_l(\mathfrak{R}^\mathcal{C}_{(\exists R.C)})$ and $p_r(\mathfrak{R}^\mathcal{C}_{(\exists R.C)})$. We consider the rules involved in generating of these fresh role objects.

\begin{itemize}[leftmargin=0.1cm, itemindent=0.2cm]
\item It holds that the number of concepts   of the form $\exists R.C$ and  $\forall R.C$ occurring in $\mathcal{O}$ is bounded by $k$. If a rule from $\rsh$ adds an object of the form $\exists R.C$ to $\mathsf{Ob}\tuple{\mathcal{O}}$, then $\exists R.C$ must occur in $\mathcal{O}$. This implies that    the number of concept objects   of the form $\exists R.C$ in $\mathsf{Ob}\tuple{\mathcal{O}}$ is bounded by $k^2$. Hence,  the number of  fresh role objects   of the form $\mathfrak{R}^\varnothing_{\exists R.C}$ added by \ref{sh-exists-filler}  to $\mathsf{Ob}\tuple{\mathcal{O}}$ is bounded by $k^2$. 

\item Two rules  \ref{sh-exists-forall-comp}, \ref{sh-exists-forall-comp-ind} can add fresh  concept objects of the form $\forall S.D$ if there is a concept $\forall R.D$ and a transitive role $S$ occur in $\mathcal{O}$. This implies that    the number of objects   of the form $\forall S.D$ in $\mathsf{Ob}\tuple{\mathcal{O}}$ is bounded by $k^2$ where $k$ bounds the number of roles $S\in \mathbf{R}_0$, and the number of concepts of the form $\forall R.D$ occurring in $\mathcal{O}$.

\item Two rules \ref{sh-exists-forall}, \ref{sh-exists-forall-comp} can add fresh  role objects   of the form $\mathfrak{R}^\mathcal{C}_{\exists R.C}$ where $\mathcal{C}$ is a subset of the set of all concept objects of the form $\forall R.C$. This implies that the number of role objects   of the form $\mathfrak{R}^\mathcal{C}_{\exists R.C}$ in $\mathsf{Ob}\tuple{\mathcal{O}}$  is bounded by $k2^{k}$ where $k$ bounds the number of concepts of the forms $\exists R.C$ and  $\forall R.C$ occurring in $\mathcal{O}$; and $2^k$ bounds the number of different $\mathcal{C}$. Therefore, $\lvert \mathbf{R}\rvert \leqslant O(k2^{k})$.

\item The rule \ref{sh-functor} adds two fresh concept objects $p_l(R)$ and $p_r(R)$ for each role object $R$. Therefore, $\lvert \mathbf{C}\rvert \leqslant O(2k2^{k})$.
 
\item \ref{sh-id} and \ref{sh-trans} add no new object.

\item \ref{sh-axioms} adds a new concept objects $\mathsf{NNF}(\neg X)\sqcup X$, $\mathsf{NNF}(\neg X)\sqcup Y, \mathsf{NNF}(\neg X), Y$ for each axiom $X\sqsubseteq Y\in \mathcal{O}$.  Therefore, the number of concept objects added by this rule is bounded by $4k$ where $k$ bounds the number of each new object.

\item All rules \ref{sh-conj-d} to \ref{sh-neg-sqcap} and  \ref{sh-exists} add no new object.

\item  \ref{sh-dist} adds new concept objects of the form $(C\sqcap D)\sqcup  (C\sqcap E), C\sqcap D, C\sqcap E$ from arrows $X\rightarrow C$ and $X \rightarrow  D\sqcup E$ where disjunction $D\sqcup E$ occurs in $\mathcal{O}$ or it is added by \ref{sh-axioms} or it is added by \ref{sh-dist} itself. It holds that  all objects produced by \ref{sh-dist}  can be regrouped in  OR-trees with branching nodes corresponding to $\sqcup$ such that the depth of trees is bounded by $d$. We show that $d\leqslant k$. 

Since a fresh concept name $p_l(\mathfrak{R}^\mathcal{C}_{(\exists P.C)})$ is introduced in $\mathsf{Hom}\tuple{\mathcal{O}}$ with arrows $X\leftrightarrows p_l(\mathfrak{R}^\mathcal{C}_{(\exists P.C)})$ by \ref{sh-exists-forall} or \ref{sh-exists-forall-comp}, it  is always isomorphic to some concept $X$ occurring in $\mathcal{O}$. However, $p_r(\mathfrak{R}^\mathcal{C}_{(\exists P.C)})$ is introduced in $\mathsf{Hom}\tuple{\mathcal{O}}$ with an arrow $p_r(\mathfrak{R}^\mathcal{C}_{(\exists P.C)}) \rightarrow D$ or $p_r(\mathfrak{R}^\mathcal{C}_{(\exists P.C)}) \rightarrow \forall S.D$  by \ref{sh-exists-forall} or \ref{sh-exists-forall-comp}.   Hence, $p_r(\mathfrak{R}^\mathcal{C}_{(\exists P.C)})$ can occur in a conjunction generated by \ref{sh-dist} but it never occurs in a disjunction as a disjunct like   $(p_r(\mathfrak{R}^\mathcal{C}_{(\exists P.C)}) \sqcup Y)$. This implies that if $p_r(\mathfrak{R}^\mathcal{C}_{(\exists P.C)})$ occurs in a node of an OR-tree, it occurs in both successors of that node. Hence, the depth $d$ of each OR-tree is bounded by the number of disjunctions occurring in $\mathcal{O}$ or added by \ref{sh-axioms}, and thus $d\leqslant k$. Therefore, the number of conjunctions generated by \ref{sh-dist} for all OR-trees is bounded by $k2^k$ up to object isomorphism.
\end{itemize}

We have proved that each rule from $\psh$ adds an exponential number of objects in the size of $\mathcal{O}$. Therefore, the number of objects added by all rules in $\psh$ is bounded by an exponential function in the size of $\mathcal{O}$, and thus the cardinality of $\mathsf{Hom}\tuple{\mathcal{O}}$   is bounded by an exponential function   in the size of $\mathcal{O}$ up to object isomorphism.
\end{enumerate}
\end{proof}

\begin{example}\label{ex:def3running}
We consider an ontology $\mathcal{O}$ consisting of the following axioms:

\begin{tabular}{ccccc}
$\mathsf{(ax1):}\{a\}\sqsubseteq \exists S.C$   
&$(\mathsf{ax2):}\{a\}\sqsubseteq \forall S.\neg D$
&$(\mathsf{ax3):}C\sqsubseteq \exists S.D$ 
&$(\mathsf{ax4):}S\circ S\sqsubseteq S$
\end{tabular}

Note that $\mathcal{O}$ is  set-theoretically inconsistent. We now check categorical consistency  of $\mathcal{O}$ by applying categorical rules  $\rsh$ to the axioms of $\mathcal{O}$ with arrows generated on the right-hand side.
\begin{align}
   &\text {\ref{sh-axioms} to } (\mathsf{ax1}), (\mathsf{ax2}) & \{a\}\rightarrow \exists S.C, \{a\}\rightarrow \forall S.\neg D
   \label{ex:run1}\\
  &\text{\ref{sh-exists-filler} to } \exists S.C & p_r(\mathfrak{R}^\varnothing_{(\exists S.C)}) \rightarrow  C     \label{ex:run3}\\
  & \text{\ref{sh-id}, \ref{sh-axioms}  to } (\mathsf{ax3}),\text{ \ref{sh-disj-d}, \ref{sh-disj-c},   \ref{sh-trans} }&C \rightarrow \top \rightarrow \neg C \sqcup C\rightarrow \neg C\sqcup \exists S.D     \label{ex:run4}\\
  & \text{\ref{sh-trans} to (\ref{ex:run4})}&C \rightarrow   \neg C\sqcup \exists S.D     \label{ex:run41}\\
  &\text{\ref{sh-dist} to } (\ref{ex:run4}) \text{ and  } C \rightarrow C &C \rightarrow (\neg C\sqcap C) \sqcup (C\sqcap \exists S.D)   \label{ex:run5}\\
  & \text{\ref{sh-neg-sqcap} and  \ref{sh-id} }&\neg C \sqcap C \rightarrow  \bot \rightarrow C\sqcap \exists S.D      \label{ex:run6}\\
  & \text {\ref{sh-disj-c} to  (\ref{ex:run6})}&(\neg C \sqcap C)\sqcup (C\sqcap \exists S.D) \rightarrow  (C\sqcap \exists S.D)       \label{ex:run7}\\
  &\text {\ref{sh-trans} to  (\ref{ex:run5}), (\ref{ex:run7}) and \ref{sh-conj-d}}& C \rightarrow   \exists S.D       \label{ex:run8}\\
   &\text {\ref{sh-exists-forall-comp} to  (\ref{ex:run1}) and } (\mathsf{ax4})& \{a\} \rightarrow p_l(\mathfrak{R}_{(\exists S.C)}^{\{\forall S.\neg D\}}), \mathfrak{R}_{(\exists S.C)}^{\{\forall S.\neg D\}}\rightarrow \mathfrak{R}^\varnothing_{(\exists S.C)}      \label{ex:run10}\\
   &\text {\ref{sh-exists-forall-comp} to   (\ref{ex:run1})  and } (\mathsf{ax4})&p_r(\mathfrak{R}_{(\exists S.C)}^{\{\forall S.\neg D\}}) \rightarrow   \forall S. \neg D    \label{ex:run12}\\
   & \text{\ref{sh-functor} to (\ref{ex:run10}), (\ref{ex:run3}) and (\ref{ex:run8})}  &p_r(\mathfrak{R}_{(\exists S.C)}^{\{\forall S.\neg D\}}) \rightarrow  p_r(\mathfrak{R}^\varnothing_{(\exists S.C)}) \rightarrow C \rightarrow \exists S.D \label{ex:run12a}\\
   &\text {\ref{sh-exists-forall}, \ref{sh-trans} to (\ref{ex:run12}) and   (\ref{ex:run12a})} &p_r(\mathfrak{R}_{(\exists S.C)}^{\{\forall S.\neg D\}}) \rightarrow  p_l(\mathfrak{R}_{(\exists S.D)}^{\{\forall  S.\neg D\}}), p_r(\mathfrak{R}_{(\exists S.D)}^{\{\forall S.\neg D\}})\rightarrow \neg D   
   \label{ex:run12b}\\
    & \text{\ref{sh-exists-forall} to (\ref{ex:run12})  and (\ref{ex:run12a})}& \mathfrak{R}_{(\exists S.D)}^{\{\forall R.\neg D\}}  \rightarrow   \mathfrak{R}^\varnothing_{(\exists S.D)}  
    \label{ex:run12b2}\\
    &\text{\ref{sh-functor} to (\ref{ex:run12b2}),   \ref{sh-exists-filler} to } \exists S.D&p_r(\mathfrak{R}_{(\exists S.D)}^{\{\forall S.\neg D\}})  \rightarrow p_r(\mathfrak{R}^\varnothing_{(\exists S.D)}) \rightarrow D 
    \label{ex:run12b3}\\
    & \text {\ref{sh-neg-sqcap} and \ref{sh-trans} to (\ref{ex:run12b})  and  (\ref{ex:run12b3})}&  p_r(\mathfrak{R}_{(\exists S.D)}^{\{\forall S.\neg D\}})  \rightarrow \bot   \label{ex:run14} \\
    & \text {\ref{sh-functor} to (\ref{ex:run14})}& p_l(\mathfrak{R}_{(\exists S.D)}^{\{\forall S.\neg D\}})  \rightarrow \bot  \label{ex:run15} \\
    & \text {\ref{sh-trans} to (\ref{ex:run15})  and (\ref{ex:run12b})}&p_r(\mathfrak{R}_{(\exists S.C)}^{\{\forall S.\neg D\}})  \rightarrow \bot 
     \label{ex:run16} \\
    & \text {\ref{sh-functor} to (\ref{ex:run16})}&p_l(\mathfrak{R}_{(\exists S.C)}^{\{\forall S.\neg D\}})  \rightarrow \bot 
     \label{ex:run17} \\
  & \text{\ref{sh-trans} to (\ref{ex:run17})  and (\ref{ex:run10})}& \{a\} \rightarrow  \bot   \label{ex:run20}
\end{align}
Hence, $\mathcal{O}$ is categorically inconsistent.\hfill$\square$
\end{example} 
%
%
%
%

\section{A new tractable extension $\mathcal{EL}^{\rightarrow}$ of 
$\mathcal{EL}$} \label{sec:newDL}
 In this section, we introduce a novel tractable logic less expressive than $\mathcal{SH}$, namely  $\mathcal{EL}^{\rightarrow}$, but it is  capable of expressing negative knowledge. 
It is known that \ref{sh-dist} can generate a disjunction of    $2^n$ disjuncts from a conjunction such as $(C_1\sqcap D_1)\sqcup \cdots \sqcup  (C_n\sqcap D_n)$. Moreover, 
\ref{sh-exists-forall} or \ref{sh-exists-forall-comp} may lead to adding an exponential number of fresh roles $\mathfrak{R}_{(\exists R.C)}^{\mathcal{C}}$ where $\mathcal{C}=\{\forall R.C_1,\cdots, \forall R.C_n\}$ is an arbitrary set of universal restrictions on $R$ occuring in an ontology. Hence,  it is necessary to drop or weaken \ref{org-dist}, \ref{org-exists-forall} and \ref{org-exists-forall-comp} from the categorical semantics of $\mathcal{SH}$ to obtain tractability.
As indicated in Example~\ref{ex:distrib}, \ref{org-dist} is   independent from  the other properties, and partly responsible for intractability. It is necessary to drop it  if we aim at obtaining a tractable logic. Moreover, \ref{org-exists-forall} and \ref{org-exists-forall-comp} are also responsible for intractability but they involve other constructors  such as role inclusion and transitive role. Hence, instead of dropping them we should replace  \ref{org-exists-forall} and \ref{org-exists-forall-comp} with  their weakened version \ref{org-exists-forall-weak} and \ref{org-exists-forall-comp-weak}  that prevent from  generating an exponential number of roles of the form $\mathfrak{R}^\mathcal{C}_{(\exists r.C)}$, i.e \ref{org-exists-forall-weak} and \ref{org-exists-forall-comp-weak} generate uniquely $\mathfrak{R}^\mathcal{C}_{(\exists R.C)}$ where  $\mathcal{C}$ is singleton.

\begin{enumerate}[leftmargin=!, labelwidth=1em, align=left]
\item[\mylabel{org-exists-forall-weak}{$\mathbb{P}_\exists^{\forall}~$}$:$]  $X\rightarrow \exists P.C, X\rightarrow \forall R.D \in \mathsf{Hom}(\mathscr{C}_c)$, $P\rightarrow R\in \mathsf{Hom}(\mathscr{C}_r)$ 
imply  $\mathfrak{R}^{\forall R.D}_{(\exists P.C)}\rightarrow \mathfrak{R}^{\varnothing}_{(\exists P.C)}\in \mathsf{Hom}(\mathscr{C}_r)$, $    X\leftrightarrows p_l(\mathfrak{R}^{\forall R.D}_{(\exists P.C)})$, $p_r(\mathfrak{R}^{\forall R.D}_{(\exists P.C)})\rightarrow D \in \mathsf{Hom}(\mathscr{C}_c)$.
 
\item[\mylabel{org-exists-forall-comp-weak}{$\mathbb{P}_{\exists\circ}^{\forall}$}$:$]  $X\rightarrow  \exists P.C, X\rightarrow  \forall R.D \in\mathsf{Hom}(\mathscr{C}_c),  P\rightarrow S, S\circ S\rightarrow S, S\rightarrow R \in \mathsf{Hom}(\mathscr{C}_r)$  imply $\mathfrak{R}^{\forall R.D}_{(\exists P.C)}\rightarrow \mathfrak{R}^{\varnothing}_{(\exists P.C)}\in \mathsf{Hom}(\mathscr{C}_r)$, $X \leftrightarrows  p_l(\mathfrak{R}^{\forall R.D}_{(\exists P.C)}), p_r(\mathfrak{R}^{\forall R.D}_{(\exists P.C)}) \rightarrow   \forall S.D \in \mathsf{Hom}(\mathscr{C}_c)$.  
\end{enumerate}

Replacing \ref{org-exists-forall} and \ref{org-exists-forall-comp} with \ref{org-exists-forall-weak} and \ref{org-exists-forall-comp-weak} leads to losing the capacity of detecting  $n$-disjointness such as $\exists R.C\sqcap\forall R.D_1\sqcap \cdots \sqcap \forall R.D_n\rightarrow \bot$ from $C\sqcap D_1\sqcap \cdots \sqcap D_n\rightarrow \bot$ while it allows to maintain the capacity of detecting  binary disjointness such as $\exists R.C\sqcap\forall R.D \rightarrow \bot$ from $C\sqcap D\rightarrow \bot$. As indicated in Example~\ref{ex:intro},  binary disjointness occurs much more frequently than $n$-disjointness with $n>2$ in biomedical ontologies. Furthermore, dropping \ref{org-dist} leads  to losing the capacity of detecting $C\sqcap (D\sqcup E)\rightarrow \bot$ from both $C\sqcap D\rightarrow \bot$ and $C\sqcap E\rightarrow \bot$ while it allows to maintain the capacity of detecting $C\sqcup D\rightarrow \bot$ from $C \rightarrow \bot$ and    $D \rightarrow \bot$, which would be sufficient for some biomedical ontologies as mentioned in Section~\ref{sec:intro}.

By dropping \ref{org-dist} and weakening \ref{org-exists-forall}, \ref{org-exists-forall-comp},  the semantics of    $\sqcup$ and $\forall$ are  weakened. Moreover, since the negation $\neg$ is defined using the property \ref{org-neg-sqcup} with the weakened $\sqcup$, the negation is also weakened. We use $\bar{\sqcup}$,  $\bar{\forall}$ and $\bar{\neg}$ to denote weakened disjunction,  weakened universal restriction and weakened negation respectively.

\begin{definition}[The logic $\mathcal{EL}^{\rightarrow}$]\label{def:ELA} Let $\mathbf{C}$, $\mathbf{R}$  and $\mathbf{I}$ be disjoint nonempty sets of concept names, role names and individuals respectively. Concepts in $\mathcal{EL}^{\rightarrow}$ over  $\mathbf{C}$, $\mathbf{R}$  and $\mathbf{I}$ are defined in the same way as those in $\mathcal{SH}$ with all constructors  except that  $\sqcup$, $\forall$ and $\neg$ are replaced with $\bar{\sqcup}$, $\bar{\forall}$ and $\bar{\neg}$. 
An ontology $\mathcal{O}$ in $\mathcal{EL}^{\rightarrow}$ is a set  of statements  allowed in an ontology   in $\mathcal{SH}$, namely  general concept inclusions (GCI),  role inclusions (RI),  inclusions for transitive role (ITR),  concept assertions, and  role assertions. The semantics of $\mathcal{O}$ is defined  using a set $\pela$ of categorical properties   obtained from  $\psh$ by dropping  \ref{org-dist}, and replacing   \ref{org-exists-forall} and  \ref{org-exists-forall-comp}  with   \ref{org-exists-forall-weak} and  \ref{org-exists-forall-comp-weak} respectively. 
\end{definition}
 
Note that  $\pela$ contains   \ref{org-conj-d}, \ref{org-conj-c}, \ref{org-disj-d}, \ref{org-disj-c}, \ref{org-neg-bot}, \ref{org-neg-top}, \ref{org-neg-sqcap}, \ref{org-neg-sqcup},  \ref{org-forall-weakened} and \ref{org-exists-comp}. Hence, De Morgan's Laws, double negation and the duality of  $\exists$ and $\forall$ remain to hold due to Lemma~\ref{lem:de-morgan}. This allows us to assume that all concepts in an $\mathcal{EL}^{\rightarrow}$ ontology are written in $\mathsf{NNF}$ when designing an algorithm for reasoning.

\begin{theorem}\label{thm:ela-minimality} Let $\mathcal{O}$ 
be an $\mathcal{EL}^\rightarrow$ ontology. There exists a unique minimal concept   category of $\mathcal{O}$ satisfying  all categorical properties in $\pela$. 
\end{theorem}
\begin{proof}  We use $\Gamma_c\tuple{\mathcal{O}}$ to denote the set of all concept categories of $\mathcal{O}$  satisfying  all properties in $\pela$. By definition, each $\mathscr{C}'_c\in \Gamma_c\tuple{\mathcal{O}}$ has a corresponding  role category, denoted $\mathscr{C}'_{c,r}$,  involved in properties from $\pela$. Let $\mathsf{H}_{c}(\mathscr{C}_m)=\displaystyle{\bigcap_{\mathscr{C}'_c\in \Gamma_c\tuple{\mathcal{O}}}\mathsf{Hom}(\mathscr{C}'_c)}$, and  $\mathsf{O}_{c}\tuple{\mathscr{C}_m}$ is the set of all objects involved in all arrows in $\mathsf{H}_{c}(\mathscr{C}_m)$. We can use the same argument to prove Theorem~\ref{thm:cat-minimality} to show that $\mathsf{H}_c(\mathscr{C}_m)$  is the set of arrows of a concept   category $\mathscr{C}_m$ for $\mathcal{O}$  such that  all properties in $\pela$ are satisfied in $\mathscr{C}_m$. Minimality and uniqueness follow from the definition of $\mathsf{H}_c(\mathscr{C}_m)$.
\end{proof}

\begin{definition}[Categorical consistency]\label{def:ELA-consistency} Let $\mathcal{O}$ 
be an $\mathcal{EL}^\rightarrow$ ontology from a set $\mathbf{C}$ of concept names, a set $\mathbf{R}$ of role names,   and a set $\mathbf{I}$ of  individuals.
We use $\mathscr{C}_c\tuple{\mathcal{O}}$ a to denote the minimal concept  category of $\mathcal{O}$ satisfying  all categorical properties in $\pela$.  
We say that $\mathcal{O}$ is (categorically) $\mathcal{EL}^\rightarrow$-inconsistent if 
$\{a\}\rightarrow \bot\in \mathsf{Hom}(\mathscr{C}_c\tuple{\mathcal{O}})$ for some $a\in \mathbf{I}$.  
\end{definition}

To check consistency of an ontology in $\mathcal{EL}^\rightarrow$, we use a set of rules, denoted  $\rela$, introduced in Table~\ref{tab:ELArules}. 
Each rule in  $\rela$ has the same presentation as those in $\rsh$ in Table~\ref{tab:SHrules}, i.e  each rule $\mathbb{R}$ in Table~\ref{tab:ELArules} is composed of three components as follows:
$\mathsf{\mathbb{R}= \displaystyle{\frac{[premise]}{[conclusion]} :  [condition]} }$. 
We can observe that \ref{el-id}, \ref{el-trans}, \ref{el-functor},  \ref{el-axioms}, \ref{el-conj-d}, \ref{el-conj-c}, \ref{el-disj-d}, \ref{el-disj-c}, \ref{el-neg-sqcap}, \ref{el-neg}, \ref{el-exists-filler},  \ref{el-exists} and  \ref{el-exists-comp} are semantically equivalent to \ref{org-syntax-role}, \ref{org-syntax-concept}, \ref{org-syntax-functor},  \ref{org-axioms}, \ref{org-conj-d}, \ref{org-conj-c}, \ref{org-disj-d}, \ref{org-disj-c}, \ref{org-neg-sqcap}, \ref{org-neg-sqcup}, \ref{org-neg-top},  \ref{org-neg-bot}, \ref{org-exists-filler}, \ref{org-exists} and  \ref{org-exists-comp} from $\pela$. It is not needed to use rules to  capture \ref{org-forall-weakened} from $\pela$ since it is supposed that all concepts are in NNF. Moreover, dropping \ref{org-dist} from the semantics of $\mathcal{EL}^{\rightarrow}$ leads to using explicitly \ref{el-neg-sqcap} and \ref{el-neg-sqcup} for reasoning in $\mathcal{EL}^\rightarrow$. In fact, \ref{el-neg-sqcap} and  \ref{el-neg-sqcup}  can be obtained from \ref{org-dist} with  \ref{org-neg-top}, \ref{org-conj-d} and \ref{org-syntax-concept}. This explains why it is not needed to capture  \ref{org-neg-sqcap} and  \ref{org-neg-sqcup} by rules  in Table~\ref{tab:SHrules} due to the presence of   \ref{sh-dist} in Table~\ref{tab:SHrules}. The rules in $\rela$ aim  to build from an $\mathcal{EL}^\rightarrow$ ontology $\mathcal{O}$ a set $\mathsf{Hom}\tuple{\mathcal{O}}$ of arrows and a set $\mathsf{Ob}\tuple{\mathcal{O}}$ of objects such that $C\in \mathsf{Ob}\tuple{\mathcal{O}}$ iff there is an arrow $C\rightarrow D \in \mathsf{Hom}\tuple{\mathcal{O}}$ or $D\rightarrow C \in \mathsf{Hom}\tuple{\mathcal{O}}$.

\begin{table}[H]
  \centering    
\hrule
  \vspace{0.2cm}
\begin{enumerate}[leftmargin=!, labelwidth=1em, align=left]
\item[\mylabel{el-id}{$\mathsf{\mathbb{R}_{id}}$}$:$] 
 
$\displaystyle{\frac{}{X \rightarrow X ~~~X \rightarrow \top~~~\bot \rightarrow X}}$

\item[\mylabel{el-trans}{$\mathsf{\mathbb{R}_{tr}}$}$:$] 
 
$\displaystyle{\frac{X \rightarrow Y~~~Y \rightarrow Z}{X \rightarrow Z} }$

\item[\mylabel{el-functor}{$\mathbb{R}_\mathsf{f}~$}$:$]

$\displaystyle{\frac{R \rightarrow S}{p_l(R) \rightarrow p_l(S) ~~~p_r(R) \rightarrow p_r(S)} }~~~$
$\displaystyle{\frac{p_r(R) \rightarrow \bot}{p_l(R) \rightarrow \bot} }$
 
\item[\mylabel{el-axioms}{$\mathbb{R_\sqsubseteq}$}$:$]
 
 $\displaystyle{\frac{}{X \rightarrow Y   } :  X\sqsubseteq Y \in \mathcal{O} }$
 
\item[\mylabel{el-conj-d}{$\mathbb{R}_\sqcap^\mathsf{d}$}$:$]

$\displaystyle{\frac{}{C \sqcap D \rightarrow C~~~ C \sqcap D \rightarrow  D } : C\sqcap D\in  \mathsf{Ob}\langle \mathcal{O}\rangle}$

\item[\mylabel{el-conj-c}{$\mathbb{R}_\sqcap^\mathsf{c}$}$:$]

$\displaystyle{\frac{X\rightarrow C ~~~ X\rightarrow D  }{X  \rightarrow  C \sqcap D} : C\sqcap D\in  \mathsf{Ob}\langle \mathcal{O}\rangle }$

\item[\mylabel{el-disj-d}{$\mathbb{R}^\mathsf{d}_\sqcup$}$:$]

$\displaystyle{\frac{C\bar{\sqcup} D\rightarrow X}{C \rightarrow  X ~~~D\rightarrow X  } }$
 
\item[\mylabel{el-disj-c}{$\mathbb{R}^\mathsf{c}_\sqcup$}$:$]

$\displaystyle{\frac{C \rightarrow X ~~~ D \rightarrow X }{C\bar{\sqcup} D \rightarrow  X} : C\bar{\sqcup} D\in  \mathsf{Ob}\langle \mathcal{O}\rangle}$

\item[\mylabel{el-neg-sqcap}{$\mathsf{\mathbb{R}}_{\neg}^\sqcap$}$:$] 
$\displaystyle{\frac{X  \sqcap C \rightarrow \bot  }{X \rightarrow  \mathsf{NNF}(\bar{\neg} C)}  }$

\end{enumerate}
\end{table}
  
\pagebreak
\begin{table}[H]
  \centering   

\begin{enumerate}[leftmargin=!, labelwidth=1em, align=left]

\item[\mylabel{el-neg-sqcup}{$\mathsf{\mathbb{R}}_{\neg}^\sqcup$}$:$] 
$\displaystyle{\frac{ \top \rightarrow X  \sqcup C   }{\mathsf{NNF}(\bar{\neg} X) \rightarrow  C}  }$

\item[\mylabel{el-neg}{$\mathsf{\mathbb{R}}_{\neg}$}$:$] 
$\displaystyle{\frac{     }{C\sqcap \mathsf{NNF}(\bar{\neg} C)  \rightarrow  \bot ~~~\top   \rightarrow  C\sqcup \mathsf{NNF}(\bar{\neg} C)}  }$

\item[\mylabel{el-exists-filler}{$\mathsf{\mathbb{R}}_\exists^\mathfrak{R}$}$:$]

$\displaystyle{\frac{ }{ \mathfrak{R}^\varnothing_{(\exists R.C)}\rightarrow R~~~ \exists R.C\leftrightarrows p_l(\mathfrak{R}^\varnothing_{(\exists R.C)}) ~~~p_r(\mathfrak{R}^\varnothing_{(\exists R.C)})\rightarrow C} : \exists R.C\in  \mathsf{Ob}\langle \mathcal{O}\rangle}$

\item[\mylabel{el-exists}{$\mathbb{\mathbb{R}}_\exists$}$:$]

$\displaystyle{\frac{R'\rightarrow R ~~~ p_r(R') \rightarrow C}{ p_l(R') \rightarrow p_l(\mathfrak{R}^\varnothing_{(\exists R.C)})}    :  \exists R.C\in  \mathsf{Ob}(\mathscr{C}_c\langle \mathcal{O}\rangle)}$

\item[\mylabel{el-exists-comp}{$\mathbb{\mathbb{R}}_\exists^\circ$}$:$]

$\displaystyle{\frac{X\rightarrow \exists S.C ~~~ C \rightarrow \exists S.D~~~S\circ S\rightarrow S}{ X\rightarrow \exists S.D} }$

\item[\mylabel{el-exists-forall}{$\mathsf{\mathbb{R}}_{\exists}^{\forall}$}$:$] 
$\displaystyle{\frac{X\rightarrow  \exists P.C ~~~  X\rightarrow \bar{\forall} R.D~~~  P \rightarrow  R}{  \mathfrak{R}^{\bar{\forall} R.D}_{(\exists P.C)}\rightarrow \mathfrak{R}^\varnothing_{(\exists P.C)} ~~~ X \leftrightarrows   p_l(\mathfrak{R}^{\bar{\forall} R.D}_{(\exists P.C)})~~~p_r(\mathfrak{R}^{\bar{\forall} R.D}_{(\exists P.C)})\rightarrow D   } }$

\item[\mylabel{el-exists-forall-ind}{$\mathsf{\mathbb{R}}_\exists^{\forall \mathsf{I}}$}$:$]

$\displaystyle{\frac{\{(a,b)\}\rightarrow P ~~~ \{a\}\rightarrow \bar{\forall} R.D ~~~P \rightarrow  R}{ \{b\}\rightarrow D}   }$

\item[\mylabel{el-exists-forall-comp}{$\mathsf{\mathbb{R}}_{\exists \circ}^{\forall}$}$:$]
$\displaystyle{\frac{ X\rightarrow  \exists P.C ~~~  X\rightarrow \bar{\forall} R.D~~~  P \rightarrow  S~~~S\circ S \rightarrow  S~~~S\rightarrow R}{  \mathfrak{R}^{\bar{\forall} R.D}_{(\exists P.C)}\rightarrow \mathfrak{R}^\varnothing_{(\exists P.C)} ~~~ X \leftrightarrows   p_l(\mathfrak{R}^{\bar{\forall} R.D}_{(\exists P.C)})~~~p_r(\mathfrak{R}^{\bar{\forall} R.D}_{(\exists P.C)})\rightarrow  \bar{\forall} S.D} }$

\item[\mylabel{el-exists-forall-comp-ind}{$\mathsf{\mathbb{R}}_{\exists \circ}^{\forall\mathsf{I}}$}$:$]  
$\displaystyle{\frac{\{(a,b)\}\rightarrow P ~~~ \{a\}\rightarrow \bar{\forall} R.D ~~~P \rightarrow  S~~~ S\circ S\rightarrow S~~~ S\rightarrow R}{ \{b\}\rightarrow \bar{\forall} S.D}   }$
 
\end{enumerate}
\hrule
\vspace{0.2cm}
 \caption{Categorical Rules $\rela$ for $\mathcal{EL}^\rightarrow$}
  \label{tab:ELArules}
\end{table}

\begin{theorem}\label{thm:el-consistency}

Let $\mathcal{O}$  be  an $\mathcal{EL}^\rightarrow$ ontology. Let  $\mathsf{Hom}\tuple{\mathcal{O}}$ be the set of arrows built from $\mathcal{O}$ by applying  rules  $\rela$  in  Table~\ref{tab:ELArules} until no rule is applicable. It holds that 
\begin{enumerate}[leftmargin=0.1cm, itemindent=0.5cm] 
 
\item\label{lem:ela-consistency:2} $\mathcal{O}$  is (categorically) $\mathcal{EL}^\rightarrow$-inconsistent iff  $\mathsf{Hom}\tuple{\mathcal{O}}$ contains   $\{a\}\rightarrow \bot$ for some individual $a$. 

\item\label{lem:ela-consistency:1} The cardinality of $\mathsf{Hom}\tuple{\mathcal{O}}$  is bounded by a function polynomial  in the size of $\mathcal{O}$ up to object isomorphism. 
\end{enumerate}
 
\end{theorem}

\begin{proof}
\begin{enumerate} [leftmargin=0.1cm, itemindent=0.5cm] 

\item  It follows straightforwardly from the definition of $\pela$ and $\rela$, i.e, one of rules in  $\rela$  adds an arrow $C\rightarrow D$ to $\mathsf{Hom}\tuple{\mathcal{O}}$ iff there is a property from $\pela$ that defines  $C\rightarrow D$ in all concept categories. 



\item We show Item~\ref{lem:sh-consistency:2}. Since there are at most two arrows between 2 objects, and a rule is not applicable if the arrows in the conclusion of the rule are already in  $\mathsf{Hom}\tuple{\mathcal{O}}$, it suffices to compute the number of objects added by rules to determine the cardinality of $\mathsf{Hom}\tuple{\mathcal{O}}$.
If an ontology $\mathcal{O}$ can be seen as a string of length $k$, then  $\lvert \mathcal{O}\rvert=k$ where $\lvert \mathcal{O}\rvert$   denotes the size of $\mathcal{O}$. Hence, the number of concept objects from $\mathcal{O}$ is bounded by  $O(k^2)$. We use  $\mathbf{C}_0$, $\mathbf{R}_0$ and $\mathbf{I}_0$  to denote respectively the sets of concept names, role names and individuals occurring in  $\mathcal{O}$.  We determine  extensions $\mathbf{C}, \mathbf{R}$ and $\mathbf{I}$ from  $\mathbf{C}_0$, $\mathbf{R}_0$ and $\mathbf{I}_0$   as follows: (i)   $\mathbf{C}, \mathbf{R}, \mathbf{I}$ are the smallest sets such that $\mathbf{C}_0\subseteq \mathbf{C}$, $\mathbf{R}_0\subseteq \mathbf{R}$ and $\mathbf{I}_0\subseteq \mathbf{I}$, and (ii) for each concept of the form $\exists R.C$ occurring in $\mathcal{O}$, $\mathbf{R}$ contains a role name $\mathfrak{R}_{(\exists R.C)}$ and $\mathbf{C}$ contains  two concept names $p_l(\mathfrak{R}_{(\exists R.C)})$ and $p_r(\mathfrak{R}_{(\exists R.C)})$;  (iii) and each pair of concepts of the form  $\exists R.C$ and  $\forall R.D$ occurring in $\mathcal{O}$, $\mathbf{R}$ contains a role name $\mathfrak{R}^{\bar{\forall} R.D}_{(\exists R.C)}$, and  $\mathbf{C}$ contains  two concept names $p_l(\mathfrak{R}^{\bar{\forall} R.D}_{(\exists R.C)})$ and $p_r(\mathfrak{R}^{\bar{\forall} R.D}_{(\exists R.C)})$. We consider the rules involved in generating of these fresh role objects.

\begin{itemize}[leftmargin=0.1cm, itemindent=0.2cm]
\item It holds that the number of concepts   of the form $\exists R.C$ and  $\bar{\forall} R.C$ occurring in $\mathcal{O}$ is bounded by $k$. If a rule from $\rela$ adds an object of the form $\exists R.C$ to  $\mathsf{Ob}\tuple{\mathcal{O}}$, then $\exists R.C$ must occur in $\mathcal{O}$. This implies that    the number of concept objects   of the form $\exists R.C$ in $\mathsf{Ob}\tuple{\mathcal{O}}$ is bounded by $k^2$. Hence,  the number of  fresh role objects   of the form $\mathfrak{R}^\varnothing_{\exists R.C}$ added by \ref{el-exists-filler} to $\mathsf{Ob}\tuple{\mathcal{O}}$ is bounded by $k^2$. 

\item Two rules  \ref{el-exists-forall-comp}, \ref{el-exists-forall-comp-ind} can add fresh  concept objects of the form $\bar{\forall} S.D$ if there is a concept $\bar{\forall} R.D$ and a transitive role $S$ occur in $\mathcal{O}$. This implies that    the number of objects   of the form $\forall S.D$ in $\mathsf{Ob}\tuple{\mathcal{O}}$ is bounded by $k^2$ where $k$ bounds the number of roles $S\in \mathbf{R}_0$, and the number of concepts of the form $\forall R.D$ occurring in $\mathcal{O}$.

\item Two rules \ref{el-exists-forall}, \ref{el-exists-forall-comp} can add fresh  role objects   of the form $\mathfrak{R}_{\exists R.C}$ and $\mathfrak{R}^{\bar{\forall} R.D}_{\exists R.C}$. This implies that the number of role objects   of the form $\mathfrak{R}_{\exists R.C}$ and $\mathfrak{R}^{\bar{\forall} R.D}_{\exists R.C}$ in $\mathsf{Ob}\tuple{\mathcal{O}}$ is bounded by $2k^{4}$. Therefore, $\lvert \mathbf{R}\rvert \leqslant O(3k^{4})$.

\item The rule \ref{el-functor} adds two fresh concept objects $p_l(R)$ and $p_r(R)$ for each role object $R$. Therefore, $\lvert \mathbf{C}\rvert \leqslant O(6k^{4})$.
 
\item \ref{el-id} and \ref{el-trans} add no new object.

\item \ref{el-axioms} adds  concept objects $X, Y$ for each axiom $X\sqsubseteq Y\in \mathcal{O}$.  Therefore, the number of concept objects added by this rule is bounded by $k^2$.

\item All rules \ref{el-conj-d} to \ref{el-disj-c}, \ref{el-exists}  and \ref{el-exists-comp}   add no new object.

\item \ref{el-neg-sqcup} and \ref{el-neg-sqcap} add $\mathsf{NNF}(\bar{\neg} C)$ for each concept object $C$, and  \ref{el-neg} adds $\mathsf{NNF}(\bar{\neg} C)\sqcup C$ and $\mathsf{NNF}(\bar{\neg} C)\sqcap C$  of each concept object $C$. Since $\mathsf{NNF}(\bar{\neg} (\mathsf{NNF}(\bar{\neg} C)\sqcup C))$  and $\mathsf{NNF}(\neg (\mathsf{NNF}(\bar{\neg} C)\sqcap C))$ are  isomorphic to  $\mathsf{NNF}(\bar{\neg} C)\sqcap C$ and   $\mathsf{NNF}(\bar{\neg} C)\sqcup C$,  it holds that \ref{el-neg-sqcup}, \ref{el-neg-sqcap} and \ref{el-neg} add to $\mathsf{Ob}\tuple{\mathcal{O}}$ at most $3m$ new concept objects where $m$ is the number of objects added to $\mathsf{Ob}\tuple{\mathcal{O}}$ by the other rules up to object isomorphism.

\end{itemize}

We have proved that each rule from $\pela$ adds a polynomial number of objects in the size of $\mathcal{O}$.  Therefore, the number of objects added by all rules in $\pela$ is bounded by a polynomial function in the size of $\mathcal{O}$, and thus the cardinality of $\mathsf{Hom}\tuple{\mathcal{O}}$   is bounded by a polynomial function   in the size of $\mathcal{O}$ up to object isomorphism.
\end{enumerate}
\end{proof}

\begin{corollary}\label{cor:SH-ELA}  Let  $\mathcal{O}$ be an $\mathcal{SH}$ ontology, and  $\mathcal{O}'$ the  $\mathcal{EL}^\rightarrow$ ontology obtained from $\mathcal{O}$ by replacing $\forall,\sqcup$ and $\neg$ with $\bar{\forall},\bar{\sqcup}$ and $\bar{\neg}$ respectively. It holds that if $\mathcal{O}'$  is  inconsistent then $\mathcal{O}$ is (categorically) $\mathcal{EL}^\rightarrow$-inconsistent.
\end{corollary}

Corollary~\ref{cor:SH-ELA} follows straightforwardly from $\psh$ and $\pela$ since  $\pela$ is weaker than $\psh$, i.e, the minimal  concept category satisfying $\pela$ is a subcategory of the minimal  concept category satisfying $\psh$. 
 
It is important to note that although $\bar{\forall}$ (resp. $\bar{\sqcup}$) is weaker than $\forall$ (resp.   $\sqcup$)  in the sense that the categorical properties describing $\bar{\forall}$ are weaker than  those describing $\forall$, a logic $\mathcal{L}'$ containing $\bar{\forall}$  is not a sublogic of $\mathcal{L}$ obtained from $\mathcal{L}'$ by replacing    $\bar{\forall}$ with $\forall$. By ``sublogic" we mean that a concept $A_0$  is satisfiable with respect to an ontology $\mathcal{O}$ expressed in  $\mathcal{L}'$ iff  $A_0$  is satisfiable   with respect to  $\mathcal{O}$ expressed in  $\mathcal{L}$.  The following example confirms this observation.   
 
\begin{example}\label{ex:weakened-forall}
We take $\mathcal{L}=\mathcal{SH}$ and define a logic $\mathcal{L}_1$ that is obtained from $\mathcal{SH}$ by replacing    $\forall$ with  $\bar{\forall}$. Let us consider an ontology  $\mathcal{O}_1$ expressed in $\mathcal{L}_1$ including  

$A_0\rightarrow \exists R.C_1 \sqcap \bar{\forall} R.C_2 \sqcap \bar{\forall} R.C_3$, 
$C_1 \sqcap C_2 \sqcap C_3\rightarrow \bot$ 
 
\noindent $A_0$ is (categorically) $\mathcal{EL}^\rightarrow$-satisfiable with respect to  $\mathcal{O}_1$ (i.e., $A_0\rightarrow \bot$ does not belong to the minimal concept category of $\mathcal{O}_1$) while $A_0$ is not set-theoretically satisfiable with respect to an ontology $\mathcal{O}_2$ expressed in $\mathcal{L}$ including  


$A_0\sqsubseteq \exists R.C_1 \sqcap \forall R.C_2 \sqcap \forall R.C_3$, 
$C_1 \sqcap C_2 \sqcap C_3\sqsubseteq  \bot$  

However, if a concept $A$ is categorically unsatisfiable with respect to an ontology $\mathcal{O}$ expressed in  $\mathcal{L}_1$, then $A$ is also (set-theoretically and categorically) unsatisfiable with respect to   $\mathcal{O}'$ obtained from $\mathcal{O}$ by replacing $\bar{\forall}$ with $\forall$. This is due to Corollary~\ref{cor:SH-ELA}.
\end{example}


\noindent 

\begin{example}\label{ex:weakened-sqcup}
\noindent We now define $\mathcal{L}_2$ that is obtained from $\mathcal{L}=\mathcal{SH}$ by replacing $\sqcup$ with $\bar{\sqcup}$. Let us consider an ontology  $\mathcal{O}_3$ expressed in $\mathcal{L}_2$ including

$A_0\sqsubseteq C_1 \sqcap (C_2 \bar{\sqcup} C_3), C_1 \sqcap  C_2  \sqsubseteq  \bot$,   $C_1 \sqcap  C_3  \sqsubseteq  \bot$  

\noindent It holds that $A_0$ is categorically satisfiable with respect to  $\mathcal{O}_3$ while $A_0$ is  unsatisfiable with respect to $\mathcal{O}_4$ expressed in $\mathcal{L}$ including
$A_0\sqsubseteq C_1 \sqcap (C_2 \sqcup C_3), C_1 \sqcap  C_2  \sqsubseteq  \bot$, $C_1 \sqcap  C_3  \sqsubseteq  \bot$.
\noindent Unsatisfiability of $A_0$ with respect to $\mathcal{O}_4$ is due to \ref{org-dist} and \ref{org-disj-c}. 
\end{example}

\section{Independence of $\bar{\forall}, \bar{\sqcup}$ and $\bar{\neg}$}

We now prove a result that affirms that the weakened constructors $\bar{\forall}, \bar{\sqcup}$ and $\bar{\neg}$ are independent from the others, i.e., we obtain $\mathcal{EL}_\bot^\circ$  from $\mathcal{EL}^{\rightarrow}$ by dropping  from $\pela$ all categorical properties related to $\bar{\forall}, \bar{\sqcup}$ and $\bar{\neg}$, namely   \ref{org-disj-d}, \ref{org-disj-c}, \ref{org-neg-sqcap}, \ref{org-neg-sqcup},  \ref{org-neg-bot}, \ref{org-neg-top},  \ref{org-forall-weakened}, \ref{org-dist},  \ref{org-exists-forall}, \ref{org-exists-forall-ind}, \ref{org-exists-forall-comp} and \ref{org-exists-forall-comp-ind}. Therefore,  $\pelc=\{\text{\ref{org-syntax-role}, \ref{org-syntax-concept},\ref{org-syntax-functor}, \ref{org-axioms},
\ref{org-conj-d}, \ref{org-conj-c},  \ref{org-exists-filler}, \ref{org-exists}, \ref{org-exists-comp}}\}$ defines the categorical semantics of $\mathcal{EL}_\bot^\circ$.
 
\begin{theorem}
\label{thm:NEWDL}
  Let $\mathcal{O}$ be an ontology in $\mathcal{EL}_\bot^\circ$. It holds that $\mathcal{O}$ is set-theoretically inconsistent iff  $\mathcal{O}$ is categorically  $\mathcal{EL}^\circ_\bot$-inconsistent.   
\end{theorem}
 
\begin{proof}  
Assume $\mathcal{O}$ is categorically inconsistent, i.e $\{a\}\rightarrow \bot\in \mathsf{Hom}(\mathscr{C}_c\tuple{\mathcal{O}})$ for some individual $a$ where $\mathscr{C}_c\tuple{\mathcal{O}}$ is an  NNF-saturated concept category of $\mathcal{O}$.
For each arrow $X\rightarrow Y$ added to $\mathsf{Hom}(\mathscr{C}_c\tuple{\mathcal{O}})$ by a categorical rule for $\mathcal{EL}^\rightarrow$, 
we have  $\mathcal{O}\models X\sqsubseteq Y$ by Theorem~\ref{thm:cat-set}. This implies that $\mathcal{O}\models \{a\}\sqsubseteq \bot$.

For the converse direction, we  define a structure, namely  universal tableau for $\mathcal{EL}^\circ_\bot$, to characterize  set-theoretical consistency of $\mathcal{O}$. First, we define the smallest set of sub-concepts $\mathsf{sub}(\mathcal{O})$ such that  
\begin{itemize}[leftmargin=0.1cm, itemindent=0.2cm]
    \item If there is  an individual $a$ that occurs in   $\mathcal{O}$, then $\{a\}\in \mathsf{sub}(\mathcal{O})$
    \item If $C\sqsubseteq D\in  \mathcal{O}$ is a GCI or CAA, then  $C,D\in \mathsf{sub}(\mathcal{O})$.
    \item If $C \sqcap D\in  \mathsf{sub}(\mathcal{O})$, then  $C,D\in \mathsf{sub}(\mathcal{O})$.
    \item If $\exists R.C\in  \mathsf{sub}(\mathcal{O})$, then  $C\in \mathsf{sub}(\mathcal{O})$.
\end{itemize}

Given an  $\mathcal{EL}^\circ_\bot$ ontology $\mathcal{O}$,  a tableau for  $\mathcal{O}$ is a graph $\tuple{V_0,\cdots,V_n,E,L}$   where $n$ is minimal, $V_i$ is the smallest set of nodes, and  $L$ a labeling function which maps each node $x\in V_i $ ($0\leqslant i\leqslant n$) to the smallest set of sub-concepts in $\mathcal{P}(\mathsf{sub}(\mathcal{O}))$    such that the following conditions are satisfied.


\begin{itemize}[leftmargin=0cm, itemindent=0.7cm]
    \item[\mylabel{el-cano-ind}{$\mathsf{T}_{id}$}:] If $a$ is an  individual occurring in $\mathcal{O}$, then  $x^{a}\in V_0$ with $\{a\}\in L(x^a)$.  If $\{(a,b)\}\sqsubseteq R\in \mathcal{O}$, then $x^b$ is called an $R$-root successor of $x^a$ and $(x^a,x^b)\in E$. Furthermore, if  $R\overlay{\sqsubseteq}{\ast} S\in \mathcal{R}^+\tuple{\mathcal{O}}$, then $x^b$ is also called an $S$-root successor of $x^a$.

    \item[\mylabel{el-cano-axioms}{$\mathsf{T}_\sqsubseteq$}:] If $C_1,\cdots , C_k\in L(x)$ and  $C_1\sqcap \cdots \sqcap  C_k \sqsubseteq D\in \mathcal{O}$, then $D\in L(x)$.
    
    \item[\mylabel{el-cano-exists-plus}{$\mathsf{T}_\exists^+$}:] If $x$  has an $S$-successor $y$ with $D\in L(y)$, $\exists S.D\in \mathsf{sub}(\mathcal{O})$ and  $\exists S.D\notin L(x)$, then  $\exists S.D\in L(x)$.
    
    \item[\mylabel{el-cano-exists-comp}{$\mathsf{T}_\exists^\circ~$}:] If  $x$ has an $S$-successor $y$ with  $\exists S.E\in  L(y)$,  $S\circ S\overlay{\sqsubseteq}{\ast} S\in \mathcal{R}^+\tuple{\mathcal{O}}$ and $\exists S.E\notin L(x)$, then   $\exists S.E\in L(x)$.

    \item[\mylabel{el-cano-exists}{$\mathsf{T}_\exists~$}:] If $\exists R.C\in L(x)$, $x\in V_i$ with $i>0$  and $x$ is not blocked, i.e $x$ has no ancestor  $x'$ such that $L(x)\subseteq L(x')$, 
    then $x$  has an  $R$-successor $y\in V_{i+1}$  with $C\in L(y)$ and $(x,y)\in E$. In this case, if $R\overlay{\sqsubseteq}{\ast} S\in \mathcal{R}^+\tuple{\mathcal{O}}$ then $y$ is also called an $S$-successor of $x$.
    
    \item[\mylabel{el-cano-sqcap}{$\mathsf{T}_\sqcap$}:] If $C\sqcap D\in L(x)$   and $\{C,D\}\nsubseteq  L(x)$, then  $C,D\subseteq L(x)$.
    
 \end{itemize}

 An $R$-path $(x,y)$ over $\tuple{V_0,\cdots,V_n, L}$ is inductively defined as follows: (i) if  $y$ is an $R$-successor of $x$, or $y$ blocks an $R$-successor $x'$ of $x$, then $(x,y)$ is an $R$-path;  (ii) if $(z,y)$ is an $R$-path, $z$ is an $R$-successor of $x$, or $z$ blocks an $R$-successor $x'$ of $x$, and $R\circ R\overlay{\sqsubseteq}{\ast} R \in \mathcal{R}^+\tuple{\mathcal{O}}$, then $(x,y)$ is an $R$-path; (iii)  if  $(x,y)$ is an $R$-path and $R \overlay{\sqsubseteq}{\ast} S \in \mathcal{R}^+\tuple{\mathcal{O}}$, then $(x,y)$ is also an $S$-path. We now define an interpretation $\mathcal{I}_0$ from $\tuple{V_0,\cdots,V_n,L}$ as follows:

\noindent $\Delta^{\mathcal{I}_0}:=\bigcup_{0\leqslant i\leqslant n}V_i\setminus \{\text{blocked nodes}\}$

\noindent $a^{\mathcal{I}_0}=x^a$ for all individual $a$
 
\noindent $A^{\Tb{I}_0} := \{x\in\Delta^{\Tb{I}_0} \mid  A\in L(x)\}$

\noindent $R^{\Tb{I}_0} := \{\tuple{x, y} \mid (x,y) \text{ is an R-path}\}$
  
We say that a tableau $\tuple{V_0,\cdots,V_n,L}$ has a clash if $\bot\in L(v)$ for some $v\in V_i$, and clash-free otherwise.  

\begin{claim}\label{claim:el-tab} 
Assume that $\tuple{V,L}$ is clash-free. It holds that    
$D_1,\cdots, D_k\in L(x)$ iff $x\in  (D_1\sqcap \cdots \sqcap D_k)^{\mathcal{I}_0}$ where $D_i$ is not a conjunction.
\end{claim}
\begin{proof}  
Case $k=1$.
We have $D_i\in L(x)$ implies  $D_i\neq \bot$ for all $x\in V$ since $\tuple{V,L}$ is clash-free.  We proceed by induction on the structure of $D_1$.

\noindent The ``$\Longrightarrow$" direction.

\noindent$\bullet$   Case  $D_1=A$. We have $x\in A^{\mathcal{I}_0}$ by the definition of $A^{\mathcal{I}_0}$.
    
\noindent$\bullet$   Case   $D_1=\exists R.E$.  By \ref{el-cano-exists},  $x$ has an $R$-successor $y$, or there is some $y$ that blocks an $R$-successor $x'$ of $x$ with $E\in L(x')$ and $L(x')\subseteq L(y)$. By induction, $y\in E^{\mathcal{I}_0}$. By the definition of $R^{\mathcal{I}_0}$, $(x, y)\in R^{\mathcal{I}_0}$. Hence, $x\in \exists R.E^{\mathcal{I}_0}$.

\noindent The ``$\Longleftarrow$" direction.

\noindent$\bullet$      Case $D_1=A$. Let $x\in  A^{\mathcal{I}_0}$. We have $A\in L(x)$ by the definition of  $A^{\mathcal{I}_0}$.
    
\noindent$\bullet$      Case $D_1=\exists R.E$. Let $x\in \exists R.E^{\mathcal{I}_0}$. By definition of $\mathcal{I}_0$ there is an $R$-path $(x,y)$ such that $y\in E^\mathcal{I}$. By induction on the structure of $D_1$, $E\in L(y)$. There are the following possibilities: (i) $x$ has an $R$-successor $y$, or there is some $y$ that  blocks an $R$-successor $x'$ of $x$ with $E\in L(x')\subseteq L(y)$. By \ref{el-cano-exists-plus},  $\exists R.E\in L(x)$; (ii)
$(z,y)$ is an $R$-path with $y\in E^{\mathcal{I}_0}$ and $x$  has an $R$-successor $z$, or there is some $z$ that  blocks an $R$-successor $x'$ of $x$ and $R\circ R\overlay{\sqsubseteq}{\ast} R \in \mathcal{R}^+\tuple{\mathcal{O}}$. We have $E\in   L(y)$ by induction on the structure of $D_1$. Moreover, by the induction hypothesis on the length of the path, we have $\exists R.E\in L(z)$ since $(z,y)$ is also an $R$-path. By \ref{el-cano-exists-comp}, we obtain $\exists R.E \in L(x)$. 

Case $k=n$. Assume that $D_1,\cdots, D_n\in L(x)$. By induction on $k$, $x\in D_i^{\mathcal{I}_0}$  for $i\in \{1..n\}$. Hence, $x\in (D_1\sqcap\cdots   \sqcap D_{n})^{\mathcal{I}_0}$. Conversely, assume that    $x\in (D_1\sqcap \cdots \sqcap   D_{n})^{\mathcal{I}_0}$. This implies that  $x\in (D_1\sqcap \cdots \sqcap D_{n-1})  ^{\mathcal{I}_0}$ and  $x\in D_{n}^{\mathcal{I}_0}$. By induction, $D_1\sqcap \cdots \sqcap D_{n-1} \in L(x)$ and $D_n\in L(x)$. By \ref{el-cano-sqcap}, $D_i\in L(x)$ for $i\in \{1..n\}$. 
\end{proof}

\begin{claim}\label{claim:el-mod} If  $\tuple{V_0,\cdots, V_n,L}$ is clash-free, then  $\mathcal{I}_0$ is a model of $\mathcal{O}$. 
\end{claim}

\begin{proof}  For this, we prove$\mathcal{I}_0\models \alpha$   for each $\alpha\in \mathcal{O}$. 
 
\begin{itemize}[leftmargin=0.1cm, itemindent=0.2cm]
    \item $\alpha = \{a\}\sqsubseteq C$. We have $C\neq \bot$ since otherwise,   $  \bot\in  L(x^a)$, which is a contradiction. By   \ref{el-cano-ind}, $C\in L(x^a)$. By Claim~\ref{claim:el-tab}, $x^a\in C^{\mathcal{I}_0}$.
    
    \item $\alpha = \{(a,b)\}\sqsubseteq R$. By  \ref{el-cano-ind}, $x^b$ is a $R$-successor of $x^a$. By the definition of $\mathcal{I}_0$ we have  $(x^a,x^b)\in R^{\mathcal{I}_0}$. 
    
    \item $\alpha = D\sqsubseteq E$.   Let $x\in D^{\Tb{I}_0}$. If $D$ is not a conjunction, then 
      $D\in L(x)$  by Claim~\ref{claim:el-tab}. Due to    \ref{el-cano-axioms}, we have  $E\in L(x)$. By Claim~\ref{claim:el-tab}, $x\in E^{\mathcal{I}_0}$. Assume $D=D_1\sqcap \cdots \sqcap D_k$. Without loss of the generality, we can assume than $D_i$ is not a conjunction since we can decompose $D_i$ into conjuncts is necessary. We have $D_i\in L(x)$  by Claim~\ref{claim:el-tab} for all $i\in \{1..k\}$. By \ref{el-cano-axioms}, $E\in L(x)$.  By Claim~\ref{claim:el-tab} again, $x\in E^{\mathcal{I}_0}$.
 
    \item $\alpha = R\sqsubseteq S$. Let $\langle x, y\rangle
    \in R^{\Tb{I}_0}$. By the definition  of $R^{\mathcal{I}_0}$, we have $\langle x, y\rangle$ is an $R$-path. By definition, $\langle x, y\rangle$ is also an $S$-path.  By the definition of $\mathcal{I}_0$, $\langle x, y\rangle
    \in S^{\Tb{I}_0}$.

    \item $\alpha = S\circ S\sqsubseteq S$. Let $\tuple{x, 
    z}, \tuple{z, 
    y}\in S^{\Tb{I}_0}$.
    According to the definition of $\Tb{I}_0$, $\tuple{x, 
    z}$ and $\tuple{z, 
    y}$  are $S$-paths.  By the definition of paths, $\tuple{x, 
    y}$ is an $S$-path.
\end{itemize}
\end{proof}

\begin{claim}\label{claim:el-tab-cat}  
If a concept $X$ is added to $L(x_C)$ by a tableau rule, namely  \ref{el-cano-ind}, \ref{el-cano-axioms}, \ref{el-cano-exists-plus}, \ref{el-cano-exists-comp}, \ref{el-cano-exists}, then $C\rightarrow X\in \mathsf{Hom}(\mathscr{C}_c\tuple{\mathcal{O}})$ where $x_C=\{a\}$ for some individual  $a$, or  $x_C$ denotes a node  added by applying \ref{el-cano-exists} to some $\exists R.C\in L(y)$ and $y$ is the predecessor of $x_C$.
\end{claim}

\begin{proof} We have  $C\in L(x_C)$. Assume that there is a sequence $S$ of applications of tableau rules to build the tableau $\tuple{V_0,\cdots, V_n, L}$.
We use $x_C$ to denote a node  added by applying \ref{el-cano-exists} to some $\exists R.C\in L(y)$ and $y$ is the predecessor of $x_C$, or  $x_C=\{a\}$ for some individual  $a$. This implies that $C\in L(x_C)$.

\begin{itemize}[leftmargin=0.1cm, itemindent=0.2cm]
\item \ref{el-cano-ind} ensures that $X\in L(x^a)$. We have $\{a\}\rightarrow X\in \mathsf{Hom}(\mathscr{C}_c\tuple{\mathcal{O}})$  due to \ref{org-axioms}.

\item  \ref{el-cano-axioms} ensures that   $X\in L(x_C)$  with $D_i\in L(x_C)$ and $D_1\sqcap \cdots \sqcap D_n\sqsubseteq  X\in \mathcal{O}$. By \ref{org-axioms}, we have $D_1\sqcap \cdots \sqcap D_n\rightarrow  X \in \mathsf{Hom}(\mathscr{C}_c\tuple{\mathcal{O}})$. By induction on the length of  $S$, we have    $C\rightarrow D_i \in \mathsf{Hom}(\mathscr{C}_c\tuple{\mathcal{O}})$.  By \ref{org-conj-c}, we have $C\rightarrow D_1\sqcap \cdots \sqcap D_n\in \mathsf{Hom}(\mathscr{C}_c\tuple{\mathcal{O}})$, and thus by \ref{org-syntax-concept}(transitivity)  $C\rightarrow   X\in \mathsf{Hom}(\mathscr{C}_c\tuple{\mathcal{O}})$.

\item \ref{el-cano-sqcap} ensures $E,F\in L(x_C)$ with  $E\sqcap F\in L(x_C)$.  By induction, we have $C\rightarrow E\sqcap F\in \mathsf{Hom}(\mathscr{C}_c\tuple{\mathcal{O}})$. By \ref{org-conj-d}, $C\rightarrow E, C\rightarrow F\in \mathsf{Hom}(\mathscr{C}_c\tuple{\mathcal{O}})$.

\item \ref{el-cano-exists-plus}  ensures $\exists R.D\in L(x_C)$, i.e $x_C$ has   an $R$-successor $x_W$ with $D\in L(x_W)$ and $\exists R.D\in \mathsf{sub}(\mathcal{O})$. By induction we have  $W\rightarrow D\in \mathsf{Hom}(\mathscr{C}_c(\mathcal{O}))$ and $\exists R'.W\in L(x_C)$, $C\rightarrow \exists R'.W\in \mathsf{Hom}(\mathscr{C}_c(\mathcal{O}))$ with $R'\rightarrow R\in \mathsf{Hom}(\mathscr{C}_r(\mathcal{O}))$.  By \ref{org-exists}, $\exists R'.W \rightarrow \exists R.D\in \mathsf{Hom}(\mathscr{C}_c(\mathcal{O}))$.    By \ref{org-syntax-concept}(transitivity)  $C\rightarrow   \exists R.D\in \mathsf{Hom}(\mathscr{C}_c\tuple{\mathcal{O}})$.

\item \ref{el-cano-exists-comp}  ensures  $\exists S.D\in L(x_C)$, i.e $x_C$  has  an  $S$-successor $x_W$ with $\exists S.D\in L(x_W)$  and $S\circ S\overlay{\sqsubseteq}{\ast} S\in \mathcal{R}^+\tuple{\mathcal{O}}$. By induction we have $W\rightarrow \exists S.D\in \mathsf{Hom}(\mathscr{C}_c(\mathcal{O}))$. By \ref{el-cano-exists-plus}, $\exists S.W\in L(x_C)$, and by induction again we have $C\rightarrow \exists S.W\in \mathsf{Hom}(\mathscr{C}_c(\mathcal{O}))$. By \ref{org-exists-comp}, we have $C\rightarrow \exists S.D\in \mathsf{Hom}(\mathscr{C}_c(\mathcal{O}))$.   
\end{itemize}
\end{proof}

By Claim~\ref{claim:el-mod}, it holds that if $\mathcal{O}\models \{a\}\sqsubseteq \bot$, i.e. $\mathcal{O}$ has no model, then  $\tuple{V_0,\cdots,V_n,L}$ has a clash. By construction, it holds that  either (i)  $\bot\in L(x^a)$ for some individual $a$,  or (ii) $\bot\in L(x_{C_i})$  for some  $x_{C_i}$ that is created by applying \ref{el-cano-exists} to some $\exists R.C_i\in L(x_{C_{i-1}})$ where $x_{C_{i-1}}$ is the predecessor of $x_{C_i}$ and $i$ is the level $x_{C_i}$ in $\tuple{V_0,\cdots,V_n,L}$. By Claim~\ref{claim:el-tab-cat},  $C_i\rightarrow \bot\in  \mathsf{Hom}(\mathscr{C}_c\tuple{\mathcal{O}})$. By the construction of $\tuple{V_0,\cdots,V_n,L}$, $x_{C_i}$ is a descendant of $x^a$ for some individual $a$. By  \ref{org-syntax-functor},   $\exists R.C_{i} \rightarrow \bot\in  \mathsf{Hom}(\mathscr{C}_c\tuple{\mathcal{O}})$. This propagation of $\bot$ reaches $x^a$ thanks to  \ref{org-syntax-functor} and Claim~\ref{claim:el-tab-cat}. We obtain  
$\{a\}\rightarrow \bot\in  \mathsf{Hom}(\mathscr{C}_c\tuple{\mathcal{O}})$. This completes the proof of the theorem.
\end{proof}


%
%

\section{Related Work}\label{sec:rw}

There have been very few works on category theory that have direct connections to DLs. \citeauthor{spivak2012}  used  category theory to define a high-level graphical language comparable with OWL for knowledge representation, rather than a foundational formalism for reasoning.
\citeauthor{CodescuMK17} introduced a categorical framework to formalize a network of aligned ontologies. The formalism on which  the framework is based is independent from the logic used in ontologies. It is shown that all global reasoning over such a network can be reduced to local reasoning for a given DL used in the ontologies involved in the network. As a consequence, the semantics of DLs employed in the ontologies still relies on set theory. 
A recent work \cite{Bednarczyk22} presented   a categorical semantics for DL games which simulate the relationship between two interpretations of a DL concept. This  allows to translate the game comonad for modal logic into  DL games.

We now discuss some results having indirect connections to DLs. \citeauthor{MOSS1999} and    \citeauthor{KUPKE2011}   considered different fragments of modal logic as \emph{coalgebraic logics}. A coalgebraic logic can be obtained from an $\mathsf{F}$-coalgebra for a functor $\mathsf{F}$  on the category of sets. Such a functor $\mathsf{F}$ is used to describe  structures (or \emph{shapes}) of Kripke-models that all formulas expressed in the coalgebraic logic should admit. This approach results in a tool allowing to establish the relationship between  expressiveness  and  model structures of a logic. 
The approach presented in this paper differs from the  coalgebra-based approaches on the following aspect: in our approach, we do use categorical language to directly  ``encode" the usual set semantics of DLs, but not to describe models of DLs. This use of category theory allows us to decompose the semantics of each logical constructor into several pieces as categorical arrows. Some of them represent the semantic interaction of different constructors which is responsible for intractable reasoning. By dropping such arrows from categorical semantics, we can get new DLs in which the computational complexity for reasoning should be lower.  Given a new DL $\mathscr{L}$ identified by our approach, it is unclear how to determine an $\mathsf{F}$-coalgebra yielding a coalgebraic logic which has the same expressiveness and computational complexity than those of $\mathscr{L}$. 

It would be relevant to discuss \emph{algebraic semantics} from which the categorical semantics introduced in this paper is widely inspired \cite[p.135]{gol06},  \cite[p.50]{saunders92}.
To define the semantics of  classical  propositional logic, one uses a \emph{Boolean algebra}  $\mathbf{BA}=\langle \mathcal{P}(D), \subseteq \rangle$ that is a \emph{lattice} where $\mathcal{P}(D)$ denotes the set of all subsets of a non-empty set $D$, $D$ the \emph{maximum} (or  $\mathbf{1}$) and $\varnothing$ the \emph{minimum} (or  $\mathbf{0}$) under $\subseteq$.  A  \emph{valuation} (or truth function) $V:\Phi_0\rightarrow \mathcal{P}(D)$  assigns a truth-value to each atomic propositional sentence $x\in \Phi_0 $ with $V(x)\in  \mathcal{P}(D)$. 
$V$ can be naturally extended to $x\wedge y, x\vee y$, $\neg x$ using set operations  such as intersection $\cap$, union $\cup$ and complement $\sim$ over $\mathcal{P}(D)$. Then, a sentence $x$ is \emph{valid}   iff  $V(x)=\mathbf{1}$. These set operations are heavily based on the  \emph{set membership} operator $\in$. It holds that $\mathbf{BA}$ is \emph{complemented} and \emph{distributive}  thanks to set membership.  A notable extension of Boolean algebra is \emph{Heyting algebra} $\mathbf{HA}=\langle H, \leq \rangle$ that is also a \emph{lattice} with \emph{exponential} $y^x$ for all $x,y\in H$ where $y^x$ is defined as $z\leq y^x $ iff $ z \wedge x \leq y$ for all $z$.  The truth function $V$ for $\wedge$ and $\vee$ is  defined respectively as a \emph{greatest lower bound} and a \emph{least upper bound} of $x,y$ under $\leq$ while  negation   is defined as   \emph{pseudo-complement}, i.e., $V(\neg x)$ is defined as a maximum set in $H$ such that $V(\neg x)$ is disjoint with $V(x)$. As a consequence, $\neg \neg x = x$ does not hold in a Heyting algebra in general. However, every Heyting algebra  is distributive.  Note that each Boolean algebra is a Heyting algebra but the converse does not hold.
\citeauthor{gol06} replaced the poset $\mathcal{P}(D)$ from $\mathbf{BA}$ with a \emph{topos} which is a particular category, and used arrows in the topos to define a truth function \cite[p.136]{gol06}.

\citeauthor{saunders92} used \emph{left and right adjoints} in a category to represent unqualified existential and universal quantifiers  $(\exists x)S(x,y)$ and $(\forall x)S(x,y)$  in FOL where $S$ is a binary predicate $S\subseteq X\times Y$, $(\exists x)S(x,y)$ denotes all $y$ such that there is some $(x',y)\in S$, and $(\forall x)S(x,y)$ denotes all $y$ such that $x'\in X$ implies $(x',y)\in S$.  For the projection $p: X\times Y\rightarrow Y$, there is  a functor $p^*:\mathcal{P}(Y)\rightarrow \mathcal{P}(X\times Y)$. They used $\exists_p$ and $\forall_p$ to denote respectively  left and right adjoints to $p^*$. That means, for all $T\subseteq Y$  we have   $\exists_p S\subseteq T$ iff $S\subseteq p^*T$ $\bigl($i.e. \emph{$\exists_p S$ is the codomain of $S$}$\bigr)$, and $T \subseteq \forall_p S$ iff $p^*T \subseteq S$ $\bigl($i.e. \emph{$p\rvert_X(S)=p\rvert_X(p^*T)=X$  where $p\rvert_X(S')$ denotes the projection of  $S'$ over $X$}$\bigr)$. It was showed that $\exists_p S=(\exists x)S(x,y)$ and $\forall_p S=(\forall x)S(x,y)$   \cite[p.57]{saunders92}.

The categorical semantics given in Section~\ref{sec:cat-semantics} uses  concept and role categories  instead of a lattice where category arrows play the role of set inclusion $\subseteq$. Due to the absence of set membership  from  categories, we adopt \citeauthor{gol06}'s approach that uses arrows from the category instead of set operations to define the truth function of DL complex concepts involving logical constructors.  The truth function $V(C\sqcap D), V(C\sqcup D)$ and $V(\neg C)$ are  defined in the same way than that of Heyting algebraic semantics while $V(\exists R.C), V(\forall R.C)$ are adapted from  the left adjoint $\exists_p R$ and right adjoint $\forall_p R$ to the functor $p^*$ introduced by \citeauthor{saunders92} as described above.  Hence, a concept $C$ in the concept category is \emph{categorically unsatisfiable} iff an arrow $V(C)\rightarrow \bot$ is contained in the category where $\bot$ is the \emph{initial} object.
 
Regarding tractable DLs, there are two DL families that have been widely studied with several notable results. \citeauthor{baader2005} introduced  $\mathcal{EL}^{++}$ that is an extension of  $\mathcal{EL}$  with different constructors such as the domain concrete constructor, composition of roles occurring in role inclusions. The authors proposed an algorithm for checking a concept subsumption $A\sqsubseteq B$ with respect to an ontology $\mathcal{O}$ in  $\mathcal{EL}^{++}$. This algorithm uses a set of \emph{completion rules} to compute $S(X)$ that is a set of concepts containing all concepts more general than a concept $X$  with respect to $\mathcal{O}$. From this, it checks whether $S(A)$ contains $B$ or $\bot$, and decides whether $A\sqsubseteq B$. \citeauthor{Kazakov2014}  introduced $\mathcal{EL}_\bot^+$ including   $\mathcal{EL}$ and role inclusions with composition of roles. The authors presented a reasoning procedure, namely \emph{consequence-based} approach, that relies on rules  directly devised from the semantics of logical constructors allowed in  $\mathcal{EL}_\bot^+$. This approach aims at discovering implicit subsumption from the axioms of an ontology. Thus, it does not need to explicitly create individuals to represent a model. 

\citeauthor{calvanese2007-DLLite} introduced  the DL-Lite family which allows for unqualified existential restrictions, inverse roles,  negation of \emph{basic} concepts and roles, concept and role inclusions. The syntax of this DL family is defined with specific restrictions such that it is sufficiently expressive to be able to formalize knowledge related to conceptual data models while tractablilty of reasoning is preserved. Such a DL allows for a separation between TBox and ABox of an ontology during query evaluation, and thus the process requiring the ABox can be delegated to an SQL Engine highly optimized. The authors of the work proposed a reasoning technique that extends the ABox with the TBox axioms using \emph{chase rules}, and showed that the obtained extended ABox represents a model of the ontology. This reasoning technique aims to extend ABox while the categorical properties we use in this paper leads to adding arrows representing TBox axioms.     

%
%

\section{Conclusion and Future Work}\label{sec:conc}

We have presented the categorical semantics of the DL $\mathcal{SH}$ with a truth function that associates an object of a concept category to each concept occurring in an $\mathcal{SH}$ ontology. Instead of using set membership with set operations to define the extension of  an  interpretation function $\mathcal{I}$ for complex concepts containing DL constructors, the categorical semantics uses categorical properties with \emph{arrows} of the category to define an extension of the truth function  for complex concepts.  An $\mathcal{SH}$ ontology is categorically inconsistent if there is an arrow $\{a\}\rightarrow \bot$ in the minimal concept  category where $\{a\}$ is an individual object and $\bot$ the initial object  of the category. We showed that inconsistency in $\mathcal{SH}$ under the set-theoretical semantics is equivalent to  inconsistency under the categorical  semantics.  

  We have also identified  two  independent categorical properties that represent the interactions between $\exists$ and $\forall$, and between $\sqcap$ and $\sqcup$. By dropping these categorical properties from those for $\mathcal{SH}$, we obtained a new  tractable DL, namely   $\mathcal{EL}^{\rightarrow}$,
  and showed that $\mathcal{EL}^{\rightarrow}$ is more expressive than $\mathcal{EL}_\bot^\circ$. This implies that   $\mathcal{EL}^{\rightarrow}$ can be defined in a decreasing way from $\mathcal{SH}$, and an increasing way from $\mathcal{EL}_\bot^\circ$.
    The main difference between the categorical and  set-theoretical  semantics is that the former makes interactions between constructors explicit while the latter  ``hides" not only  interactions between constructors but also  constructors themselves. For instance, $\mathcal{EL}^{(-)}$ eclipses the constructors $\forall$ and $\sqcup$. This difference results from renouncing set membership in the categorical semantics. Without rigidity of set membership, we need  more semantic constraints in the categorical definitions to compensate this absence. The more semantic constraints are explicitly expressed the more possibilities are available to define new logics with different reasoning complexities.  In terms of reasoning under the categorical semantics, the algorithm devised from the categorical rules   for $\mathcal{SH}$ consists of saturating the concept category in a deterministic way to determine whether some particular arrow is contained in the saturated concept category. Analogously, the tractable algorithm devised from the categorical rules   for $\mathcal{EL}^\rightarrow$  has the same behavior.
    


The reason we use category theory in this work is not because we need a specific result from it.  The added value of category theory lies in its ability to isolate from the DL semantics the properties that are responsible for intractability.

  Regarding informal meaning of the categorical semantics in respect to  the set-theoretical one in the point of view by a knowledge engineer, it is debatable  because, for instance,  ``the conjunction of two sets contains exactly all elements of both sets" is not necessarily more intuitive than   ``the conjunction of two sets is the greatest set that is smaller than each one".  
  Moreover, it is not always possible to interpret a concept as a set of individuals. For instance, if one wishes to represent medical terms such as $\mathsf{Disorder, Syndrome, Disease}$ as concepts in an ontology, it is not clear how to determine their individuals \cite{Nordelfelt95}. In such a situation, an axiom $\mathsf{Disease}\sqsubseteq  \mathsf{Disorder}$ would not  make sense 
  while 
  an arrow $\mathsf{Disease}\rightarrow \mathsf{Disorder}$   under the categorical semantics would make sense because these objects refer to specific health states but  $\mathsf{Disease}$ is more concrete and measurable.
  
We plan to optimize and implement our algorithm using well-known techniques from existing reasoners such as ELK \cite{Kazakov2014}, CEL \cite{BaaLutSun-06}. Especially, all rules implemented in ELK except for those related to role compositions have  counterparts in  categorical properties. 
Moreover, the categorical properties related to $\bar{\forall}$ and $\bar{\sqcup}$ would not seem problematic since they require to check   premises related to a limited number of forms of objects without considering exhaustively all pairs of objects. Once the algorithm for $\mathcal{EL}^\rightarrow$ is implemented in a reasoner, namely $\mathsf{Ela}$, the next step consists in evaluating the performance and the usefulness of $\mathcal{EL}^\rightarrow$. 
For this, we take biomedical ontologies in $\mathcal{EL}^\circ_\bot$  containing  negative knowledge in labels such as \emph{lacks\_part}, \emph{absence\_of}, and translate them into concepts using  $\neg, \forall$  as described in Example~\ref{ex:intro}. The resulting ontologies may be expressed in $\mathcal{SH}$ without $\sqcup$.
We also need a  reasoner based on the set-theoretical semantics  such as $\mathsf{Konclude}$ \cite{Steigmiller14} to  benchmark against $\mathsf{Ela}$. 
There are the following two possibilities (recall that concept set-theoretical  satisfiability in $\mathcal{SH}$ implies concept categorical  satisfiability in  $\mathcal{EL}^\rightarrow$ according to Corollary~\ref{cor:SH-ELA}): (i)  both $\mathsf{Ela}$ and $\mathsf{Konclude}$ return the same answer to the query of whether  a concept $C$ is  satisfiable with respect to an ontology $\mathcal{O}$ as described above. In this case, a better performance of $\mathsf{Ela}$ should be expected since $\neg, \forall$  are interpreted as $\bar{\neg}, \bar{\forall}$  in $\mathcal{EL}^\rightarrow$ that   is tractable; (ii)   $\mathsf{Ela}$ returns YES and $\mathsf{Konclude}$ returns NO to the same query.  This happens when the ontology contains (or entails) $n$-disjointness with $n\geqslant 3$, i.e., $C_1\sqcap \cdots \sqcap C_n\sqsubseteq \bot$.  As mentioned in Section~\ref{sec:intro}, it is quite rare since most biomedical ontologies do not contain $n$-disjointness with $n\geqslant 3$. We also plan to perform DL queries (e.g instance checking) with  $\mathsf{Ela}$ over populated large medical ontologies. It is expected that   $\mathsf{Ela}$ will return more instances than    a set-theoretical semantics-based reasoner for the same query due to  the weakened constructors. An evaluation  by a health domain expert on the results obtained by the two reasoners would  clarify the relevance of the additional instances found by  $\mathsf{Ela}$.


From a theoretical point of view, it is interesting to investigate how to represent other DL constructors in categorical language to be able to design new DLs with different reasoning complexities. For an inverse role $R^-$, we need to add at least categorical  properties such as $p_l(R)\leftrightarrows p_r(R^-), p_l(R^-)\leftrightarrows p_r(R)$, and those that allow to deal with the interaction between  existential restriction with inverse role and universal restriction. For number restrictions such as $\geqslant\!\!n R.C$, $\leqslant\!\!n R.C$, they could be simulated using fresh concept/individual names, existential and universal restrictions and disjunction. 
 Another interesting topic would be to study how to answer conjunctive queries with respect to ontologies expressed in  $\mathcal{EL}^\rightarrow$. Indeed, the minimal  concept category of an ontology represents a set of interesting models of the ontology. This construction  is very close to a \emph{universal model} that is an important  tool to investigate conjunctive query answering.

\bibliographystyle{apalike}
\bibliography{ELA2025}

\bigskip

\appendix{\textbf{Appendix}}

\begin{proof}[Proof of Lemma \ref{lem:label-concept}]
We proceed by induction on the structure of $C$. Let $x\in \mathcal{V}_i$. We show that if $C\in \mathcal{L}(x)$ then
$x\in C^\Tb{I}$.  
      
\noindent\textbf{Case $C = \{a\}$}.  By construction, $\{a\}\in \mathcal{L}(x^a)$ and $a^\mathcal{I}=x^a$, $\{a\}^\mathcal{I}=\{x^a\}$. Hence, $a^\mathcal{I}\in \{a\}^\mathcal{I}$.
    
\noindent\textbf{Case $C = A$}. If $A\in \mathcal{L}(x)$ then $x\in A^\Tb{I}$ by 
    definition of $\Tb{I}$. 
    
\noindent\textbf{Case $C = \neg A$}. Assume that $\neg A\in \mathcal{L}(x)$. 
    By contradiction, assume that $x\in A^\Tb{I}$. By the
    definition of $\Tb{I}$ we have $A\in \mathcal{L}(x)$, which contradicts    clash-freeness of $T$.
    Thus $x\notin A^\Tb{I}$, \emph{i.e.} $x\in (\neg A)^\Tb{I}$.
    
\noindent\textbf{Case $C = D_1\sqcap D_2$}. Assume $D_1\sqcap D_2\in 
    \mathcal{L}(x)$, then the $\mathbf{T_{\sqcap}}$   implies that $D_i\in \mathcal{L}(x)$
    for $i = 1,2$. Therefore by induction, we have $x\in D_i^\Tb{I}$ for 
    $i = 1,2$, hence $x\in D_1^\Tb{I}\cap D_2^\Tb{I}$ and by set semantics
    of $\sqcap$, we have $x\in (D_1\sqcap D_2)^\Tb{I}$.
    
\noindent\textbf{Case $C = D_1\sqcup D_2$}. Assume that $D_1\sqcup D_2 \in \mathcal{L}(x)$. By the construction of $T$, $x\in \mathbb{S}(x')$ for some $x'$. Thus,   $\mathcal{L}(x)\cap\{D_1,
    D_2\}\neq \varnothing$ according to $\mathbf{T}_\sqcup$. That means that either $D_1\in \mathcal{L}(x)$ or $D_2\in \mathcal{L}(x)$. By induction, we have $x\in D_1^\Tb{I}$ or $x\in D_2^\Tb{I}$, and thus $x\in D_1^\Tb{I}\cup D_2^\Tb{I}$.  By the set-theoretical  semantics of $\sqcup$, we have  
    $x\in (D_1\sqcup D_2)^\Tb{I}$.
   
\noindent\textbf{Case $C = \exists R.D$}. Assume $\exists R.D\in \mathcal{L}(x)$.  By the construction of $\mathbf{T}$ from  $T$,   there are 2 possibilities; (i) $x$ has an   $R$-successor $y$ with
   $D\in \mathcal{L}(y)$ in  $\mathbf{T}$ such that  $y$ is non-blocked. By induction hypothesis,  we have $y\in D^\Tb{I}$, and  by the definition of $R^\mathcal{I}$ we have  $\tuple{x,y}\in R^\Tb{I}$. Hence,  $x\in(\exists R.D)^\Tb{I}$;  (ii) $y$ is blocked by an ancestor $z$  such that $\mathcal{L}(y)\subseteq \mathcal{L}(z)$. This implies that $D\in \mathcal{L}(z)$. By induction hypothesis, we have  $z\in D^\Tb{I}$, and by the definition of $R^\mathcal{I}$ we have 
     $\tuple{x,z}\in R^\Tb{I}$. Hence $x\in(\exists R.D)^\Tb{I}$.
  
\noindent   \textbf{Case $C= \forall R.D$}. Assume that $\forall R.D\in \mathcal{L}(x)$ and there is some $y$ such that $(x,y)\in R^\mathcal{I}$. We have to show that  $y\in D^\mathcal{I}$. According to the definition of $R^\mathcal{I}$, $(x,y)$ is a $R$-path.  We proceed by induction on the length 
   of  $(x,y)$, denoted $\ell(x,y)$,  that is inductively defined as follows: (i) if $(x,y)$ is an $R$-edge, $y$ is non-blocked; or
   $(x,y')$ is an $R$-edge, $y$ blocks $y'$, then $\ell(x,y) = 1$; (ii) if $(x,z)$ is an $R$-path  with  $\ell(x,z) = n -1, n>1$, and $(z,y)$ is an $R$-edge such that either
    $y$ is non-blocked, or $(z,y')$ is an $R$-edge, $y$ blocks $y'$, then  $\ell(x,y) = n$.
    From the definition of $\ell(x,y)$ and   $R$-path, it  holds that if $\ell(x,y)>1$  then   $\{S\circ S\overlay{\sqsubseteq}{\ast} S, S\overlay{\sqsubseteq}{\ast} R\}\subseteq \mathcal{R}^+\tuple{\mathcal{O}}$ where it is possible that $S=R$. 
 
   \begin{enumerate}[leftmargin=0.1cm, itemindent=0.5cm] 
   \item $\ell(x,y) = 1$. Assume that $\{S\circ S\overlay{\sqsubseteq}{\ast} S, S\overlay{\sqsubseteq}{\ast} R\}\nsubseteq \mathcal{R}^+\tuple{\mathcal{O}}$. By definition of paths, there are two possibilities : (i)  $(x,y)$ is an $R$-edge and  $y$ is non-blocked.  By $\mathbf{T}_\forall$,   $D\in \mathcal{L}(y)$. By induction $y\in D^\mathcal{I}$.
   (ii)   $(x,y')$ is an $R$-edge and   $y$ blocks $y'$ with $\mathcal{L}(y')\subseteq \mathcal{L}(y)$.   By $\mathbf{T}_\forall$,   $D\in \mathcal{L}(y')$, and thus $D\in \mathcal{L}(y)$. By induction $y\in D^\mathcal{I}$.
   
   Assume that $\{S\circ S\overlay{\sqsubseteq}{\ast} S, S\overlay{\sqsubseteq}{\ast} R\}\subseteq \mathcal{R}^+\tuple{\mathcal{O}}$. By definition of paths, there are two possibilities:  (i)  $(x,y)$ is an $S$-edge and  $y$ is non-blocked. Hence, $(x,y)$ is an $R$-edge. By $\mathbf{T}_\forall$,   $D\in \mathcal{L}(y)$. By induction $y\in D^\mathcal{I}$.  By $\forall^+$-rule, $\forall S.D\in \mathcal{L}(y)$; 
   (ii)   $(x,y')$ is an $S$-edge, and thus an $R$-edge, and   $y$ blocks $y'$ with $\mathcal{L}(y')\subseteq \mathcal{L}(y)$.     By $\mathbf{T}_\forall$,   $D\in \mathcal{L}(y')$, and thus $D\in \mathcal{L}(y)$. By induction $y\in D^\mathcal{I}$. By $\mathbf{T}^+_\forall$, $\forall S.D\in \mathcal{L}(y')$, and thus $\forall S.D\in \mathcal{L}(y)$. 
       
   \item Assume that  $\ell(x,y)=n>1$, i.e there is some $z$ such that $\ell(x,z)=n-1$, $\ell(z,y)=1$.  According to the definition of paths, it holds that $\{S\circ S\overlay{\sqsubseteq}{\ast} S, S\overlay{\sqsubseteq}{\ast} R\}\subseteq \mathcal{R}^+\tuple{\mathcal{O}}$ such that  $(x,z)$ is an  $S$-path, and (i) either   $(z,y)$ is an $S$-edge, $y$ is non-blocked, (ii) or $(z,y')$ is an $S$-edge, $y$ blocks $y'$. By induction on the length of the path, we have $\forall S.D\in \mathcal{L}(z)$. By the $\forall$-rule and $\mathcal{L}(y')\subseteq \mathcal{L}(y)$,   $D\in \mathcal{L}(y)$. By induction $y\in D^\mathcal{I}$.
   By the $\forall^+$-rule and $\mathcal{L}(y')\subseteq \mathcal{L}(y')$, we have  $\forall S.D\in \mathcal{L}(y)$. 
   \end{enumerate}       
\end{proof}

\begin{proof}[Proof of Lemma~\ref{lem:tab-mod}]
 Assume that there is a clash-free tableau $\mathbf{T}=\langle \mathcal{V}_1,\cdots,\mathcal{V}_n, \mathcal{E},   \mathcal{L}\rangle$  for $\mathcal{O}$, and $\mathcal{I}$ a tableau interpretation from $\mathbf{T}$. We show that $\mathcal{I}$ is a model of   $\mathcal{O}$. For this, we have to show that $\mathcal{I}$ satisfies all axioms in $\mathcal{O}$.
    \begin{itemize}[leftmargin=0.1cm, itemindent=0.5cm] 
    
    \item $C^{\mathcal{I}}\subseteq D^\mathcal{I}$ for all $C\sqsubseteq D\in \mathcal{O}$. We have 
    already proven in Lemma~\ref{lem:label-concept} that if $C\in \mathcal{L}(x)$ for $x\in \Delta^{\Tb{I}}$, then
    $x\in C^{\Tb{I}}$. Therefore, we need to prove that if $C\sqsubseteq D\in\tO$ and $x\in C^{\Tb{I}}$
    then $x\in D^{\Tb{I}}$. 
    If $C = \{a\}$ for  an individual $a$, then $\{a\}\sqsubseteq D\in\tO$ implies $\{a\},D\in \mathcal{L}(x^a)$, 
    consequently $a^\Tb{I}\in D^\Tb{I}$ by Lemma~\ref{lem:label-concept}.
    Assume that  $C\sqsubseteq D\in\tO$ and $x\in C^{\Tb{I}}$ and $C\neq\{a\}$ for all  individual $a$, 
    then by  $\mathbf{T}_\sqsubseteq$, we have $\neg C\sqcup D\in \mathcal{L}(x)$, which leads to 
    $\{\neg C, D\}\cap \mathcal{L}(x)\neq\varnothing$ by  $\mathbf{T}_\sqcup$. However, if 
    $\neg C\in \mathcal{L}(x)$ then by Lemma~\ref{lem:label-concept} we must have 
    $x\in(\neg C)^{\Tb{I}}$, which contradicts   $x\in C^{\Tb{I}}$. Therefore, we have $D\in \mathcal{L}(x)$ and then, 
    by Lemma~\ref{lem:label-concept} again, we have $x\in D^{\Tb{I}}$.
    
    \item $R^{\mathcal{I}}\subseteq S^{\mathcal{I}}$ for each $R \sqsubseteq S \in \mathcal{O}$. Let $(x,y)\in R^{\mathcal{I}}$.   According to the  definition of $R^\mathcal{I}$, $(x,y)$ is an $R$-path.    By the definition of $S^\mathcal{I}$, $(x,y)\in S^{\mathcal{I}}$.

    \item $S^{\mathcal{I}}\circ S^{\mathcal{I}}\subseteq S^\mathcal{I}$ for  $S\circ S \sqsubseteq S \in \mathcal{O}$.  Assume that $(x,y)\in S^\mathcal{I}$, $(y,z)\in S^\mathcal{I}$. By the definition of $S^\mathcal{I}$, $(x,y)$ and $(y,z)$ are $S$-paths. By the definition of paths, they are composed of a sequence of $S$-edges. This implies that   $(x,z)$ is an $S$-path, thus $(x,z)\in S^\Tb{I}$. 
    \end{itemize}
 
\end{proof}

\begin{proof}[Proof of Lemma~\ref{lem:mod-tab}]
Assume that $\mathcal{O}$ is consistent. This implies that $\mathcal{O}$ has a model $\mathcal{I}=\langle \Delta^{\mathcal{I}}, \cdot^{\mathcal{I}}\rangle$.  Let $T=\langle V_0,\cdots,V_n, E,L\rangle$ be a universal tableau for  $\mathcal{O}$  with $\mathbb{S}(x)$ and $\mathbb{T}(x)$ for $x\in V_i$ and  $i\in \{0,n\}$ according to Definition~\ref{def:tree}.  We show that there is a clash-free tableau $\mathbf{T}=\langle \mathcal{V}_0,\cdots,\mathcal{V}_n,  \mathcal{E},\mathcal{L} \rangle$    that can be built from  $\mathcal{I}$ and $T$. For this, we use $\mathcal{I}$ and $\mathbb{T}(x)=\tuple{V',E',L'}$ to choose the  right node from $\mathbb{S}(x)$ for $x\in V_i$ with $i\in \{0,n\}$ in order to add it to $\mathcal{V}_i$.

We need to define inductively a sequence of steps that allows to generate a clash-free $\mathbf{T}$  satisfying all tableau properties given in Definition~\ref{def:uni-tableau} and \ref{def:sh-tableau}. For this,  we use a function $\pi$ from $\bigcup_{i\in \{1,n\}}\mathcal{V}_i$ to $\Delta^\mathcal{I}$  such that the following properties are maintained for $x\in \mathcal{V}_i$ and  $i\in \{0,n\}$ when updating $\mathbf{T}$.
\begin{align}
    C\in \mathcal{L}(u) &\text{ implies } \pi(u)\in C^\mathcal{I} \label{tableau-prop1}\\
    \text{If } u' \text{ is an } R\text{-successor of } u &\text{ then } (\pi(u), \pi(u'))\in R^\mathcal{I} \label{tableau-prop2}
   \end{align}
The function $\pi$ will be initialized, extended and updated. 

\begin{itemize}[leftmargin=0.1cm, itemindent=0.5cm] 
    \item If $a$ is an individual, then $a^\mathcal{I}\in 
        \Delta^\mathcal{I}$.  By construction, there is some $x^a\in V_0$ with $\mathbb{S}(x^a)$ and $\mathbb{T}(x^a)=\tuple{V',E',L'}$ rooted at some $u^a$. 
        We add $u^a$ to $\mathcal{V}_0$ and set $\mathcal{L}(u^a)=L'(u^a)$, and define $\pi(u^a)=a^\mathcal{I}$.
        If $\{a\}\sqsubseteq C\in \mathcal{O}$, then $a^\mathcal{I}\in 
        C^\mathcal{I}$. By  Item~\ref{claim:1:0}, Definition~\ref{def:tree},     we have $C\in L'(u^a)$ if   $\{a\}\sqsubseteq C$, and thus  $C\in \mathcal{L}(u^a)$.
        If $\{(a,b)\}\sqsubseteq R\in \mathcal{O}$ then $(a^\mathcal{I}, b^\mathcal{I})\in R^\mathcal{I}$, and thus   $\pi(u^a)=a^\mathcal{I}$ and $\pi(u^b)=b^\mathcal{I}$.  We define $u^b$ to be an $S$-root successor of $u^a$ for $R\overlay{\sqsubseteq}{\ast} S\in \mathcal{R}^+\tuple{\mathcal{O}}$. Therefore, Properties~(\ref{tableau-prop1})  and  (\ref{tableau-prop2}) are preserved. 
        
        \item  Assume   $C\sqsubseteq D\in\tO$. Let   $u\in \mathcal{V}_i$ where $u$ is a node of $\mathbb{T}(x)=\tuple{V',E',L'}$ rooted at $v_0$ with some $x\in V_i$ by the construction of $\mathcal{V}_i$. 
        We have $L'(v_0)\subseteq \mathcal{L}(u)$ by the construction of $\mathbb{T}(x)$. It holds that   $\neg C\sqcup D\in  L'(v_0)$  by Item~\ref{claim:1:0}, Definition~\ref{def:tree}.
        Hence,  $\neg C\sqcup D\in  \mathcal{L}(u)$.
        Since $\Tb{I}$ is a model, 
        we have $\Tb{I}\models C\sqsubseteq D$,  and thus  $\Tb{I}\models \top\sqsubseteq
        \neg C\sqcup D$. Therefore $\pi(u)\in(\neg C\sqcup D)^\Tb{I}$ and 
        Property~(\ref{tableau-prop1}) is preserved. 

        \item Assume that $C_1\sqcap C_2\in \mathcal{L}(u)$ for $u\in \mathcal{V}_i$. By construction,  $u$ is a node  of $\mathbb{T}(x)=\tuple{V',E',L'}$ with $x\in V_i$ and          $L'(u)=\mathcal{L}(u)$. It holds that    $C, D\in  L'(u)$   by Item~\ref{claim:1:1}, Definition~\ref{def:tree}. Due to induction hypothesis of Property~(\ref{tableau-prop1}) we  have  $\pi(u)\in(C_1\sqcap C_2)^\Tb{I}$, and thus
        $\pi(u)\in C_i^\Tb{I}$ for all  $i\in\{1,2\}$. Hence, Property~(\ref{tableau-prop1}) is preserved.

        \item  Assume that   $\exists R.C\in \mathcal{L}(u)$ with $u\in \mathcal{V}_i$ such that $u$ has no $R$-successor $u'\in \mathcal{V}_{i+1}$  with $C\in \mathcal{V}(u')$, and   $u$ is not blocked.  By construction,  $u$ is a node  of $\mathbb{T}(x)=\tuple{V',E',L'}$ with $x\in V_i$ and          $L'(u)=\mathcal{L}(u)$.
        By \ref{def:uni-tableau:exists}, $x$ has an $R$-successor $y\in V_{i+1}$  with   $C \in L(y)$. Due to induction hypothesis of  Property~(\ref{tableau-prop1}), we have  $\pi(u)\in \exists R.C^\mathcal{I}$. Since $\mathcal{I}$ is a model,   there is some $d\in \Delta^\mathcal{I}$ such that $d\in C^\mathcal{I}$ and $\langle \pi(u), d\rangle \in R^\mathcal{I}$. By definition, there is a tree  $\mathbb{T}(y)=\tuple{V'_y,E'_y,L'_y}$ rooted at $u_0$. 
        We add  $u_0$ to $\mathcal{V}_{i+1}$ as an $R$-successor of $u$, and  set $\mathcal{L}(u_0)=L'_y(u_0)$. Since $C\in L'_y(u_0)$ by the construction of $\mathbb{T}(y)$, we have $C\in \mathcal{L}(u_0)$.  
        We extend $\pi$ by defining  $\pi(v_0)=d\in C^\mathcal{I}$.  This implies that Property~(\ref{tableau-prop1}) is preserved. Assume that $R\overlay{\sqsubseteq}{\ast} S\in \mathcal{R}^+\tuple{\mathcal{O}}$.   Since $\mathcal{I}$ is a model, it holds that $(\pi(u),\pi(v_0))\in S^\mathcal{I}$. Hence, Property~(\ref{tableau-prop2}) is also preserved.

        \item Assume $\forall R.D\in \mathcal{L}(u)$ with $u\in \mathcal{V}_i$ and $v\in \mathcal{V}_{i+1}$ is an $R$-successor or root successor of $u$. By construction,  $u$ is a node  of $\mathbb{T}(x)=\tuple{V',E',L'}$ with $x\in V_i$,           $L'(u)=\mathcal{L}(u)$; and $v$ is a node  of $\mathbb{T}(y)=\tuple{V'_y,E'_y,L'_y}$ with            $L'_y(v)=\mathcal{L}(v)$.
        By \ref{def:uni-tableau:forall}, it holds that $D\in L'_y(v)=\mathcal{L}(v)$.  Due to induction hypothesis of   Properties~(\ref{tableau-prop1}) and (\ref{tableau-prop2}), we have  $\pi(u)\in \forall R.D^\mathcal{I}$ and  $\tuple{\pi(u), \pi(v)}\in R^\mathcal{I}$. This implies $\pi(v)\in D^\mathcal{I}$. Hence, Property~(\ref{tableau-prop1}) is preserved.  
         
        \item    Assume $\forall R.D\in \mathcal{L}(u)$ with  $S\circ S\overlay{\sqsubseteq}{\ast} S, S \overlay{\sqsubseteq}{\ast} R \in \mathcal{R}^+\tuple{\mathcal{O}}$, 
        and $v\in \mathcal{V}_{i+1}$ is an $S$-successor of $u$.  By construction,  $u$ is a node  of $\mathbb{T}(x)=\tuple{V',E',L'}$ with $x\in V_i$,           $L'(u)=\mathcal{L}(u)$; and $v$ is a node  of $\mathbb{T}(y)=\tuple{V'_y,E'_y,L'_y}$ with            $L'_y(v)=\mathcal{L}(v)$.
        By \ref{def:uni-tableau:forall-trans}, it holds that $\forall S.D\in L'_y(v)=\mathcal{L}(v)$. 
        Due to induction hypothesis of Property~(\ref{tableau-prop1}), we have $\pi(u)\in (\forall R.D)^\Tb{I}$. We have to  show that  $\pi(v)\in(\forall S.D)^\Tb{I}$.
        Let us assume that there is $z\in\Delta^\Tb{I}$ such that $(\pi(v),z)\in S^\Tb{I}$. Since 
        $v$ is an $S$-successor of $u$, this means that $(\pi(u),\pi(v))\in S^\Tb{I}$ by  Property~(\ref{tableau-prop1}).
        Moreover, $\Tb{I}$
        is model and  $S\circ S\subtransrole S\in\mathcal{R}^+\tuple{\mathcal{O}}$, it holds that  $(\pi(u), z)
        \in S^\Tb{I}$, and  thus   $z\in D^\Tb{I}$. Therefore,  we obtain $\pi(v)\in(\forall S.D)^\Tb{I}$, and thus Property~(\ref{tableau-prop1}) holds.

        \item Assume that $C_1\sqcup C_2\in \mathcal{L}(u)$ with $u\in \mathcal{V}_i$.  By construction,  $u$ is a node  of $\mathbb{T}(x)=\tuple{V',E',L'}$ with $x\in V_i$,           $L'(u)=\mathcal{L}(u)$; 
        Due to induction hypothesis of Property~(\ref{tableau-prop1}), we have $\pi(u)\in(C_1\sqcup C_2)^\Tb{I}$. This implies that $\pi(u)\in C_i^\Tb{I}$ for some $i\in \{1,2\}$.  
        Assume that $u$ has 2 successors $u_1,u_2$ in $\mathbb{T}(x)$ and $C_1\in L'(u_1)$ and $C_2\in L'(u_2)$. In this case, we update $u:=u_i$. Hence, $\pi(u)\in C_i^\mathcal{I}$, and thus     Property~(\ref{tableau-prop1}) is preserved.   
    \end{itemize}
Due to the blocking condition, $n$ is finite.    Moreover, we never remove anything from $\mathbf{T}$ when constructing $\mathbf{T}$. Hence, our construction terminates  such that  Properties~\ref{tableau-prop1} and \ref{tableau-prop2} are preserved. We now show that $\mathbf{T}$ has no clash.     By contradiction, assume that there is a clash in $\mathbf{T}$, this 
    means that there is some $x\in \mathcal{V}_i$ such that $A, \neg A\in \mathcal{L}(x)$.  According to Property~\ref{tableau-prop1}, we have $\pi(x)\in A^\mathcal{I}\cap \neg A^\mathcal{I}$, which is a contradiction since $\mathcal{I}$ is a model.
\end{proof} 
\end{document}